\documentclass[reqno,oneside,11pt]{amsart}

\usepackage{url,amsmath,amscd,amsfonts,amsthm,graphicx,marginnote,amssymb,bm,mathtools,MnSymbol}

\usepackage[utf8]{inputenc}
\usepackage{braket,todonotes}
\usepackage{color}
\usepackage{thm-restate}
\usepackage{lipsum}

\theoremstyle{plain}
\newtheorem{thm}{Theorem}[section]
\newtheorem{cor}[thm]{Corollary}
\newtheorem{lem}[thm]{Lemma}
\newtheorem{prop}[thm]{Proposition}
\newtheorem{conj}[thm]{Conjecture}

\newtheorem{prob}[thm]{Problem}
\theoremstyle{definition}
\newtheorem{defin}[thm]{Definition}

\newtheorem{rem}[thm]{Remark}

\newtheorem{ass}[thm]{Assumption}

\newcommand\blsfootnote[1]{%
  \begingroup
  \renewcommand\thefootnote{}\footnote{#1}%
  \addtocounter{footnote}{-1}%
  \endgroup
}

\numberwithin{equation}{section}

\let\frontmatter\relax
\def\mainmatter{\def\baselinestretch{1}\normalfont \setlength{\parskip}{0.5em}}

\usepackage[margin=1.3in]{geometry}

\usepackage{etoolbox}

\makeatletter
\def\l@subsection{\@tocline{2}{0pt}{2.5pc}{5pc}{}}

\renewcommand{\tocsection}[3]{%
  \indentlabel{\@ifnotempty{#2}{\bfseries\ignorespaces#1 #2.~}}\bfseries#3}
\renewcommand{\tocsubsection}[3]{%
  \indentlabel{\@ifnotempty{#2}{\ignorespaces#1 #2.~}}#3}

\makeatother

\makeatletter
\renewcommand{\section}{\@startsection
{section}%                   % the name
{1}%                         % the level
{\z@}%                       % the indent / 0mm
{-\baselineskip}%            % the before skip / -3.5ex \@plus -1ex \@minus -.2ex
{0.8\baselineskip}%          % the after skip / 2.3ex \@plus .2ex
{\centering\scshape\large}} % the style

\renewcommand{\subsection}{\@startsection
{subsection}%                   % the name
{2}%                         % the level
{\z@}%                       % the indent / 0mm
{-0.8\baselineskip}%            % the before skip / -3.5ex \@plus -1ex \@minus -.2ex
{0.5\baselineskip}%          % the after skip / 2.3ex \@plus .2ex
{\normalfont \bf \normalsize}} % the style

\renewcommand{\subsubsection}{\@startsection
{subsubsection}%                   % the name
{3}%                         % the level
{\z@}%                       % the indent / 0mm
{-0.8\baselineskip}%            % the before skip / -3.5ex \@plus -1ex \@minus -.2ex
{0.5\baselineskip}%          % the after skip / 2.3ex \@plus .2ex
{\normalfont \it \normalsize}} % the style

\makeatother

\let\emptyset\varnothing
\DeclareMathOperator*{\Res}{Res}

\usepackage{mathtools}

\newcommand{\zhu}{\operatorname{Zhu}}
\newcommand{\ad}{\operatorname{ad}}

\newcommand{\g}{{\mathfrak{g}}}
\newcommand{\cA}{{\mathcal{A}}}

\newcommand{\bea}{\begin{eqnarray}}
\newcommand{\eea}{\end{eqnarray}}
\newcommand{\beq}{\begin{equation}}
\newcommand{\eeq}{\end{equation}}
\newcommand{\dd}{\mathrm{d}}

\newcommand{\op}[1]{\operatorname{#1}}

\usepackage[plainpages=false]{hyperref}

\definecolor{darkgreen}{rgb}{0.1, 0.8, 0.1}
\definecolor{amber}{rgb}{1.0,0.49,0}

\UseRawInputEncoding

\begin{document}
\frontmatter

\title{Whittaker vectors for $\mathcal{W}$-algebras from topological recursion}

\author{Ga\"etan Borot}
\address{Humboldt Universit\"at zu Berlin, Institut f\"ur Ma\-the\-ma\-tik und Ins\-ti\-tut f\"ur Phy\-sik \\ Ru\-do\-wer Chau\-s\-see 25, 12489 Berlin, Germany.}
\address{Max-Planck-Institut f\"ur Mathematik, Vivatsgasse 7, 53111 Bonn, Germany.}
\email{gaetan.borot@hu-berlin.de}

\author{Vincent Bouchard}
\address{Department of Mathematical \& Statistical Sciences,
University of Alberta, 632 CAB\\
Edmonton, Alberta, Canada T6G 2G1}
\email{vincent.bouchard@ualberta.ca}

\author{Nitin K. Chidambaram}
\address{Max Planck Institut f\"ur Mathematik \\ 
Vivatsgasse 7, 53111 Bonn, Germany}
\email{kcnitin@mpim-bonn.mpg.de}

\author{Thomas Creutzig}
\address{Department of Mathematical \& Statistical Sciences,
University of Alberta, 632 CAB \\
Edmonton, Alberta, Canada T6G 2G1}
\email{creutzig@ualberta.ca}

%\subjclass[2010]{81R10, 14N10, 51P05}

%\keywords{$\mathcal{W}$ algebras, topological recursion, Airy structures, enumerative geometry} 

\begin{abstract} 
	We identify Whittaker vectors for $\mathcal{W}^{\mathsf{k}}(\mathfrak{g})$-modules with partition functions of higher Airy structures. This implies that Gaiotto vectors, describing the fundamental class in the equivariant cohomology of a suitable compactification of the moduli space of $G$-bundles over $\mathbb{P}^2$ for $G$ a complex simple Lie group, can be computed by a non-commutative version of the Chekhov--Eynard--Orantin topological recursion. We formulate the connection to higher Airy structures for Gaiotto vectors of type A, B, C, and D, and explicitly construct the topological recursion  for type A (at arbitrary level) and type B (at self-dual level). On the physics side, it means that the Nekrasov partition function for pure $\mathcal{N} = 2$ four-dimensional supersymmetric gauge theories can be accessed by topological recursion methods.\end{abstract}

\vspace{0.5cm}

\blsfootnote{\begin{small} \noindent \hspace{-0.67cm} \textsc{MSC Classification:} 14H81, 17B69, 81R60, 81T13, 81T40  \end{small}}

\maketitle

\newpage

\tableofcontents

\mainmatter

\newpage

%%%%%%%%%%%%%%%%%%%%%%%%%%%%%%%%%%%%%%%%%%%%%%%%%%%%%%%%%%%

\section{Introduction}

\subsection{Motivation}

The Alday--Gaiotto--Tachikawa (AGT) correspondence \cite{Alday:2009aq} relates four-dimensional gauge theories to two-dimensional conformal field theories whose underlying vertex algebras are $\mathcal W$-algebras. On the gauge theory side, one is interested in moduli spaces of instantons, while on the conformal field theory side, the objects of interest are conformal blocks.  This physics correspondence translates into an interesting connection between geometry and representation theory, which was phrased mathematically in a theorem of Schiffman and Vasserot \cite{SV} and Maulik and Okounkov \cite{MaulikOk}.
The geometric object of interest is  the moduli space of rank $r$ torsion-free coherent sheaves on $\mathbb P^2$ with a choice of framing at infinity. The connection to representation theory is that the equivariant Borel--Moore homology carries an action of a $\mathcal W$-algebra of $\mathfrak{gl}_r$ via correspondences. In fact, the equivariant Borel--Moore homology is isomorphic to a Verma module, which is
specified by the equivariant torus parameters $(\epsilon_1,\epsilon_2,\mathbf{Q})$ of the torus action on the framing. The level of the $\mathcal W$-algebra is determined by the torus action on the $\mathbb{P}^2$. 
The $\mathcal W$-algebra action is quasi-unitary with respect to the intersection pairing and most importantly the fundamental class, also called the Gaiotto state, is a Whittaker vector for the $\mathcal W$-algebra, as first conjectured by Gaiotto in \cite{Gai}. Nekrasov's partition function, which is the norm of the Gaiotto state, is then the norm of the Whittaker vector. 

This correspondence has been generalized in various ways: to all simple and simply-laced Lie groups \cite{Braverman:2014xca}, to spiked instantons \cite{Rapcak:2018nsl}, and to three complex dimensions \cite{Chuang:2019qdz}. Furthermore, a general program for real four-dimensional manifolds has been outlined by Feigin and Gukov \cite{Feigin:2018bkf}. However, only in the first generalization to simply-laced Lie groups is the role of Whittaker vectors presently understood. 

It is fair to say that the AGT correspondence has led to a deep connection between geometry and representation theory of $\mathcal W$-algebras, and that this connection is still  far from being completely understood. It  serves however as inspiration for progress. For example, on the geometric side  a major insight was the Nakajima--Yoshioka blow-up equations, asserting that the instanton partition function for the blow-up of $\mathbb{P}^2$ decomposes into an infinite sum over a lattice of torus parameters of products of partition functions \cite{Nakajima:2003pg}. Verma modules (and thus Whittaker vectors) satisfy a similar blow-up equation due to the recent insight that integrable representations of affine vertex algebras serve as a translation functor of $\mathcal W$-algebras and their modules \cite{Bershtein:2013oka, Arakawa:2020oqo}. 

Our aim is to relate this picture to topological recursion and its reformulation by Kontsevich and Soibelman in terms of Airy structures. More precisely, the key observation in this article is that Whittaker vectors can be constructed as partition functions of higher Airy structures, and thus are computed by topological recursion. The topological recursion is a universal recursive structure initially identified by Eynard, Chekhov and Orantin in the study of large $N$ expansion of matrix models \cite{E1MM,CE06,EO2MM,EORev,BEO} building on earlier work of Ambj\o{}rn--Chekhov--Makeenko \cite{ACM}. In the past fifteen years, it has found many other applications in enumerative geometry and in the geometry of the moduli space of curves e.g. \cite{EOwp,EInter,Norbu,DBOSS,TheseLewanski,Milanov,DoNorburyBessel,Norclass,ACEH}, in integrable systems \cite{BEInt,dxdsys,Belliard1,BDBKS} and in particular Painlev\'e equations \cite{TRPainleve}, in quantization problems \cite{MulaseDumitrescu,MarchalOrantin,EynardGF}, in two-dimensional conformal field theory, and in knot theory  \cite{DiFuji2,BEknots,BESeifert,BorotBrini}. It can be interpreted as a recursion on the topology of worldsheets \cite{EOGeo,GRpaper}. We refer to \cite{EORev,EynardICM} for reviews and \cite{Ebook,BorotLN} for introductory material. In the approach of Chekhov--Eynard--Orantin, the initial data for the recursion is packaged in a spectral curve. In particular, the Bouchard--Klemm--Mari{\~{n}}o--Pasquetti conjecture \cite{BKMP} later proved in \cite{EOBKMP,Liu1} states that taking as initial data the mirror curve of a toric Calabi--Yau $3$-fold $\mathfrak{X}$, topological recursion computes the all-genera open Gromov--Witten invariants  of $\mathfrak{X}$ and can be considered as an appropriate definition of the $B$-model topological string theory. From the physical perspective, this proposal is consistent with geometric engineering of gauge theories from topological strings, the identification of the four-dimensional limit of the mirror curve with the Seiberg--Witten curve, and the known fact that the genus $0$ part (prepotential) of the Nekrasov partition function is governed by special geometry on the Seiberg--Witten curve in a way compatible with the properties of the genus $0$ sector of the topological recursion.

In this article, combining representation theory of $\mathcal{W}$-algebras, its aforementioned connections with four-dimensional gauge theory established in \cite{SV,MaulikOk,Braverman:2014xca}, and the formalism of Airy structures, we prove that the topological recursion for the $\Lambda \rightarrow 0$ limit of a spectral curve (or a non-commutative version thereof) governs the all-genera expansion of certain Whittaker vectors and the corresponding instanton partition functions. Here $\Lambda$ is an energy scale and is treated as a formal parameter near $0$. The extension of this result to finite values of $\Lambda$ will be treated in a subsequent work\footnote{Even for finite $\Lambda$, we stress that the spectral curve to which one applies topological recursion to get the Whittaker vector is not the Seiberg--Witten curve, but rather the ``UV curve''. The Seiberg--Witten curve itself is a double cover of the UV curve and appears when one wants to access the Nekrasov partition function via the topological recursion. We will return to this in a future work.} We see the present work as a first step towards using topological recursion techniques to put all the aforementioned topics under the same roof.

\subsection{Main results and outline}

In Section 2, we review how Whittaker vectors for $\mathcal{W}$-algebras appear in four-dimensional supersymmetric gauge theories, state Theorem~\ref{ThmBFN} (taken from  \cite{Braverman:2014xca}) which is our starting point, and recall the notion of Airy structures and topological recursion. The equivariant parameters for the standard action of $(\mathbb{C}^*)^2$ on $\mathbb{P}^2$ are denoted $(\epsilon_1,\epsilon_2)$, and we set
$$
\kappa = -\frac{\epsilon_2}{\epsilon_1},\qquad \hbar = - \epsilon_1\epsilon_2.
$$
The variable $\kappa$ relates to the level of the $\mathcal{W}$-algebra, while the parameter $\hbar$ can be interpreted as a cohomological grading \cite{Braverman:2014xca}. Throughout the article we will study formal expansions when $\hbar \rightarrow 0$, which correspond to genus expansions in string theory interpretations, while $\kappa$ can either take finite values, or be of order $\hbar^{-\frac{1}{2}}$ (the Nekrasov--Shatashvili regime, precisely defined in Section~\ref{NSregime}).

In Section 3, we first provide some background on vertex algebras. We then prove two new results. Firstly, we establish the existence of Whittaker vectors in Theorem \ref{whittaker}, and secondly relate the existence theorem to a particular condition on ideals in vertex algebras in Theorem~\ref{Lem24}. The latter is directly relevant for the construction of Airy structures. 
In fact our existence result holds for any simple Verma-type module of a vertex algebra that has a commutative Zhu algebra. This is useful, since one eventually would like to study Whittaker vectors and their connection to geometry for the $\mathcal W_{n, m, l}$-algebras, \emph{i.e.} the algebras that act on moduli spaces of spiked instantons \cite{Rapcak:2018nsl}.

In Section 4, we prove our first main result, \emph{i.e.} the identification of the $A$-type Whittaker vector (Gaiotto state) with the partition function of an Airy structure (Theorem~\ref{Whit1prop}) and hence its computation from topological recursion. Using structural properties of the partition function (Lemma~\ref{lem:prophbar}), we also show how to extract the corresponding instanton partition function (Proposition~\ref{instton}): as a certain sum over graphs with vertex weights specified by the topological recursion amplitudes.

In Section 5, we rewrite the topological recursion for these partition functions, first in the case $\epsilon_1 + \epsilon_2 = 0$, as a period computation on the unramified spectral curve
\begin{equation}
\label{S01w}
\prod_{a = 1}^r \big(y -  \tfrac{Q_a}{x}\big) = 0,
\end{equation}
thus highlighting the role played by the geometry of spectral curves as in Chekhov--Eynard--Orantin's original approach: this is our second main result, Theorem~\ref{lemTRA}. For comparison, the Seiberg--Witten curve in this case \cite{Marsh} is related to the curve with the following equation, up to a change of variables,
\begin{equation}
\label{AS01w} \prod_{a = 1}^r \big(y - \tfrac{Q_a}{x}\big) =  (-1)^{r + 1}\tfrac{\Lambda}{x^{r + 1}} .
\end{equation} 
For the present paper, the relevant geometry is not \eqref{AS01w} but its $\Lambda \rightarrow 0$ limit \eqref{S01w}. We comment in Section~\ref{Generald} on the treatment of more general unramified spectral curves and the differences with the Chekhov--Eynard--Orantin topological recursion that concerns \emph{ramified} spectral curves. Our result provides a geometric meaning to the unramified topological recursion whose existence was anticipated for $r = 2$ in \cite{ABCD}. For general $\epsilon_1,\epsilon_2$, we find a recursive period computation that is best expressed in terms of the geometry of the regular D-module on $\mathbb{C}$ generated by
\begin{equation}
\label{dnsidun} \hat{\mathcal{Y}} = \prod_{a = 1}^r \big((\epsilon_1 + \epsilon_2)\partial_x + \tfrac{Q_a - (\epsilon_1 + \epsilon_2)}{x}\big). 
\end{equation} 
This is our third main result, Theorem~\ref{lemTRAarbit}. One can view this as a non-commutative generalization of the topological recursion. In Section~\ref{NCTR} we comment on the relation with non-commutative versions of the topological recursion that have already appeared in the literature \cite{CEMq1,CEMq2,BelEyn0}, mainly for the $r = 2$ cases; the $r > 2$ case that we propose in this work is new and its structure is more intricate.

In Sections \ref{Sec:AiryB}-\ref{Sec:AiryCD}, we extend Theorem~\ref{Whit1prop} to the B-type (Theorem~\ref{BAiry}), the C-type (Theorem~\ref{Prop73C}) and the D-type (Theorem~\ref{prop:typeD}) Whittaker vectors. In the B and C cases, our result is conditional to a conjectural check of the ``subalgebra property" of Airy structures, which is equivalent to a conjecture on the existence of Whittaker vectors: it is expected to be true in physics \cite{KMST} but currently lacks a mathematical proof. In the B case, we also describe in Theorem~\ref{TRBcase} the spectral curve description of the corresponding topological recursion, which involves the geometry of a two-fold ramified cover of the unramified spectral curve
$$
\prod_{a = 1}^n \big(y - \tfrac{Q_a}{2 x}\big) = 0.
$$
We note that Giacchetto, Kramer and Lewa\'nski have a(n apparently different) B-type topological recursion for ramified spectral curves also appears in enumerative geometry and integrable systems, in the context of spin Hurwitz numbers and BKP tau functions \cite{GKLBspin}.

Although the techniques of this paper can \emph{in principle} be adapted to the EFG cases, there are currently some shortcomings. While a geometric construction of Whittaker vectors of type $E$ can be found in \cite{Braverman:2014xca} and an explicit basis of strong generators for $\mathcal{W}(E_6)$ is described in \cite{ABCDEFG}, their expression is rather lengthy ($E_7$ and $E_8$ would be even more complicated).  Hence it seems rather difficult to formulate topological recursion on an explicit spectral curve in these cases. 
In physics it is believed that Whittaker vectors of modules of orbifolds of $\mathcal W$-algebras of simply-laced types are the physical relevant objects for non-simply-laced types and Whittaker vectors for such orbifolds are conjectured to exist in \cite{ABCDEFG}.

 \subsection{Notational conventions}
 \label{s:notation}
 
 We use $[n]$ to denote the set of integers $\{1,\ldots,n\}$. We use the notation $\mathbb{Z}_{>0} = \{1,2,3,\ldots \}$ for the set of positive integers and $\mathbb{Z}_{\geq 0} = \{0,1,2,\ldots\}$ for the set of non-negative integers. $\mathbf{L} \vdash [n]$ means that $\mathbf{L}$ is an unordered sequence of pairwise disjoint non-empty subsets of $[n]$ whose union is $[n]$, and $|\mathbf{L}|$ denotes the length of this sequence. We use $\sqcup$ for a disjoint union. If $z_1,\ldots,z_n$ is a $n$-tuple of variables and $I$ is a subset of $[n]$, we denote $z_I = (z_i)_{i \in I}$.
 
 When talking about modes of the fields in a vertex algebra, we will always use regular math fonts to denote the standard modes, and sans-serif fonts to denote the $\hbar$-rescaled modes as defined in Section \ref{sec:filtrations}. To be consistent, we will also use sans-serif fonts to denote the operators of an Airy structure in the algebra $\mathcal{D}^{\hbar}$.

\subsection*{Recent work} 
After the first version of this article was posted, Osuga \cite{Osuga} has extended some of our results to the case of Gaiotto vectors of $\mathcal{N} = 1$ superconformal blocks, whose square-norm is conjecturally related to instanton counting in $\mathcal{N} = 2$ supersymmetric gauge theories on $\mathbb{C}^2/\mathbb{Z}_2$. This extension is based on the formalism of super Airy structures and super topological recursion developed in \cite{BCHORS,BOsuga}.

\subsection*{Acknowledgments}

We thank Satoshi Nawata for explanations and references on Gaiotto vectors and Argyres--Douglas theories, Albrecht Klemm and Davide Scazzuso for remarks on Seiberg--Witten curves, and Andrew Linshaw for helpful discussions. We thank Kento Osuga for pointing out an error in Proposition~\ref{instton} in an earlier version of this work, and for various related and unrelated helpful discussions.  V.B. and T.C. acknowledge the support of the National Science and Engineering Research Council of Canada. G.B. and N.K.C. acknowledge the excellent working conditions at the Max Planck Institute for Mathematics, Bonn, where part of their work was conducted.

%%%%%%%%%%%%%%%%%%%%%%%%%%%%%%%%%%%%%%%%%%%%%%%%%%%%%%%%%%%%%%%%%%%%5

\section{Preliminaries}

\subsection{Instanton moduli spaces and \texorpdfstring{$\mathcal{W}$}{W}-algebras}

\label{SecInstanton}

As a starting point, we review the connection between instanton moduli spaces and $\mathcal{W}$-algebra modules for arbitrary simply-laced Lie groups, following \cite{Braverman:2014xca}. We mostly use the notation of \textit{loc.cit.} The main result is that the Gaiotto vector for the instanton moduli space, whose norm computes the instanton partition function for pure $\mathcal{N}=2$ four-dimensional supersymmetric gauge theory, is a Whittaker vector for a $\mathcal{W}^{\mathsf{k}}(\mathfrak{g})$-module at an appropriate level. Such Whittaker vectors are the main objects of study in this paper.

\subsubsection{Uhlenbeck spaces and intersection cohomology}\label{Sec:Uhlenbeck}
We consider a simple and simply-laced Lie group $ G $ over $ \mathbb C $. The corresponding Lie 	algebra is denoted by $ \mathfrak g $ and the Cartan subalgebra by $ \mathfrak h $. The dual Coxeter number is denoted $h^{\vee}$. When $G$ is of type $A$, the \emph{instanton number} of a $ G $-bundle on $ \mathbb P^2 $ is defined as the second Chern class of the associated complex vector bundle; it is defined more generally in \cite[Section 2.1]{Braverman:2014xca}.

We define $ \operatorname{Bun}_G^d $ to be the moduli space of $ G $-bundles over $ \mathbb P^2 $ with instanton number $ d \geq 0$ and a fixed trivialization on the line at infinity $ l_\infty \subset \mathbb{P}^2$ (called the \emph{framing at infinity}).  This is a smooth quasi-affine variety of dimension $ 2d h^\vee $. There exists a partial compactification of $ \operatorname{Bun}_G^d $ to an affine variety $ \mathcal U^d_G $, called the \emph{Uhlenbeck space of $ G $}.
 
 Set-theoretically, the Uhlenbeck space $ \mathcal U^d_G  $ can be described as 
 \[
	 \mathcal U^d_G  = \bigsqcup_{d' = 0}^d \operatorname{Bun}_G^{d'} \times {\rm Sym}^{d-d'}(\mathbb A^2).
 \]
 As an algebraic variety, $ \mathcal U^d_G $ is always singular, and carries an action of $ G \times  \operatorname{GL}_2(\mathbb{C}) $, where $ G $ acts by changing the framing at infinity, while $ \operatorname{GL}_2(\mathbb{C}) $ acts on the underlying $ \mathbb P^2 $. We will be interested in the action of the subgroup $ \mathbb G := G \times \mathbb{T}^2 $, where $\mathbb{T}^2 \subset  \operatorname{GL}_2(\mathbb{C}) $ is the diagonal torus. We consider the Cartan Lie subalgebra $ \mathfrak h \times \mathbb C^2 $ of  $ \mathbb G $, and denote an element of it by 
 \[
	 (Q, \epsilon_1, \epsilon_2), \qquad \text{where } Q =  (Q_1,\ldots, Q_{\ell}),
 \] and $\ell$ is the rank of $G$.
 
 We will be interested in the $ \mathbb G $-equivariant intersection cohomology $ \op{IH}^*_{\mathbb G} (\mathcal U^d_G)  $ of the Uhlenbeck space with respect to a stratification that we do not describe here. We note that it is a module over $\mathbb{A}_G :=  H^*_{\mathbb G} (\op{pt}) \cong \mathbb C[\mathfrak h \times \mathbb C^2 ]^{\mathfrak{W}} $, the ring of regular functions on $ \mathfrak h \times \mathbb C^2 $ that are invariant under the Weyl group $\mathfrak{W}$. We use $\mathbb{F}_G$ to denote the field of fractions of $\mathbb{A}_G$. We note that there is a canonical unit cohomology class $ \ket{1^d} \in \op{IH}^*_{\mathbb G} (\mathcal U^d_G) $ for every $ d \geq 0 $.
 
 By localizing  to $ \mathbb{F}_G $, we introduce the following $\mathbb{F}_G$-vector spaces
 \[
 M_{\mathbb F_G}^d(Q) := \op{IH}^*_{\mathbb G} (\mathcal U^d_G) \mathop{\otimes}_{\mathbb A_G} \mathbb F_G, \qquad M_{\mathbb F_G}(Q) := \bigoplus_{d \geq 0} M^{d}_{\mathbb F_G} (Q).
 \] 
 The \emph{Gaiotto vector} is defined to be:
 \begin{equation}
 \label{Gaiodef}
\ket{\mathfrak{G}} := \sum_{d \geq 0} \Lambda^{h^\vee d}\,\ket{1^d} \in M_{\mathbb{F}_G}(Q)[\![\Lambda]\!],
\end{equation}
where $\Lambda$ is a (redundant) parameter introduced for convenience.

 \subsubsection{\texorpdfstring{$ \mathcal W $}{W}-algebra action}
 
 Let $\ell$ be the rank of  $\mathfrak{g}$ and $d_1,\ldots,d_{\ell}$ the primitive exponents. For each $\mathsf{k} \in \mathbb{C}$, $\mathcal{W}^{\mathsf{k}}(\mathfrak{g})$ is a vertex operator algebra (VOA) strongly and freely generated by fields $(W^i)_{i = 1}^\ell$, with $W^i$ of conformal degree $d_i + 1$. We have $d_1 = 1$ and the field of lowest degree can be chosen as the Virasoro field $W^1 = L$. There is a canonical choice up to scale for the generator of maximal conformal degree (from representation theory), and the choice can be pinned down up to a sign (from geometry). We denote the  modes of these generators by
 $$
 W^i(z) = \sum_{m \in \mathbb{Z}} \frac{W^i_m}{z^{m + d_i + 1}},
 $$
 and the shifted level by $\kappa := \mathsf{k} + h^\vee$ . As we will explain in Section \ref{sec:filtrations}, we introduce a formal parameter $\hbar$ that keeps track of the filtration by conformal weight. We denote by 
 \begin{equation}\label{eq:rescaledW}
 \mathsf{W}^i_m = \hbar^{\frac{d_i+1}{2}} W^i_m
 \end{equation}
  the rescaled modes, which correspond to the rescaled modes $\widetilde{W}_m^{(i)}$ presented in Appendix B of \cite{Braverman:2014xca}.
 
 Braverman, Finkelberg and Nakajima \cite{Braverman:2014xca} show that that the Gaiotto vector \eqref{Gaiodef} is a Whittaker vector for a $\mathcal W^{\mathsf{k}}(\mathfrak{g})$-module at a well-chosen level. The result in the case of $G = {\rm SL}_r(\mathbb{C})$ was obtained previously by Schiffman and Vasserot \cite{SV} and Maulik and Okounkov \cite{MaulikOk}. More precisely, a $\mathbb{F}_G$-localized version of the main result of \cite{Braverman:2014xca} is the following:
 \begin{thm}[{\cite[Theorem 1.4.1]{Braverman:2014xca}}]
	 \label{ThmBFN} The space  $M_{\mathbb{F}_G}(Q)$ can be equipped with the structure of a $\mathcal W^{\mathsf{k}}(\mathfrak g)$-module with shifted level $\kappa = - \frac{\epsilon_2}{\epsilon_1}$, satisfying:	
	 \begin{enumerate}
	 \item The module $ M_{\mathbb F_G}(Q) $ is isomorphic to the Verma module of $ \mathcal W^{\mathsf{k}}(\mathfrak g) $ with highest weight $\lambda  = \frac{Q}{\epsilon_1} - \rho$, where $ \rho $ is the half-sum of the positive roots of $ \mathfrak g $.
	 \item Under the above isomorphism, the bilinear pairing $  \braket{\ \cdot\ |\  \cdot\  }  $ \footnote{ This bilinear pairing is described in more detail in Section~\ref{Sec:instanton} and is referred to as the Kac--Shapovalov form in \cite{Braverman:2014xca}. }  on the Verma module corresponds to a twisted Poincar\'e pairing on $ M_{\mathbb F_G}(Q) $ and $M^{d}_{\mathbb F_G} (Q)$ is an eigenspace for $L_0$ with eigenvalue $d$.
	 \item For any $d \geq 0$, $i \in [\ell]$ and $m > 0$, we have
	 \[
		 \mathsf{W}^i_m \ket{1^d} = \pm \delta_{i,\ell} \delta_{m,1} \ket{1^{d-1}}
	 \]
	 for some sign that may depend on $d$ and $\mathsf{k}$, and by convention $\ket{1^{-1}} = 0$.	 \end{enumerate}	 
 \end{thm}
\begin{rem}
It is expected but not proved that the sign is always $+1$, and we will assume that it is positive for the purposes of this paper.
\end{rem}
For the Gaiotto vector \eqref{Gaiodef} we get
\begin{equation}\label{eq:gaiotto}
\forall i \in [\ell],\,\,\forall m > 0,\qquad \mathsf{W}^i_m \ket{\mathfrak G} = \Lambda^{h^\vee} \delta_{i,\ell} \delta_{m,1} \ket{\mathfrak G},
\end{equation}
which corresponds to the state identified by the AGT correspondence. A vector $\ket{\mathfrak{G}}$ satisfying this equation is called a \emph{Whittaker} vector --- see Definition~\ref{d:whittaker}. In comparison to \cite{Braverman:2014xca}, our $Q$ is their $\mathbf{a}$, our $\Lambda^{2h^\vee}$ is their $Q$, and our $\mathsf{W}^i_m$ correspond to their rescaled generators $\widetilde{W}^{(i)}_m$ presented in Appendix B. In fact, \cite{Braverman:2014xca} establishes an integral version of the  correspondence by working with (four different) non-localized versions of the intersection cohomology groups on the geometric side. On the algebraic side, the authors define a  $ \mathcal W $-algebra over the ring $ \mathbb A_G$ denoted by $ \mathcal W_{\mathbb A_G} (\mathfrak g) $ which localizes to the $ \mathcal W^{\mathsf{k}}(\mathfrak g) $ introduced above. They also prove some general properties such as the presentation of $ \mathcal W^{\mathsf k} (\mathfrak g)  $ as a kernel of screening operators. Working with the integral version of the $ \mathcal W $-algebra allows them to exploit a certain cohomological grading that they call $ ^c\op{deg} $ and which will correspond to our $ \hbar $-grading.

\begin{rem}
	The connection between Whittaker vectors and instanton moduli spaces given in Theorem \ref{ThmBFN} is only proved for the case where $G$ is a simply-laced Lie group \cite{Braverman:2014xca}. However, a similar connection between Whittaker vectors and instanton moduli spaces is expected to hold for the general case of $G$ an almost simple simply-connected algebraic group. In particular, in the case of $G$ being of type $B$ (resp. $C$), the proposal from physics \cite{KMST} is that the Gaiotto vector should be a Whittaker vector for an appropriate twisted $\mathcal{W}^{\mathsf{k}}(\mathfrak{g})$-module, where $\mathfrak{g}$ is a of type $A$ (resp. $D$). We will study such Whittaker vectors from the point of view of Airy structures in Sections \ref{Sec:AiryB} and \ref{Sec:AiryCD}.
\end{rem}

 \subsubsection{Instanton partition function}\label{Sec:instanton}
 Here we describe the pairing, denoted $  \braket{\ \cdot\ |\  \cdot\  }  $  in agreement with the bra-ket notation, on the Verma modules of $ \mathcal W^{\mathsf{k}}(\mathfrak g) $. There exists a canonical embedding of $\mathcal{W}^{\mathsf{k}}(\mathfrak{g})$ into the Heisenberg VOA algebra $\mathcal{H}(\mathfrak{g})$. We will fix an orthonormal basis of $\mathfrak{h}$, such that the latter is generated by the modes $\mathsf{J}_m^a$ indexed by $a \in [\ell]$ and $m \in \mathbb{Z}$ that satisfy the commutation relations\footnote{The $ \mathsf{J}^a_m $ defined here correspond to the $ \widetilde{P}^a_m $ of  \cite{Braverman:2014xca} (see section 6.3 of \textit{loc.cit.}).}
\begin{equation}
\label{Heissanh}
\forall m,n \in \mathbb{Z},\quad \forall a,b \in [\ell],\qquad [{\mathsf{J}}_m^a,{\mathsf{J}}_{n}^{b}] = -\epsilon_1\epsilon_2\, m \,\delta_{m + n,0}\delta_{a,b}.
\end{equation}
Consider the linear involution $ \iota$ that acts on the Heisenberg algebra modes as follows,
\begin{equation}
\label{invol}
\iota({\mathsf{J}}^a_m) = - {\mathsf{J}}^a_{-m} + 2(\epsilon_1+\epsilon_2) \delta_{m,0},
\end{equation}
and extend it to monomials in modes $\mathsf{X}_1,\ldots,\mathsf{X}_p$ by setting $\iota(\mathsf{X}_1\cdots \mathsf{X}_p) = \iota(\mathsf{X}_1)\cdots \iota(\mathsf{X}_p)$. The pairing on the Verma module $M(\lambda)$ corresponding to a highest-weight vector $\ket{\lambda}$ is defined uniquely by requiring  $\langle\lambda|\lambda\rangle = 1$ and for any mode $\mathsf{X}$ of $ \mathcal W^{\mathsf{k}}(\mathfrak g) $.
		 \[
 	 	\forall u,v \in M(\lambda),\qquad \langle \lambda|\lambda\rangle = 1, \qquad \langle u|\mathsf{X} v\rangle = \langle \iota(\mathsf{X}) u | v\rangle.
 	 	 \]
 	 	 
 	 	 Using this form,  the (Nekrasov) instanton partition function for pure $\mathcal{N} = 2 $ four-dimensional supersymmetric gauge theory gets identified with the norm of the Gaiotto vector 
 	 	 \begin{equation}\label{eq:inst}
	 	 	 \mathfrak{Z}(Q, \epsilon_1, \epsilon_2,\Lambda) = \sum_{d \geq 0} \Lambda^{2 h^\vee d} \braket{1^d|1^d}\quad \in \mathbb{F}_{G}[\![\Lambda]\!],
 	 	 \end{equation}
 	 	 where $\Lambda$ can be interpreted as an energy scale.

\subsection{Airy structures and topological recursion}

\label{SecAiryTR}

The Whittaker vectors introduced in the previous section are the main object of study in this paper. In particular, we will show that we can construct Whittaker vectors for various types of $\mathcal{W}$-algebras using Airy structures and topological recursion. Let us now review the basic definition and properties of Airy structures, and briefly outline the connection with topological recursion.

\subsubsection{Airy structures}

The concept of quantum Airy structures was first introduced by Kontsevich and Soibelman in \cite{KS}, as an abstract algebraic structure generalizing the Chekhov--Eynard-Orantin topological recursion \cite{CE06,EO}. The construction was further studied in \cite{ABCD} --- see also \cite{BorotLN} for a pedagogical introduction. 

The original construction of quantum Airy structures was generalized to higher Airy structures in \cite{BBCCN18} building on earlier work of Milanov \cite{Milanov}, as an algebraic counterpart to the higher topological recursion of \cite{BHLMR,BE2,BE}. In this paper, when we talk about Airy structures, we mean higher crosscapped quantum Airy structures, as defined in \cite{BBCCN18}: we think of the original Kontsevich--Soibelman construction as being the special case in which Airy structures are quadratic and have only integral powers of $\hbar$ (crosscapped means that we allow half-integer powers of $\hbar$). Here we only review the basic facts about Airy structures, and refer the interested reader to Sections 2 of \cite{BBCCN18} and \cite{BKS} for more details.

Let $V$ be a $\mathbb{C}$-vector space. If $V$ is finite-dimensional, let $D$ be its dimension, and $I=[D]$. If $V$ is countably infinite, let $I = \mathbb{Z}_{>0}$.
 Let $(y_i )_{i \in I}$ be a basis for $V$, and $( x_i )_{i \in I}$ be the dual basis for $V^*$. Let $\mathcal{M}^{\hbar}= \mathbb{C}[\![(x_i)_{i \in I}]\!][\![\hbar^{\frac{1}{2}}]\!]$ be the algebra of formal functions on $V$, and $\mathcal{D}^{\hbar} = \mathbb{C}[\![(x_i)_{i\in I}, (\hbar \partial_{x_i})_{i \in I} ]\!][\![\hbar^{\frac{1}{2}}]\!]$ be the completed algebra of differential operators on $V$. We introduce an algebra grading on $\mathcal{M}^{\hbar}$ and $\mathcal{D}^{\hbar}$ as follows:
\begin{equation}\label{eq:grading}
	\deg(x_i) = \deg(\hbar \partial_{x_i}) = \deg(\hbar^{\frac{1}{2}})=1.
\end{equation}

\begin{defin}\label{d:Airy}
	An \emph{Airy structure in normal form} is a family of differential operators $S = \{\mathsf{H}_i\,\,|\,\,i \in I\}$ with $\mathsf{H}_i \in \mathcal{D}^{\hbar}$ for any $i \in I$, and satisfying the following two requirements:
	\begin{enumerate}
		\item {\bf Degree condition:}
			\begin{equation}
				\forall i \in I,\qquad \mathsf{H}_i = \hbar \partial_{x_i} + \mathsf{P}_i,
			\end{equation}
			where $\mathsf{P}_i \in \mathcal{D}^{\hbar}$ is a sum of terms of degree $\geq 2$. 
		\item {\bf Subalgebra condition:}
			The left $\mathcal{D}^{\hbar}$-ideal generated by $S$ (denoted $\mathcal{D}^{\hbar} \cdot S$)  is a graded Lie subalgebra of $\mathcal{D}^{\hbar}$ equipped with the Lie bracket $\hbar^{-1}[\cdot,\cdot]$. That is, 
			\begin{equation}
				[\mathcal{D}^{\hbar} \cdot S, \mathcal{D}^{\hbar} \cdot S] \subseteq \hbar \mathcal{D}^{\hbar} \cdot S.
			\end{equation}
			In other words, there exists $\mathsf{f}_{i,j}^k \in \mathcal{D}^{\hbar}$  such that $[\mathsf{H}_i,\mathsf{H}_j] = \sum_{k \in I} \hbar \mathsf{f}_{i,j}^k\,\mathsf{H}_k$ for any  $i,j \in I$.
			\end{enumerate}

			An \emph{Airy structure} is a set of differential operators $S = \{ \mathsf{H}_i ~|~ i \in I \}$ such that $\hat{\mathsf{H}}_i = \sum_{j \in I} \mathsf{N}_{i,j} \mathsf{H}_j$ is an Airy structure in normal form, for some invertible matrix $\mathsf{N}$.
\end{defin}

We note that being an Airy structure does not depend on a choice of basis for $V$, but being an Airy structure in normal form does.

\begin{rem}\label{r:infinite}
When $V$ is infinite-dimensional (as it will always be the case in this article), $\mathcal{D}^{\hbar}$ should be a suitable completion of the algebra of differential operators such that $\mathcal{M}^{\hbar}$ is a well-defined $\mathcal{D}^{\hbar}$-module. We will not dwell on this completion here and rather refer the interested reader to \cite[Section 2.1.2]{BKS}. The key point that can be routinely checked throughout the paper (see e.g. Remark~\ref{refinite}) is that all computations we do give finite answers.
\end{rem}

The key result in the theory of Airy structures is the existence and uniqueness theorem due to Kontsevich and Soibelman.

\begin{thm}[{\cite[Theorem 2.4.2]{KS}}]\label{t:KS}
	Let $S= \{\mathsf{H}_i\,\,|\,\,i \in I\}$ be an Airy structure. Then the system of differential constraints
	\begin{equation}
		\forall \mathsf{H}_i \in S, \qquad \mathsf{H}_i \mathcal{Z} = 0,
	\end{equation}
	has a unique solution of the form $\mathcal{Z} = e^F$, where $F \in \hbar^{-1} \mathcal{M}^{\hbar}$ has only terms of positive degree with respect to the grading \eqref{eq:grading}.
	We call $\mathcal{Z}$ the \emph{partition function} of the Airy structure $S$.
\end{thm}
Since $\mathcal{Z} = e^F$ with $F \in \hbar^{-1} \mathcal{M}^{\hbar}$ containing only terms of positive degree, we can write an explicit decomposition as
	\begin{equation}\label{eq:PF}
		 F = \sum_{\substack{g \in \frac{1}{2} \mathbb{Z}_{\geq 0}, n \in \mathbb{Z}_{>0} \\ 2g-2+n >0}} \frac{\hbar^{g-1}}{n!} \sum_{\alpha_1, \ldots, \alpha_n \in I} F_{g,n}[\alpha_1\, \cdots\, \alpha_n] \prod_{i=1}^n x_{\alpha_i} .
	\end{equation}
	Concretely, Airy structures are interesting because of the existence and uniqueness of this solution to the differential constraints. Once existence of a solution of the form \eqref{eq:PF} is established, the differential constraints can be used to construct it quite explicitly. Indeed, applying the differential constraints, one obtains a recursion on $2g - 2 + n$ for the coefficients $F_{g,n}[\alpha_1\, \cdots\, \alpha_n]$ (which is the degree of the corresponding terms in $F$ according to the grading \eqref{eq:grading}). We refer the reader to \cite[Section 2.2.2]{BBCCN18} for the explicit form of the recursion.

	\subsubsection{Topological recursion}
	
	\label{s:TR}

In fact, the original motivation of Kontsevich and Soibelman in defining quantum Airy structures was to construct an algebraic structure underpinning the Chekhov--Eynard--Orantin (CEO) topological recursion \cite{CE06,EO}. The CEO topological recursion is a formalism that constructs a collection of symmetric differential forms $\omega_{g,n}$, with $g \in \frac{1}{2} \mathbb{Z}_{\geq 0}$, $n \in \mathbb{Z}_{>0}$, and $2g-2+n >0$, through residue analysis on a spectral curve\footnote{In the original formulation of topological recursion \cite{CE06,EO} and its higher generalization \cite{BE,BHLMR,BE2}, it was assumed that $g \in \mathbb{Z}_{\geq 0}$, but the formalism is easily extended to $g \in \frac{1}{2} \mathbb{Z}_{\geq 0}$ \cite{BBCCN18,BKS}.}. It turns out that the recursive formalism can be reformulated as a special case of Airy structures.

More precisely, the differentials $\omega_{g,n}$ constructed by topological recursion have a very specific pole structure. As a result, there exists a particular basis of differentials on the spectral curve such that the expansion of the $\omega_{g,n}$ in this basis is finite. From the coefficients of this expansion, which we can write as $F_{g,n}[\alpha_1\, \cdots\, \alpha_n]$, one can construct a partition function $\mathcal{Z}$, and show that, for any spectral curve satisfying some admissibility condition, $\mathcal{Z}$ is the partition function of an Airy structure constructed as a particular module for a direct sum of $\mathcal{W}^{\mathsf{k}}(\mathfrak{gl}_r)$-algebras  at self-dual level \cite{BBCCN18}. The recursive structure inherent in the CEO topological recursion is then recovered by applying the differential constraints on $\mathcal{Z}$.

 All Airy structures coming from vertex algebras can in principle be converted to a spectral curve topological recursion, adapting the strategy presented in \cite{BKS}, but the details of how the geometry of the spectral curve appears in the recursion depends on the vertex algebra module under consideration. We will work these details out in Sections~\ref{Sec:TRA} and \ref{Sec:TRB} for the Airy structures of type A and B that are met in this paper.

\section{Whittaker modules}

We saw in Section \ref{SecInstanton} that, for a simply-laced Lie group $G$, the Gaiotto vector, whose norm gives the instanton partition function for pure $\mathcal{N}=2$ four-dimensional supersymmetric gauge theory, is a Whittaker vector for a $\mathcal{W}^{\mathsf{k}}(\mathfrak{g})$-module at an appropriate level, where $\mathfrak{g}$ is the Lie algebra of $G$. In this section we study in more details the definition and properties of Whittaker vectors for $\mathcal{W}^{\mathsf{k}}(\mathfrak{g})$-modules, where $\mathfrak{g}$ is an arbitrary simple Lie algebra. We also explain how Airy structures can be used to construct Whittaker vectors.

\subsection{Background on vertex algebras}
\label{SecBackground}

\subsubsection{Vertex algebras}

The best reference for our purposes is  \cite{Ara}. Let $V$ be a vertex algebra and 
\[
Y(v, z) = \sum_{n \in \mathbb Z} v_n z^{-n-1} \quad \in  \ \text{End}(V)[\![z^{\pm 1}]\!]
 \]
be the field associated to $v\in V$. The $v_n$ are called modes. A vertex algebra is usually defined over $\mathbb C$; however, vertex algebras are defined over general commutative rings and we will eventually consider vertex algebras over the field $\mathbb C(\hbar^{\frac{1}{2}})$ for some indeterminate $\hbar^{\frac{1}{2}}$. In the latter case we will always write $\mathsf{V}$ for $V \otimes_{\mathbb C} \mathbb C(\hbar^{\frac{1}{2}})$.

 The vacuum is denoted by $\ket{0}$ and the translation operator by $T$. It satisfies
 \[
 [T, Y(v, z)] = Y(Tv, z) = \partial_z Y(v, z). 
 \]
A Hamiltonian $H$ is an element in  $\text{End}(V)$ that satisfies
\[
\forall v \in V, \qquad [H, Y(v, z)] = Y(Hv, z) + z Y(Tv, z),
\]
and  acts semisimply on $V$. For example if $V$ is conformal in the sense that it has a non-trivial action of a Virasoro vertex algebra, then the zero-mode of the Virasoro field is a Hamiltonian. If $V$ is a vertex algebra with Hamiltonian $H$, a vector of $H$-eigenvalue $\Delta$ is said to be of conformal weight $\Delta$ and the subspace of vectors of conformal weight $\Delta$ is denoted by $V_{-\Delta}$. We assume that 
\begin{equation}\label{eq:posV}
V = \bigoplus_{n \in \mathbb{Z}_{\geq 0}} V_{-n}
\end{equation}
and that $V_0 = \mathbb C. \ket{0}$. A mode $v_n$ is called positive (respectively negative)  if $[H, v_n] = \Delta v_n$ for some $\Delta <0$ (respectively $\Delta >0$). 
We denote by $\cA$ the suitably completed algebra (the current algebra) of modes of $V$ and by $\cA_\pm$ the subalgebra generated by positive (resp. negative) modes and by $\cA_{\geq 0}, \cA_{\leq 0}$  the subalgebras generated by non-negative respectively non-positive modes. See  \cite[Section 3]{Ara} for technical details on current algebras and the Appendix of \cite{Ara} for details on the completion.

Let $I$ be an index set and $U = \{ w^i \,\,|\,\, i \in I \} \subset V$.
The set $U$ strongly generates $V$ if $V$ is spanned by  $w^{i_1}_{-n_1} \cdots\,  w^{i_s}_{-n_s} \ket{0}$ with $w^{i_1}, \ldots, w^{i_s} \in U$ and $n_1, \ldots, n_s > 0 $. Let us fix an ordering  $\preceq$ on the set $U$. Then, we say that the monomial $w^{i_1}_{-n_1} \cdots\,  w^{i_s}_{-n_s}$ in $\cA$ as well as the monomial 
$w^{i_1}_{-n_1} \cdots\,  w^{i_s}_{-n_s} \ket{0}$ in $V$
is lexicographically ordered if $n_j \geq  n_{j+1}$, and whenever $n_j = n_{j+1}$, we have $w^{i_j} \succeq w^{i_{j+1}}$. We say that $U$ is a minimal strong generating set if no proper subset of $U$ strongly generates $V$. The vertex algebra $V$ is freely generated by $U$ if  lexicographically ordered monomials form a basis of $V$ and such a basis is called a Poincar\'e--Birkhoff--Witt (PBW) basis.  

Three families of strongly and freely generated vertex algebras with Hamiltonian are:
\begin{enumerate}
	\item
$\mathcal{W}$-algebras \cite{Ara}, 
\item generalized free field algebras \cite{CL1}, 
\item the $\mathcal{W}_\infty$-algebras of \cite{Lin, KL}.
\end{enumerate}
All three appear in connection to geometry, and they are in fact tightly connected \cite{CL1, CL2}. 
The $\mathcal W_{n, m , l}$-algebras that act on moduli spaces of spiked instantons \cite{Rapcak:2018nsl} coincide conjecturally with the corner vertex algebras of Gaiotto and Rap{\v{c}\'a}k \cite{GR}, proven in some cases \cite{CGN}. The corner vertex algebras are quotients of $\mathcal{W}_\infty$ \cite{CL1, CL2} and hence are usually not strongly freely generated. With this in mind we will also consider non-strongly  freely generated vertex algebras.

\subsubsection{Filtrations}

\label{sec:filtrations}

We say that $\mathcal{F} = \{ \mathcal{F}_p V\,\, |\,\, p \in \mathbb Z \}$ is an increasing, separated and exhaustive filtration of $V$ if  $\mathcal{F}$ satisfies 
\[
\mathcal{F}_{p-1}V \subseteq \mathcal{F}_pV, \qquad \bigcap_{p \in \mathbb Z} \mathcal{F}_p V = \{0\}, \qquad \bigcup_{p \in \mathbb Z} \mathcal{F}_p V = V.
\]
Such a filtration is a vertex algebra filtration if in addition
$\ket{0} \in \mathcal{F}_0 V/\mathcal{F}_{-1}V$, $T(\mathcal{F}_p V) \subseteq \mathcal{F}_p V$ and $v_n(\mathcal{F}_p V) \subseteq \mathcal{F}_{p+q}V$ for all $v \in \mathcal{F}_q V$ and all $p, q, n \in \mathbb Z$. If there is a Hamiltonian $H$, we also require compatibility of the filtration with respect to $H$, that is $H(\mathcal{F}_p V) \subseteq \mathcal{F}_p V$. The associated graded vertex algebra is then 
\[
\text{gr}^{\mathcal{F}} V = \bigoplus_{p \in \mathbb Z} \mathcal{F}_p V/\mathcal{F}_{p-1}V.
\]

We continue with the  setup of the previous section and in addition require that the strong generators are eigenvectors of the Hamiltonian. The conformal weight of a vector $w^i$ is denoted $\Delta_i$.  Let $\mathcal{F}_p V$ be the subspace of $V$ spanned by elements $w^{i_1}_{-n_1} \cdots\,  w^{i_s}_{-n_s} \ket{0}$  with $\Delta_{i_1} + \cdots + \Delta_{i_s} \leq  p$. 
Then $\mathcal F = \{ \mathcal F_p V\,\, |\,\, p \in \mathbb Z \}$ is an increasing, separated and exhaustive vertex algebra filtration --- the standard filtration of Li \cite{HLi}. In fact, because of \eqref{eq:posV}, it follows that $\mathcal{F}_p V = 0$ for all $p < 0$. Moreover, the associated graded vertex algebra is commutative. This filtration induces a filtration of the current algebra $\cA$ denoted by $\mathcal F_p \cA$ \cite[Section 3.13]{Ara}. The restriction of this filtration to $\cA_-$ and $\cA_{\leq 0}$ is exhaustive.

A filtration of a $V$-module is defined in analogy. 
Let $M$ be a $V$-module, and $\mathcal G = \{ \mathcal{G}_p V\,\, |\,\, p \in \mathbb Z \}$ be an increasing, separated and exhaustive filtration of $M$. That is,  $\mathcal{G}$ satisfies 
\[
\mathcal{G}_{p-1}M \subseteq \mathcal{G}_p M, \qquad \bigcap_{p \in \mathbb Z} \mathcal{G}_p M = \{0\}, \qquad \bigcup_{p \in \mathbb Z} \mathcal{G}_p M = M.
\]
Such a filtration is a vertex algebra module filtration compatible with $\mathcal F$ if in addition
 $T(\mathcal{G}_p M) \subseteq \mathcal{G}_p M$ and $v_n(\mathcal{G}_p M) \subseteq \mathcal{G}_{p+q}M$ for all $v \in \mathcal{F}_q V$ and all $p, q, n \in \mathbb Z$. If there is a Hamiltonian $H$, we also require compatibility of the filtration with respect to $H$, that is $H(\mathcal{G}_p M) \subseteq \mathcal{G}_p M$. The associated graded module is then 
\[
	\text{gr}^{\mathcal{G}} M = \bigoplus_{p \in \mathbb Z} \mathcal{G}_p M /\mathcal{G}_{p-1}M,
\]
which is a $\text{gr}^{\mathcal{F}} V$-module. In the following we will denote such a compatible filtration by the same symbol $\mathcal F$ as the vertex algebra filtration.

Let $\mathcal F$ be Li's standard filtration, such that the associated graded vertex algebra is commutative. The algebra of modes $\cA$ is a Lie algebra with respect to the commutator $[\cdot, \cdot]$. We can turn it into a graded Lie algebra using the induced filtration. Let $\hbar$
be a formal variable and consider $\cA[\![\hbar^{\frac{1}{2}}]\!]$. We introduce a grading where $\deg \hbar^{\frac{1}{2}} = 1$, and $\deg x= 0$ for all $x\in \cA$.
Define the map $\eta: \cA \rightarrow \cA[\![\hbar^{\frac{1}{2}}]\!]$ by $\eta(x):= \hbar^\frac{p}{2} x$  for 
 $x \in \mathcal F_p \cA$, but $x \notin \mathcal F_{p-1} \cA$. We note that the image of $\eta$ is in $\cA[\![\hbar^{\frac{1}{2}}]\!]$, since for Li's standard filtration, the induced filtration on the algebra of modes satisfies $\mathcal{F}_p \mathcal{A} = 0$ for all $p<0$.
 \begin{lem}\label{l:Ahbar}
Let $\cA^{\hbar}$ be the subalgebra of $\cA[\![ \hbar^{\frac{1}{2}} ] \!]$ generated by the set  $\{ \eta(x)~|~x \in \cA \} \subset \cA[\![ \hbar^{\frac{1}{2}}]\!]$. Then $\cA^{\hbar}$ is a graded Lie subalgebra, in the sense that the commutator $[\cdot, \cdot]$ respects the grading. Moreover,
$$[\cA^{\hbar} , \cA^{\hbar}] \subseteq \hbar^{\frac{1}{2}} \cA^{\hbar}.$$
 \end{lem}
 
 \begin{proof}
 Since $\mathcal F_p \cA \cdot \mathcal F_q \cA \subseteq \mathcal F_{p+q} \cA$, $\cA^{\hbar}$ is a subalgebra of $\cA[\![\hbar^{\frac{1}{2}}]\!]$. By definition of the map $\eta$, it is clearly a graded Lie subalgebra. Moreover since the associated graded algebra of $V$ is commutative, it follows that  $[ \mathcal F_p \cA,  \mathcal F_q \cA] \subseteq \mathcal F_{p+q-1} \cA$ and hence $[\cA^{\hbar} , \cA^{\hbar}] \subseteq \hbar^{\frac{1}{2}} \cA^{\hbar}.$ 
 \end{proof}
 
As explained in Section \ref{s:notation},  we use regular math fonts to denote modes in $\cA$, and sans-serif math fonts to denote the corresponding modes in $\cA^{\hbar}$. For instance, if $x \in \cA$, then $\mathsf{x} = \eta(x) \in \cA^{\hbar}$. According to the grading, it follows that $ \deg (\mathsf{x}) = p$ when $x
 \in \mathcal F_p \cA$ but not in $ \mathcal F_{p-1} \cA$. 
 
 As mentioned earlier, we will also sometimes consider the vertex algebra over the field $\mathbb{C}(\hbar^{\frac{1}{2}})$, which we denote by $\mathsf{V} := V \otimes_{\mathbb C}\mathbb C(\hbar^{\frac{1}{2}})$.

\begin{defin}\label{d:invol}
Let $V$ be a vertex algebra with Hamiltonian $H$ and strong generating set  $U= \{w^i\,\, |\,\, i \in I\}$ consisting of eigenvectors $w^i$ of $H$ of integral conformal weight $\Delta_i$. 
Let $\rho$ be an involution of $V$ such that its action on strong generators is given by $\rho(w^i) = (-1)^{\Delta_i}w^i$. Then we say that $\rho$ is compatible with Li's standard filtration with respect to $U$.
\end{defin}

\begin{lem}\label{l:oneh}
Let $V, U$ as in Definition~\ref{d:invol} and $\rho$ an involution compatible with Li's standard filtration. Then
$$
[\cA^{\hbar} , \cA^{\hbar}] \subseteq \hbar \cA^{\hbar}.
$$
\end{lem}

\begin{proof}
		
	The involution $\rho$ induces a $\mathbb Z_2$-grading on $V$ and hence on $\cA$. Let $x$ be a homogeneous element of $\cA $ with respect to Li's standard filtration $\mathcal F$ on $V$. This implies that $x \in \mathcal F_p \cA$ but not in $\mathcal F_{p-1} \cA$ so that by the definition of Li's filtration $\rho(x) = (-1)^px$. It follows immediately that $[ \mathcal F_p \cA,  \mathcal F_q \cA] \subseteq \mathcal F_{p+q-2} \cA$, and hence $[\cA^{\hbar} , \cA^{\hbar}] \subseteq \hbar \cA^{\hbar}$.
	
\end{proof}

\begin{rem}
Some examples of vertex algebras that have an involution compatible with Li's filtration with respect to a strong generating set are
\begin{enumerate}
\item Vertex algebras that are strongly generated by elements of even conformal weight. The main examples are firstly the even $\mathcal W_\infty$-algebra of \cite{KL} and any of its quotients, these are the orthosymplectic $Y$-algebras \cite{GR, CL2} and secondly principal $\mathcal W$-algebras of Dynkin types $B_n, C_n, E_7, E_8, G_2, F_4$ and $D_{2n}$;
\item Principal $\mathcal W$-algebras of type $D_{2n+1}$, e.g.  by \cite{CL2};
\item The $\mathcal W_\infty$-algebra of \cite{Lin} (see section 5 of \cite{Lin} for its properties including the involution) and any of its quotients, these are the $Y$-algebras of unitary type   \cite{GR, CL1} and they include principal $\mathcal W$-al\-ge\-bras of type $A$;
\item The principal $\mathcal W$-algebra of $E_6$, see Appendix D of \cite{KMST};
\item Any integer graded generalized free field algebra. 
\end{enumerate}
\end{rem}

\subsubsection{Verma modules}

Let $\zhu(V)$ be Zhu's associative algebra (a quotient of the algebra of modes of weight zero) \cite{Zhu}, see also \cite[Section 3.12]{Ara}.
The Zhu functor $\zhu(\cdot)$ associates to a lower bounded $V$-module $M$ a Zhu-algebra module structure on the top level subspace of $M$. 
Let $E$ be a finite-dimensional module for $\zhu(V)$. Then $E$ induces an $\cA_{\geq 0}$-module by letting positive modes act trivially. The induced module
\[
M(E) := \cA \otimes_{\cA_{\geq 0} } E
\]
is a lower bounded graded $V$-module, which we call the Verma module with top level $E$ \cite{NT}.
The induction functor $M(\cdot)$ is left adjoint to $\zhu(\cdot)$. 
 Li's standard filtration $\mathcal F$ induces a filtration on $\zhu(V)$, since the Zhu algebra is a quotient of the algebra of zero-modes. If $E$ is filtered compatible with $\mathcal F$, then since $\mathcal F$ also induces a filtration on $\cA$, $M(E)$ must be filtered as well, \emph{i.e.}
\[
\mathcal F_p M(E) = \bigoplus_{q+r = p} \mathcal F_q \cA\cdot \mathcal F_r E.
\]
This is compatible with $\mathcal F$, since $\mathcal F_p \cA\cdot \mathcal F_q M(E) \subseteq \mathcal F_{p+q} M(E)$.  

Let $n$ be the dimension and $\{ \mu_1, \ldots, \mu_n \}$ a basis of $E$. 
Let $I$ be an index set. 
Assume that $V$ is strongly  generated by $U = \{ w^i\,\, |\,\, i \in I \} \subset V$ and all elements of $U$ are eigenvectors for the Hamiltonian $H$. 
Let $w^i \in U$ be a strong generator of conformal weight $\Delta_i$, then it is convenient to set 
\[
Y(w^i_{-1} \ket{0}, z) = \sum_{n  \in  \mathbb Z}w^i_n z^{-n-1} =  \sum_{n \in  \mathbb Z} W^i_n z^{-n-\Delta_i},
\]
\emph{i.e.}, we redefine the modes as $W^i_{n} = w^i_{n+\Delta_i -1}$. The $W^*_n$ then have weight $-n$. 
Then $M(E)$ is spanned by  monomials
\begin{equation}\label{spanningset}
	W^{i_1}_{-n_1} \cdots\,  W^{i_s}_{-n_s} \mu_i,
\end{equation}
with $w^{i_1}, \ldots, w^{i_s} \in U$, $n_1, \ldots, n_s \geq 1$, and $i \in [n]$. 
Such a monomial is in $\mathcal F_p$ with $p = \Delta_{i_1} + \cdots + \Delta_{i_s} + q_i$ with $q_i$ minimal such that $\mu_i \in \mathcal F_{q_i} E$. Thus the filtration $\mathcal F$ on $M$ is separated and exhausting. Moreover the spanning set \eqref{spanningset} can be taken to be lexicographically ordered as any reordering only adds terms in lower degrees of the filtration. 

Assume that $I$ is finite, say $|I| = \ell$, then the module $M(E)$  is graded by the Hamiltonian $H$ with finite dimensional weight spaces,
 \[
M(E) = \bigoplus_{n \in \mathbb Z_{\geq 0}} M_{ h_M -n}(E)
\]
for some $h_M \in \mathbb C$. The top graded piece  $M_{h_M}(E)$ is a $\zhu(V)$-module and isomorphic to $E$. 

The contragredient dual of $M(E)$ is denoted by $M'(E)$ and as a graded vector space it is the graded dual, that is
\[
M'(E) = \bigoplus_{n \in  \mathbb Z_{\geq 0}} \text{Hom}_{\mathbb C}(M_{ h_M -n}(E), \mathbb C).
\]
Assume that $\zhu(V)$ is commutative. In that case $\zhu(V)$ is  generated by the zero modes of the strong generators, \emph{i.e.}, it is isomorphic to a quotient of  $\mathbb C[W^1_0, \ldots, W^\ell_0]$. 
Consider a one-dimensional module $E$. Then the action of $\zhu(V)$ is given by a character $\chi : \zhu(V) \rightarrow \mathbb C$. Let $\lambda = (\lambda_1, \ldots, \lambda_\ell)$ and $\lambda_i = \chi(W^i_0)$. We write $M^\lambda$ for $M(E)$ and say that $M^\lambda$ is a highest-weight module of highest weight $\lambda$. The conformal weight of the top level is then denoted by $h_\lambda$. 

Fix a vector  $\ket{\lambda} \in  M^\lambda_{ h_\lambda}$ and $\bra{\lambda} \in \text{Hom}_{\mathbb C}(M^\lambda_{ h_\lambda }, \mathbb C)$ dual to $\ket{\lambda}$. We write this as $\braket{\lambda | \lambda} =1$. Denote the action of $x_n$ on $\bra{\lambda}$ by $\bra{\lambda x_n}$ and that of $x_n $ on  $\ket{\lambda}$ by $\ket{x_n \lambda}$. 
This determines a unique bilinear pairing $ \braket{\ \cdot\ |\  \cdot\  }: M'^\lambda \times M^\lambda \rightarrow \mathbb C$, that is specified by 
$\braket{\lambda x_{n}| \lambda} = \braket{\lambda | \iota(x_n) \lambda}$ 
for any $n \in \mathbb{Z}$ and any $x_n \in \cA$, with an anti-algebra involution $\iota: \cA \rightarrow \cA$ that is described (for example) in  \cite[Proposition 3.9.1]{Ara}, especially  $\iota^2(x_n) =x_n $ for any element $x_n \in \cA$. Moreover, $\iota$ maps the $H$-eigenspace of conformal weight $\Delta$ to the eigenspace of conformal weight $-\Delta$. 
The pairing $  \braket{\ \cdot\ |\  \cdot\  }  $ is non-degenerate if and only if $M^\lambda$ is simple.  This pairing was described in Section~\ref{Sec:instanton} in the case of principal $ \mathcal W $-algebras of simply-laced Lie algebras.

\begin{rem}\label{rem:liftedVerma}
We can go through the exact same construction by allowing weight labels to depend on $\hbar$: for example, we can let them be rational functions in $\hbar^{\frac{1}{2}}$.  In this case we set $\lambda_i =\chi(\mathsf{W}_0^i)$ and denote the corresponding highest-weight Verma module by $\mathsf{M}^\lambda$. We then require that $\lambda = (\lambda_1,\, \ldots\,, \lambda_\ell)$ is independent of $\hbar$ (which means that the weight labels $\lambda_i$ for the standard Verma module $M^\lambda$ become monomials in $\hbar^{-\frac{1}{2}}$ of degree equal to the conformal weight of $W^i$).
  This means that $\mathsf{M}^\lambda$ is a module of $V \otimes_{\mathbb C}\mathbb C[\hbar^{-\frac{1}{2}}]$, and hence, also a module of $\mathsf{V} = V \otimes_{\mathbb C}\mathbb C(\hbar^{\frac{1}{2}})$.
\end{rem}

\subsection{Whittaker vectors}
\label{SecWhi1}

We now construct Whittaker vectors in the setup described in Section \ref{SecBackground}. More precisely, we show the existence and uniqueness of certain Whittaker vectors, under a given set of assumptions on the vertex algebra. 
Our assumptions are:
\begin{ass}\label{assumption} $V$ is a vertex algebra satisfying:
\begin{enumerate} 
\item $V = \bigoplus\limits_{n \in \mathbb{Z}_{\geq 0}} V_{-n}$ is  graded by a Hamiltonian $H$ with $V_0 = \mathbb C.\ket{0}$. 
\item The set $U = \{  W^1, \ldots, W^{\ell}\}$ is an ordered finite minimal strong generating set of $V$. The $W^i$ are $H$-eigenvectors and the ordering is such that the conformal weight $\Delta_i$ of $W^i$ satisfies $\Delta_i \geq \Delta_j > 0$ for $i>j$. 
\item There is an involution $\rho$ that is compatible with Li's standard filtration with respect to $U$ (see Definition \ref{d:invol}).
\item $\zhu(V)$ is commutative. 
\item The modes of the corresponding fields are defined as 
\[
W^j(z) := \sum_{n \in \mathbb Z} W^j_n z^{-n-\Delta_j}, \]
with the rescaled fields $\mathsf{W}^j(z)$ and modes $\mathsf{W}^j_n \in \cA^{\hbar}$.
 Assume that $W^\ell_{-1}$ is not nilpotent in $\cA$, that is $(W^\ell_{-1})^d \neq 0$ for any $d \in \mathbb Z_{>0}$. 
\end{enumerate}
\end{ass}
\noindent Throughout this section, we assume that the vertex algebra $ V $ satisfies the above assumptions.

We let $\chi$ be a linear form on $\text{span}_{\mathbb C}(\mathsf{W}^1_1, \ldots, \mathsf{W}^\ell_1)$ and set $\chi_i := \chi(\mathsf{W}^i_1)$.

\begin{defin}\label{d:whittaker}
	A \emph{Whittaker module} associated to $\chi$ is a $\mathsf{V}$-module $M$ generated by a vector $\ket{v}$ (the \emph{Whittaker vector}) satisfying
\[
\forall i \in [\ell],\,\,\forall m > 0, \qquad \mathsf{W}^i_m \ket{v} =\chi_i  \delta_{m, 1}\ket{v}.
\]
\end{defin}

Let $i_0 \in [\ell]$ be the minimal index with $\chi_{i_0} \neq 0$. If we change the $(\mathsf{W}^j)_{j > i_0}$ by adding appropriate multiples of iterated derivatives of $\mathsf{W}^{i_0}$ the Whittaker condition reduces to 
\[
\forall j \in [\ell],\,\,\forall m > 0,\qquad \mathsf{W}^j_m \ket{v} = \chi_{i_0} \delta_{m,1}\delta_{j, i_0}\ket{v}.
\]
Furthermore, by rescaling of $\mathsf{W}^{i_0}$ we can assume without loss of generality that $\chi_{i_0} =1$. 

We now show existence and uniqueness of Whittaker vectors with $i_0 = \ell$ under Assumption \ref{assumption}.

Let $\mathsf{M}^\lambda$ be a simple Verma module of $\mathsf{V}$ of highest weight $\lambda$. For any $d \in \mathbb Z_{\geq 0}$ consider the finite-dimensional subspace
$\mathsf{M}^\lambda_{h_\lambda -d }$ of $ \mathsf{M}^\lambda $. Let $S_d$ be a set of lexicographically ordered monomials in the negative modes such that $\left\{\, \mathsf{x}\ket{\lambda} \,\,|\,\, \mathsf{x} \in S_d\right\}$ is a basis of $\mathsf{M}^\lambda_{h_\lambda -d }$. Moreover, we require that $(\mathsf{W}^{\ell}_{-1})^d$ is in $S_d$. 
Note that $(\mathsf{W}^{\ell}_{-1})^d\ket{\lambda}$ is in $\mathcal F_{\ell d}\mathsf{M}^\lambda$ and since the conformal weight of $\mathsf{W}^\ell$ is maximal among the strong generators in $U$, we have
$$
\mathcal F_{\ell d+1}\mathsf{M}^\lambda \cap \mathsf{M}^\lambda_{h_\lambda -d } = \{0\}.
$$

As $\mathsf{M}^\lambda$ is simple, the bilinear pairing $ \braket{\ \cdot\ |\  \cdot\  }$ is non-degenerate. Since $\iota$ maps the $H$-eigenspace of conformal weight $\Delta$ to the eigenspace of conformal weight $-\Delta$, it restricts non-degenerately to the conformal weight subspace $\mathsf{M}'^{\lambda}_{h_\lambda -d } \times \mathsf{M}^\lambda_{h_\lambda -d }$.
Thus there exists a vector $\ket{w^d}$  of conformal weight $d + h_\lambda  $, uniquely specified by $\braket{\lambda | \iota(\mathsf{x}) w^d} =1$ for $\mathsf{x}=   (\mathsf{W}^{\ell}_{-1})^d$ and $\braket{\lambda |\iota(\mathsf{x}) w^d} =0$ for $\mathsf{x}$  in $S_d \setminus (\mathsf{W}^{\ell}_{-1})^d $.  By convention, we set $\ket{w^d} = 0$ for $d < 0$.
 
 Consider  a not necessarily lexicographically ordered vector $\mathsf{W}^{i_r}_{-m_r} \cdots\, \mathsf{W}^{i_1}_{-m_1} \ket{\lambda}$ in $\mathcal F_p M^\lambda$ and let $\sigma$ be the permutation on these modes, such that 
 $\sigma(\mathsf{W}^{i_r}_{-m_r}) \,\cdots\, \sigma(\mathsf{W}^{i_1}_{-m_1}) \ket{\lambda}$
 is lexicographically ordered. Then the difference 
 \[
 \mathsf{W}^{i_r}_{-m_r} \cdots\, \mathsf{W}^{i_1}_{-m_1} \ket{\lambda} - \sigma(\mathsf{W}^{i_r}_{-m_r}) \,\cdots\, \sigma(\mathsf{W}^{i_1}_{-m_1}) \ket{\lambda}
 \]
 is in $\mathcal F_{p-1} \mathsf{M}^\lambda$ since $\text{gr}^{\mathcal F} \mathsf{V}$ inherits commutativity from  $\text{gr}^{\mathcal F} V$.  We can deduce the following result.
  \begin{lem}\label{lemma2.2}
 Let $\mathsf{x} =\mathsf{W}^{i_r}_{-n_r} \cdots\, \mathsf{W}^{i_1}_{-n_1}$ not necessarily lexicographically ordered, such that $\mathsf{x}\ket{\lambda} \neq (\mathsf{W}^{\ell}_{-1})^d \ket{\lambda}$. Then 
 $\braket{\lambda|\iota(\mathsf{x}) w^d}  =0$. 
 \end{lem}
 \begin{proof}
 Since the pairing preserves conformal weight, the conformal weight of $\mathsf{x}$ needs to be $d$. But then the filtration degree of $\mathsf{x}$ is at most $d\ell$ and so the 
 difference 
  $$\mathsf{W}^{i_r}_{-n_r} \cdots\, \mathsf{W}^{i_1}_{-n_1}\ket{\lambda} - \sigma(\mathsf{W}^{i_r}_{-n_r})\,\cdots\, \sigma(\mathsf{W}^{i_1}_{-n_1})\ket{\lambda}$$
  is in $\mathcal F_{\ell d-1} \mathsf{M}^\lambda$. Thus the expansion of $\mathsf{W}^{i_r}_{-n_r} \cdots\, \mathsf{W}^{i_1}_{-n_1}\ket{\lambda}$ in terms of basis vectors does not contain 
$(\mathsf{W}^{\ell}_{-1})^d \ket{\lambda}$. Hence,  $\braket{\lambda|\iota(\mathsf{x}) w^d}  =0$,  which proves the lemma.
 \end{proof}
 \begin{cor} \label{cor2.3}
 For any $m >0$, $i \in [\ell]$ and $d \geq 0$, we have  $\mathsf{W}^i_m \ket{w^d}  = \delta_{i,\ell}\delta_{m, 1} \ket{w^{d-1}}$.
 \end{cor}
\begin{proof}
 Assume that $ \mathsf{W}^i_m \ket{w^d} \neq 0$; then by the non-degeneracy of the bilinear form, there exists $\mathsf{y} =\mathsf{W}^{i_r}_{-n_r} \cdots\, \mathsf{W}^{i_1}_{-n_1}$, such that
 $\braket{\lambda|\iota(\mathsf{y})\mathsf{W}^i_m {w^d}} \neq 0$. Let us introduce the notation $ \mathsf{x}' $, such that   $\iota(\mathsf{y})  \mathsf{W}^i_m = \iota(\mathsf{x}')$, which implies that $\mathsf{x}' = \mathsf{W}^i_{-m} \mathsf{y}$.  Applying Lemma~\ref{lemma2.2} to $ x' $, we see that $ i = \ell $ and $ m = 1 $. The identification of the vector $ \mathsf{W}^\ell_{1} \ket{w^d} $ with $ \ket{w^{d-1}} $ follows from the  uniqueness of the $\ket{w^d}$ for all $ d $.
\end{proof}

\begin{thm}\label{whittaker} Assume that $V$ satisfies Assumption \ref{assumption} and that the  Verma module $\mathsf{M}^\lambda$ of  $\mathsf{V}$ is simple. We define the vector $ \ket{w} $ in  $ \mathsf{M}^\lambda[\![\Lambda]\!] $ as 
	\begin{equation*}
	\ket{w} := \sum_{d \geq 0} \Lambda^d \ket{w^d}.
	\end{equation*}
	Then,  this vector $ \ket{w} $ is the unique vector in $  \mathsf{M}^\lambda[\![\Lambda]\!] $  satisfying
	\begin{equation}\label{eq:ketw}
\forall i \in [\ell],\quad \forall m > 0,\qquad \mathsf{W}^i_m \ket{w} = \Lambda\delta_{i, \ell}\delta_{m,1}  \ket{w}.
	\end{equation}
\end{thm}
\begin{proof}
The property~\eqref{eq:ketw} follows directly from Corollary~\ref{cor2.3}. As for uniqueness, the projection of any  vector satisfying equation~\eqref{eq:ketw} to the conformal weight subspace $\mathsf{M}_{h_\lambda-d}$ for $d \geq 0$ must satisfy the properties that define  $\Lambda^d \ket{w^d}$ uniquely.
\end{proof}

Theorem \ref{whittaker} establishes existence and uniqueness of a Whittaker vector $\ket{w}$ annihilated by all positive modes except $\mathsf{W}^\ell_1$, which acts as a constant, for vertex algebras $V$ satisfying Assumption \ref{assumption} and such that $\mathsf{M}^\lambda$ is simple as a $\mathsf{V}$-module.

\subsection{Left ideals}
\label{SecWhi2}

In this section we establish an important property that will be needed to realize Whittaker vectors as Airy structures. We introduce the following notation. Let $I_a$ be a multi-index set. We write:
\begin{itemize}
	\item $S_a = \{\mathsf{W}^i_m~|~(i,m) \in I_a \} \subset \mathcal{A}^{\hbar}$ for the corresponding subset of rescaled modes of the strong generators of the vertex algebra;
	\item $\mathcal{A}^{\hbar} \cdot S_a$ for the left $\mathcal{A}^{\hbar}$-ideal generated by the modes in $S_a$.
\end{itemize}

Consider the index set $I_+ = \{(i,m) ~|~i \in [\ell], m \geq 1\}$. Then $S_{+}$ is the subset of positive modes. Because of the existence of the involution $\rho$, we know that $[\cA^{\hbar}, \cA^{\hbar}] \subseteq \hbar \cA^{\hbar}$. Furthermore, because the vertex algebra is strongly generated, 
 we know that the left ideal $\mathcal{A}^{\hbar} \cdot S_{+}$ satisfies the following property (see e.g.\footnote{This reference shows that the left ideal generated by the set of non-negative modes satisfies this property, but the same argument holds for the left ideal generated by positive modes. \label{f:subalgebra}} \cite[Proposition 3.14]{BBCCN18}):
\begin{equation}
	\label{Asplus} [\mathcal{A}^{\hbar} \cdot S_{+}, \mathcal{A}^{\hbar} \cdot S_{+}] \subseteq  \hbar \mathcal{A}^{\hbar}\cdot S_{+}.
\end{equation}
In other words, the commutator of any two positive modes in $S_+$ is in $\hbar \mathcal{A}^{\hbar}\cdot S_{+}$.

What we prove below is a stronger statement: assuming that a Whittaker vector exists, we show that the commutator of any two positive modes is in fact in ($\hbar$ times) the left ideal generated \emph{only by the positive modes that annihilate the Whittaker vector}. This stronger requirement will be needed to construct Whittaker vectors as Airy structures. 

To prove this property however we need to strengthen the assumptions of the last section: we require that $V$ is freely generated. In this case the last assumption in Assumption \ref{assumption} is automatically satisfied. 
\begin{ass}\label{assumption2} $V$ is a vertex algebra satisfying:
\begin{enumerate} 
\item $V = \bigoplus\limits_{n \in \mathbb{Z}_{\geq 0}} V_{-n}$ is  graded by a Hamiltonian $H$ with $V_0 = \mathbb C.\ket{0}$. 
\item The set $U = \{  W^1, \ldots, W^{\ell}\}$ is an ordered finite free strong generating set of $V$. The $W^i$ are $H$-eigenvectors of conformal weight $\Delta_i>0$.
\item There is an involution $\rho$ that is compatible with Li's standard filtration with respect to $U$ (see Definition \ref{d:invol}).
\item $\zhu(V)$ is commutative. 
\end{enumerate}
\end{ass}

Recall that $\mathsf{M}^\lambda$ is the Verma module of $\mathsf{V}$ generated by a highest-weight vector  $\ket{\lambda}$ of highest weight $\lambda = (\lambda_1, \ldots, \lambda_\ell)$. Then $\mathsf{W}^i_0  \ket{\lambda} = \lambda_i  \ket{\lambda}$ for any $i \in [\ell]$. Fix a $t \in [\ell]$, and let $I_{t} = \{ (i,m)~|~i\in\{t, \ldots, \ell\},\,\,m = 1\}$. Then $S_{t} = \{\mathsf{W}^t_1, \ldots, \mathsf{W}^\ell_1 \}$. Let $I_{+ \setminus t} = I_+ \setminus I_t$, and define $S_{+ \setminus t}$ and $\mathcal{A}^{\hbar} \cdot S_{+ \setminus t} $ as usual. Then $S_{+ \setminus t}$ is the set of all positive modes except $\mathsf{W}^t_1, \ldots, \mathsf{W}^\ell_1$.

\begin{thm}\label{Lem24} Let $V$ be such that Assumption \ref{assumption2} holds. Let $\ket{w_t^\lambda}$ be  a putative vector in $\mathsf{M}^\lambda[\![\Lambda_t, \ldots, \Lambda_{\ell}]\!]$ satisfying 
	\begin{equation*}
	\mathsf{W}^i_m \ket{w_t^\lambda} = \begin{cases}
	\Lambda_i \ket{w_t^\lambda} & \text{if $\mathsf{W}^i_m \in S_t$,}\\
	0 & \text{if $\mathsf{W}^i_m \in S_{+ \setminus t}$,}
	\end{cases}
	\end{equation*}
	and such that the constant term of $\ket{w_t^\lambda}$ is the highest-weight vector $\ket{\lambda}$ of $\mathsf{M}^\lambda$. If such a $\ket{w_t^\lambda}$ exists for generic $\lambda$, then for all $\mathsf{W}^i_m, \mathsf{W}^j_n \in S_+$,
$$
[\mathsf{W}^i_m, \mathsf{W}^j_n] \in \hbar \mathcal A^{\hbar} \cdot S_{+\setminus t}.
$$
\end{thm}
\begin{proof}

Let $R$ be the set of all $\lambda \in \mathbb{C}^\ell$ such that the Whittaker vector $\ket{w^\lambda_t}$ exists.
As before we fix an ordering for monomials in $\cA^{\hbar}$ by saying that  
\[
\mathsf{W}^{i_1}_{n_1} \cdots\, \mathsf{W}^{i_r}_{n_r}
\]
is ordered if $n_j \geq n_{j-1}$, and $i_j \geq i_{j-1}$ when $n_j=n_{j-1}$. We say that a polynomial is ordered if each of its monomials is ordered. Moreover we call a monomial good if the rightmost term $\mathsf{W}^{i_r}_{n_r}$ satisfies $n_r \geq 1$,  and $i_r < t$ when $n_r =1$. In other words, a monomial is good if it is in $\mathcal{A}^{\hbar} \cdot S_{+\setminus t}$.

By Assumption \ref{assumption2} and Lemma \ref{l:oneh}, we know that $[\mathsf{W}^i_m, \mathsf{W}^j_n] \in \hbar \mathcal A^{\hbar} \cdot S_{+}$.
We want to show that all the monomials on the right-hand-side are good for any $\mathsf{W}^i_r, \mathsf{W}^j_s \in S_+$. By the Whittaker property, we have $[\mathsf{W}^i_r, \mathsf{W}^j_s]\ket{w_t^\lambda} =0$ for any $\lambda \in R$. We prove the lemma by showing that an ordered polynomial $\mathsf{A} \in \mathcal{A}^{\hbar}$ satisfying $\mathsf{A} \ket{w^\lambda_t} =0$ for $\lambda \in R$ must be good.

Consider $\mathsf{A} \in \mathcal{A}^{\hbar}$, such that $\mathsf{A}\ket{w^\lambda_t} =0$ for any $\lambda \in R$, and assume that  
\[
\mathsf{A} = \mathsf{B} \cdot p(\mathsf{W}^1_0, \ldots, \mathsf{W}^{\ell}_0)\cdot q(\mathsf{W}^t_1, \ldots, \mathsf{W}^\ell_1) ,
\]
where $\mathsf{B} \in \mathcal{A}^{\hbar}$ is an ordered monomial having a negative mode as rightmost factor, so that $\mathsf{B}$ can be identified with a PBW-basis vector of any Verma module. Here, $p$ (resp. $q$) are polynomials in $\ell$ (resp. $\ell+1-t$) variables. We have:
\begin{equation}
\begin{split}
 0 &= \mathsf{A} \ket{w_{t}^{\lambda}} \\
 &= \mathsf{B}\cdot  p(\mathsf{W}^1_0, \ldots, \mathsf{W}^{\ell}_0)\cdot q(\mathsf{W}^t_1, \ldots, \mathsf{W}^\ell_1) \ket{w_t^\lambda} \\
 &= \mathsf{B}\cdot p(\mathsf{W}^1_0, \ldots, \mathsf{W}^{\ell}_0)\cdot q(T_t, \ldots, T_\ell) \ket{w_t^\lambda}. \\
\end{split}
\end{equation} As the action of the zero modes $ \mathsf{W}^i_0 $ preserves the grading (by $ H $) on $ \mathsf{M}^\lambda $, we can
restrict $\mathsf{A}$ to the constant term of $\ket{w_t^\lambda}$, that is $\ket{\lambda}$, to get the following vanishing:
\begin{equation}
\begin{split}
 0 &= \mathsf{B} \cdot p(\mathsf{W}^1_0, \ldots, \mathsf{W}^{\ell}_0) \cdot q(T_t, \ldots, T_\ell) \ket{\lambda} \\
 &= p(\lambda_1, \ldots, \lambda_\ell) q(T_t, \ldots, T_\ell)\,\mathsf{B} \ket{\lambda}.
 \end{split}
\end{equation}
Since  $\lambda$ is generic, we deduce that $pq=0$ and hence $\mathsf{A}=0$. Now, let
\[
\mathsf{A}= \sum_\alpha \mathsf{B}_\alpha \cdot p_\alpha
(\mathsf{W}^1_0, \ldots, \mathsf{W}^{\ell}_0) \cdot q_\alpha(\mathsf{W}^t_1, \ldots, \mathsf{W}^\ell_1) 
\]
be a finite sum of elements of the previous type, such that the $\mathsf{B}_\alpha \ket{\lambda}$ are linearly independent. By the same argument as above, we get that
$$
 0  = \sum_\alpha p_\alpha (\lambda_1, \ldots, \lambda_\ell) q_\alpha(T_t, \ldots, T_\ell) \mathsf{B}_\alpha \ket{\lambda}.
$$
By linear independence, each summand must vanish and hence for each $\alpha$ the product $p_\alpha q_\alpha$ must vanish, and hence $\mathsf{A}=0$. It follows that any ordered polynomial $\mathsf{A}$ that annihilates all Whittaker vectors is good.
 \end{proof}

 \subsection{Existence of Whittaker vectors generically}

 \label{SecExistence}

 In Theorem \ref{whittaker}, we showed that Whittaker vectors --- annihilated by all positive modes except $\mathsf{W}^\ell_1$ which acts as a constant --- exist if the vertex algebra satisfies Assumption \ref{assumption} and the Verma module is simple. But when does this apply? We would like to justify that for certain vertex algebras such Whittaker vectors exist generically. This can be done for principal $\mathcal{W}$-algebras, as for those, Verma modules are generically simple, and their top level subspace is one-dimensional.

\subsubsection{Affine vertex algebras}
\label{Sec_affine}

Let $\g$ be a simple complex Lie superalgebra and $\langle \cdot,\cdot \rangle\,:\, \g \times \g \rightarrow \mathbb C$ a non-degenerate invariant bilinear form  $\langle \cdot,\cdot \rangle$ on $\g$. 
We normalize $\langle \cdot,\cdot \rangle$ in the usual way, such that long roots have norm $2$. 
Let 
\[
\widehat \g = \g \otimes_{\mathbb C} \mathbb C[t^{\pm 1}] \oplus \mathbb C .\mathsf{K} \oplus \mathbb C. \mathsf{d}
\]
be the affinization of $\g$. Here $\mathsf{K}$ is central and $\mathsf{d}$ is a derivation; its action on modules can be identified with minus the Hamiltonian $\mathsf{d} = -H=-L_0$. For $x \in \g$ write $x_n$ for $x \otimes t^n$. These elements are called positive (resp. negative modes) if $n > 0$ (resp. $n < 0$). Then the commutation relations are
\[
\forall m,n \in \mathbb{Z},\quad \forall x,y \in \g,\qquad [x_m, y_n] = [x, y]_{m+n} + \delta_{m+n, 0}\, m\,\langle x, y \rangle K. 
\]
Let $\mathsf{k} \in \mathbb C$ and  $\widehat \g_{\mathsf{k}}$ be the Lie superalgebra obtained from $\widehat \g$ by setting $K=\mathsf{k}$.
Let $\widehat \g_{\mathsf{k}}^\pm$ the subalgebras generated by the positive (resp. negative) modes. The subalgebra of zero-modes is identified with $\g$ (ignoring the derivation). The universal enveloping algebras of $\widehat \g_{\mathsf{k}}$, $\widehat \g_{\mathsf{k}}^\pm$ and $\widehat \g_{\mathsf{k}}^\pm \oplus \g$ are denoted by $\cA, \cA_\pm, \cA_{\geq 0}, \cA_{\leq 0}$, and we construct the $\hbar$-rescaled versions as usual.

Let $\mathcal B$ be a basis of $\g$. 
For $\mathsf{k} \in \mathbb C$, the universal affine vertex superalgebra $V^{\mathsf{k}}(\g)$ of $\g$ at level $\mathsf{k}$ (associated to $\kappa$) is generated by fields $\{ X^x(z) = \sum\limits_{n\in\mathbb{Z}} x_n z^{-n-1}\,\, |\,\, x \in \g\}$ with OPE
\[
	X^x(z) X^y(w) \sim \frac{\mathsf{k}\,\langle x, y \rangle}{(z-w)^2} + \frac{X^{[x, y]}(w)}{(z-w)}.
\]
The set $\{X^x \,\,|\,\, x \in \mathcal B \}$ strongly and freely generates  $V^{\mathsf{k}}(\g)$.
As $\widehat \g_{\mathsf{k}}$-modules we have $V^{\mathsf{k}}(\g) \cong \cA \otimes_{\cA_{\geq 0}} \mathbb C. \ket{0}$, \emph{i.e.} it is the Verma module induced from the trivial representation 
$\mathbb C. \ket{0}$ of $\cA_{\geq 0}$. More generally let $\rho$ be a representation of $\g$. Then $\rho$ induces an ${\cA_{\geq 0}}$-module by letting ${\cA_{+}}$ act trivially. The Verma module of $V^{\mathsf{k}}(\g)$ with top level $\rho$ is then 
\[
V^{\mathsf{k}}(\rho) = \cA \otimes_{\cA_{\geq 0}} \rho.
\]
Let $\g = \g_- \oplus \mathfrak h \oplus \g_+$ be as usual a triangular decomposition of $\g$ into a Cartan subalgebra $\mathfrak h$ and positive and negative part. A highest-weight vector $v$ is an eigenvector of  $\mathfrak h$ that is annihilated by $\g_+$. 
If $\rho$ is an irreducible  highest-weight representation of $\g$ of highest weight $\lambda$, we write $V^{\mathsf{k}}(\lambda)$ for $V^{\mathsf{k}}(\rho)$. In this case the conformal weight of the top level is given by 
\[
	h_\lambda = \frac{(\lambda, \lambda +2\rho)}{2(\mathsf{k}+ h^\vee)},
\]
with $\rho$ the Weyl vector and $h^\vee$ the dual Coxeter number of $\g$.

A module is simple if it has no proper graded submodule. 
 It is well-known that $V^{\mathsf{k}}(\lambda)$ is generically simple:
\begin{lem}\label{lem:simpleaffine}
$V^{\mathsf{k}}(\lambda)$ is simple for generic $\lambda \in \mathfrak h^*$ and arbitrary $\mathsf{k}$. 
\end{lem}
\begin{proof}
	Assume that $V^{\mathsf{k}}(\lambda)$ has a proper indecomposable graded submodule $M$. Since $\g_+$ acts nilpotently on every graded subspace of $V^{\mathsf{k}}(\lambda)$, the same needs to be true for $M$:
$M$ needs to be a quotient of $V^{\mathsf{k}}(\mu)$ for some $\mu \in \lambda  + Q$  with $Q$ the root lattice. Thus $h_\lambda + n = h_\mu$ for some positive integer $n$. Let $\mu  = \lambda  +\beta$ for some $\beta \in Q$. Then
\[ 
h_\mu - h_\lambda  =  \frac{(\beta, \beta +2(\lambda+\rho))}{2(\mathsf{k}+h^\vee)}.
\]
It is easy to see that for generic $\lambda$ this cannot be an integer. 
\end{proof}

\subsubsection{Principal \texorpdfstring{$\mathcal{W}$}{W}-algebras}

The general construction of $\mathcal{W}$-algebras is due to \cite{KW} (a condensed review can be found in \cite[Section 3]{CL1}).
Let $(f, h, e)$ be an $\mathfrak{sl}_2 $-triple in $\g$. Here the Cartan subalgebra element $h$ is normalized such that $[h, e] = e$ and $[h, f] = -f$. Then $\g$ is $\frac{1}{2}\mathbb{Z}$-graded by $h$-eigenvalues.  

There exists a free field vertex superalgebra $C(\g, f)$ depending on $\g$ and $f$ and the zero-mode ${\rm d}^f$ of an odd field in $C(\g, f, \mathsf{k}) = V^{\mathsf{k}}(\g) \otimes C(\g, f)$, such that $(C(\g, f, \mathsf{k}), {\rm d}^f)$ is a complex whose homology is a vertex superalgebra, the \emph{$\mathcal{W}$-algebra of $\g$ at level $\mathsf{k}$ associated to $f$}: 
\[
	\mathcal{W}^{\mathsf{k}}(\g, f) := H_*\big(C(\g, f, \mathsf{k}), {\rm d}^f\big).
\]
If $f$ is principal nilpotent, then one usually omits  $f$ in the above notation.

The complex is $\mathbb Z$-graded and the degree is called the \emph{ghost number}. The main structural theorem is  \cite[Theorem 4.1]{KW}. It tells us that the all homologies vanish except in degree zero. Moreover, it states that as a graded vector space, $\mathcal{W}^{\mathsf{k}}(\g, f) \cong V(\g^f)$, and a set of strong and free generators is associated to a homogeneous basis of $\g^f$. (As usual we denote  the kernel of $f$ on $\g$  by $\g^f$. It is the space of lowest weight vectors for the $\mathfrak{sl}_2$-triple). The conformal weight of a generator corresponding to a homogeneous element of $h$-eigenvalue $-n$ is $1-n$.
The homology 
\[
M^{\mathsf{k}}(\lambda,  f) = H_*\big(C(\g, f) \otimes V^{\mathsf{k}}(\lambda), {\rm d}^f\big)
\]
is a module for $\mathcal{W}^{\mathsf{k}}(\g, f)$, and it is in fact a Verma module. 

The following is believed to be true, but a proof does not seem to exist in the literature: 
\begin{conj}\label{conj:simple}
$M^{\mathsf{k}}(\lambda, f)$ is generically simple. 
\end{conj}
If $f$ is principal nilpotent and $\g$ is a simple Lie algebra, then Conjecture \ref{conj:simple} is well-known to be true. Let us denote $M^{\mathsf{k}}(\lambda, f)$ just by $M^{\mathsf{k}}(\lambda)$ in the principal nilpotent case.
By \cite[Lemma 7.5.1]{Ara} the minus reduction of a Verma module of $V^{\mathsf{k}}(\g)$ is a Verma module for $\mathcal{W}^{\mathsf{k}}(\g)$ and by \cite[Theorem 7.6.3]{Ara} the minus reduction of the simple quotient of a Verma module of antidominant (not dominant) highest weight is simple. The following result is a direct consequence of Lemma \ref{lem:simpleaffine}.
\begin{cor}[\cite{Ara}]
For generic $\lambda$ a Verma module $M^{\mathsf{k}}(\lambda)$ of a principal $\mathcal{W}^{\mathsf{k}}(\g)$-algebra 
 is  simple. 
\end{cor}
Let $\{W^1, \ldots, W^\ell \}$ be a minimal strong generating set of a principal $\mathcal{W}^{\mathsf{k}}(\g)$-algebra and $W^\ell$ the generator of maximal conformal weight. Denote by $\mathsf{W}^i$ the corresponding rescaled generators, with modes $\mathsf{W}^i_m \in \mathcal{A}^{\hbar}$. As usual, we denote by $\mathsf{M}^{\mathsf{k}}(\lambda)$ the Verma module constructed from rescaled modes. It is generically simple as well;  Remark \ref{rem:liftedVerma} says that the weight labels $\lambda_i$ for the standard Verma module $M^{\mathsf{k}}(\lambda)$ become polynomials in $\hbar^{-\frac{1}{2}}$, and hence $\mathsf{M}^{\mathsf{k}}(\lambda)$ is a module over both $\mathcal{W}^{\mathsf{k}}(\g) \otimes_{\mathbb C} \mathbb C[\hbar^{-\frac{1}{2}}]$ and $\mathcal{W}^{\mathsf{k}}(\g) \otimes_{\mathbb C} \mathbb C(\hbar^{-\frac{1}{2}})$. The evaluation $\hbar^{-\frac{1}{2}} \mapsto z \in \mathbb C$ thus gives a $\mathcal{W}^{\mathsf{k}}(\g)$ Verma module that must be generically simple. Thus $\mathsf{M}^{\mathsf{k}}(\lambda)$ cannot have a proper submodule.
By Theorems \ref{whittaker} and  \ref{Lem24}, it then follows that:

\begin{prop}
\label{thecoth}For generic $\lambda$ a Verma module $\mathsf{M}^{\mathsf{k}}(\lambda)$ of a principal $\mathcal{W}^{\mathsf{k}}(\g)$-algebra admits a unique Whittaker vector 
 $\ket{w} \in \mathsf{M}^{\mathsf{k}}(\lambda)[\![\Lambda]\!]$ satisfying
\begin{equation}
\label{giufngizgn}\forall i \in [\ell],\quad \forall m > 0,\qquad \mathsf{W}^i_m \ket{w} =\Lambda \delta_{i, \ell}\delta_{m,1} \ket{w}.
\end{equation}
Moreover the condition on the positive modes stated in Theorem \ref{Lem24} holds, that is, for any $\mathsf{W}^i_m, \mathsf{W}^j_n \in S_+$,
	\begin{equation}\label{eq:condpost}
		[\mathsf{W}^i_m, \mathsf{W}^j_n] \in \hbar \mathcal{A}^{\hbar} \cdot S_{+ \setminus \ell}.
\end{equation}
\end{prop}

	Proposition \ref{thecoth} holds for any simple Lie algebra $\mathfrak{g}$. We will use this result to construct Whittaker vectors explicitly using Airy structures for the case where $\mathfrak{g}$ is of type A or D. By the discussion in Section \ref{SecInstanton}, these Whittaker vectors correspond to Gaiotto vectors for the corresponding instanton moduli spaces.

	However, we are also interested in the Gaiotto vectors for instantons of type B and C. From physics, those are not expected to be given by Whittaker vectors for $\mathcal{W}^{\mathsf{k}}(\mathfrak{g})$-modules with $\mathfrak{g}$ of type B or C; rather, they are expected to correspond to Whittaker vectors for certain twisted $\mathcal{W}^{\mathsf{k}}(\mathfrak{g})$-modules with $\mathfrak{g}$ again of type A and D respectively (see \cite{KMST} for instance). We expect the arguments of Sections~\ref{SecWhi1} and \ref{SecWhi2} to extend to twisted modules, which would allow us to prove an analogue of Proposition~\ref{thecoth} about existence of Whittaker vectors, and in particular of Theorem~\ref{Lem24} about the condition on the positive modes \eqref{eq:condpost}, for these cases as well. However, we do not have a proof at this stage. As a result, the necessary condition on the positive modes \eqref{eq:condpost} for the construction of Whittaker vectors for instantons of type B and C using Airy structures are stated as Conjectures \ref{theconjB} and \ref{theconjC} later in the text.

\subsubsection{The self-dual level of principal $\mathcal W$-algebras}

\label{ss:selfdual}

All of our results hold for principal $\mathcal{W}^{\mathsf{k}}(\mathfrak{g})$ at arbitrary level. However, there is one choice of level which is particularly nice, where the construction of Whittaker vectors using Airy structures simplifies considerably --- the self-dual level.

A second way to realize $\mathcal W$-algebras is as subalgebras of a combination of free field algebras and affine vertex algebras, characterized as the joint kernel of a set of screening charges \cite{Genra1}. The principal $\mathcal W$-algebras are nice from this point of view as they are realized as subalgebras of a Heisenberg vertex algebra of the same rank as the corresponding Lie algebra. 

A third realization of principal $\mathcal W$-algebras of type $ADE$ is as a coset \cite{Arakawa:2018iyk}, namely $\mathcal W^{\mathsf{k}}(\g) \cong \text{Com}(V^\ell(\g), V^{\ell-1}(\g) \otimes L_1(\g))$ where $\mathsf{k}$ and $\ell$ are related by the formula
\[
	\frac{1}{\mathsf{k}+ h^\vee} + \frac{1}{\ell + h^\vee} =1,
\]
with $h^\vee$ the dual Coxeter number of $\g$. This coset statement holds for generic level or equivalently as one parameter vertex algebras. Large level limits, that is $\mathsf{k} \rightarrow \infty$, of such structures exist and are isomorphic to orbifolds \cite{Creutzig:2012sf, Creutzig:2014lsa}. In our specific case the orbifold limit of the coset is
\[
	\mathcal W^{-h^\vee +1}(\g) \cong L_1(\g)^G,
\]
with $G$ the compact Lie group whose Lie algebra is $\g$. But $L_1(\g)$ is nothing else than the lattice vertex algebra of the root lattice $Q$ of $\g$ and the orbifold is in degree zero, \emph{i.e.} it is a subalgebra of the Heisenberg subalgebra of $L_1(\g)$. 
Principal $\mathcal W$-algebras enjoy Feigin--Frenkel duality \cite{FF}, that is in the simply-laced case $\mathcal W^{\mathsf{k}}(\g) \cong \mathcal W^\ell(\g)$ for $(\mathsf{k}+h^\vee)(\ell + h^\vee) = 1$. 
We observe that if $\ell = - h^\vee +1$ then $\ell$ is at the Feigin--Frenkel self-dual level. This level is thus rather special both from the coset point of view of $\mathcal W$-algebras and from the perspective of Feigin--Frenkel duality. 

In particular, the coset point of view tells us that the subgroup of $G$ that restricts to  automorphisms of the Heisenberg subalgebra leaves the $\mathcal W$-algebra invariant. Thus twisted modules for the Heisenberg algebra restrict to untwisted $\mathcal{W}$-modules. This point of view was exploited in \cite{BBCCN18,BKS} to construct higher Airy structures as $\mathcal{W}$-modules, but in this paper such twists will not be needed.

\subsection{Whittaker vectors and Airy structures}

\label{s:WhitAiry}

Proposition \ref{thecoth} establishes existence and uniqueness of Whittaker vectors --- annihilated by all positive modes except $\mathsf{W}^\ell_1$ which acts as a constant --- for arbitrary Verma modules of principal $\mathcal{W}^{\mathsf{k}}(\mathfrak{g})$-algebras. We now outline the main ideas behind our construction of Whittaker vectors using Airy structures, that will be put in practice in the next sections. We also explain how this approach may lead to the construction of new Whittaker vectors and ``Whittaker-like'' vectors for vertex algebras.

\subsubsection{Airy structures and $\mathcal{W}^{\mathsf{k}}(\mathfrak{g})$-modules}

Recall the definition of Airy structures in Definition \ref{d:Airy}; from now on we will only consider Airy structures where the dimension of the vector space $V$ is countably infinite. We are therefore looking for a set of differential operators $\{ \mathsf{H}_k ~|~k \in \mathbb{Z}_{>0}\}$ that satisfy two properties: the degree and subalgebra conditions. It turns out that vertex algebras, in particular $\mathcal{W}^{\mathsf{k}}(\mathfrak{g})$-algebras, provide a natural source of Airy structures, as explained in \cite{BBCCN18}.

Let $V$ be a vertex algebra that satisfies Assumption \ref{assumption}. Let $S \subset \cA^{\hbar}$ be a subset of the modes of the strong generators. The idea is to find a representation for the modes $\mathsf{W}^i_m \in S$ as differential operators that satisfy the degree and subalgebra conditions from Definition \ref{d:Airy}.

In the language of vertex algebras, the subalgebra condition of Definition \ref{d:Airy} can be reformulated as follows:

\begin{defin}\label{d:subVOA}
	Let $S \subset \cA^{\hbar}$ be a subset of modes of the strong generators of a vertex algebra. We say that $S$ satisfies the \emph{subalgebra condition} if $\mathcal{A}^{\hbar} \cdot S$ satisfies the property:
	\[
		[\mathcal{A}^{\hbar} \cdot S, \mathcal{A}^{\hbar} \cdot S] \subseteq \hbar \mathcal{A}^{\hbar} \cdot S.
	\]
	In other words, given any two modes in $S$, their commutator is in $\hbar \mathcal{A}^{\hbar} \cdot S$.
\end{defin}
We already mentioned in \eqref{Asplus} that $S_+$ satisfies the subalgebra condition whenever the vertex algebra is strongly generated and has an involution $\rho$, but we could also consider other general subsets of modes (as was done in \cite{BBCCN18} for instance) satisfying the subalgebra condition.

Now we are looking for a representation of $V$ into an algebra of differential operators, hence having a module $M$ of functions on which these act. 
If $S$ is a set satisfying the subalgebra condition, the restriction of this representation to $\mathsf{W}_m^i \in S$ automatically satisfies the subalgebra condition in Definition~\ref{d:Airy}. It remains to examine whether we can choose such a representation obeying the degree condition, and this must be constructed on a case by case basis for the given vertex algebra and the chosen module $\mathsf{M}$.

When this is realized, we  obtain an Airy structure, and it gives us (Theorem~\ref{t:KS}) a unique partition function $\mathcal{Z}$ annihilated by the differential operators representing $\mathsf{W}_m^i \in S$. Equivalently it describes a state $\ket{v} \in \mathsf{M}$ satisfying 
\[
	\forall \mathsf{W}^i_m \in S, \qquad \mathsf{W}^i_m \ket{v} = 0.
\]
Besides, this state can be constructed recursively by the topological recursion as discussed in Section \ref{SecAiryTR}.

\subsubsection{Airy structures and Whittaker vectors}

Can one obtain Whittaker states in this way? The answer is yes, when the subset of modes has a special property closely connected to the condition on the positive modes stated in Theorem~\ref{Lem24}. The key is the following elementary observation, which we formulate in the countably infinite-dimensional setting to match our purposes, but is of course valid in all generality.

\begin{lem}\label{l:AiryW}
	Let $S=\{ \mathsf{H}_k~|~k \in \mathbb{Z}_{>0}\}$ be an Airy structure. Suppose that there exists a subset $S' \subset S$, with $S'' = S \setminus S'$, such that
	\begin{equation}\label{e:special}
		\forall  \mathsf{H}_i, \mathsf{H}_j \in S, \qquad [\mathsf{H}_i, \mathsf{H}_j] \in \hbar \mathcal{D} \cdot S''.
	\end{equation}

	Then, the set of differential operators $\overline{S} = 
\{\overline{\mathsf{H}}_k ~|~k \in \mathbb{Z}_{>0}\}$ defined by	
	\[
		\overline{\mathsf{H}}_k = \begin{cases}	\mathsf{H}_k -  \hbar^{\alpha_k+1}T_k & \text{if $\mathsf{H}_k \in S'$,}\\
			\mathsf{H}_k & \text{if $\mathsf{H}_k \in S''$,}
		\end{cases}
	\]
	for any constants $T_k \in \mathbb{C}$ and $\alpha_k \in \frac{1}{2} \mathbb{Z}_{\geq 0}$, also form an Airy structure, whose partition function\footnote{To keep light notations, we use $\mathcal{Z}$ to denote the partition function associated to the shifted Airy structure $\overline{\mathsf{H}}_k$, but it should not be confused with the partition function associated to the unshifted Airy structure $\mathsf{H}_k$, which would be annihilated by all $\mathsf{H}_k$.} $\mathcal{Z}$ is annihilated by the $\overline{\mathsf{H}}_k$s, which in terms of the original differential operators becomes
	\[
		\mathsf{H}_k \mathcal{Z} = \begin{cases}
			\hbar^{\alpha_k+1}T_k \mathcal{Z} & \text{if $\mathsf{H}_k \in S'$,}\\
			0 & \text{if $\mathsf{H}_k \in S''$.}
		\end{cases}
	\]
\end{lem}

\begin{proof}
Due to the property that
	\begin{equation}\label{eq:supersmall}
		\forall \mathsf{H}_i, \mathsf{H}_j \in S,\qquad [\mathsf{H}_i, \mathsf{H}_j] \in \hbar \mathcal{D}^{\hbar} \cdot S'' .\,
	\end{equation}
	the commutator of any two differential operators in $S$ is a linear combination of the $\mathsf{H}_k$s not in $S'$, with coefficients in $\mathcal{D}^{\hbar}$. It is then possible to shift the $\mathsf{H}_k \in S'$ by constant terms, as \eqref{eq:supersmall} still holds for the shifted operators. Thus the shifted operators still satisfy the subalgebra condition in Definition \ref{d:Airy}. To make sure that they also satisfy the degree condition, we must make sure that the shifts have degree two or more according to the grading in $\mathcal{D}^{\hbar}$. This is achieved if they come with a factor of $\hbar^{\alpha_k+1}$ for any $\alpha_k \in \frac{1}{2}\mathbb{Z}_{\geq 0}$.
\end{proof}

It now becomes clear how Airy structures can be used to construct Whittaker vectors. We are looking for Airy structure representations for the positive modes of the strong generators of a vertex algebra, such that the special property \eqref{e:special} is satisfied with $S'$ consisting of one-modes. In fact, we can be more general. Let us define the following property for subsets of modes:
\begin{defin}\label{d:extra}
	Let $S \subset \cA^{\hbar}$ be a subset of rescaled modes of the strong generators of a vertex algebra. Let $S' \subset S$, and set $S'' = S \setminus S'$. We say that $S'$ is \emph{extraneous} if
	\begin{equation}\label{eq:extra}
		\forall  \mathsf{W}^i_m, \mathsf{W}^j_n \in S, \qquad[\mathsf{W}^i_m, \mathsf{W}^j_n] \in \hbar \mathcal{A}^{\hbar} \cdot S'' .
	\end{equation}
\end{defin}

\begin{rem}\label{r:extrasub}
We note that if \eqref{eq:extra} is satisfied for some subset $S' \subset S$, then $S$ necessarily satisfies the subalgebra condition Definition \ref{d:subVOA}. Thus \eqref{eq:extra} is a stronger condition, and any set of modes $S$ that contains an extraneous subset necessarily satisfies the subalgebra condition. \end{rem}

The following result is a direct consequence of Lemma \ref{l:AiryW}.
\begin{cor}
	Suppose that we have constructed an Airy structure representation for a subset of modes $S$ of the strong generators of a vertex algebra, and that $S' \subset S$ is extraneous. Let $S'' = S \setminus S'$. Then there exists a unique vector $\ket{w}$  in $ \mathsf{M} $ such that
	\[
		\mathsf{W}^i_m \ket{w} = \begin{cases}
			 \hbar^{\alpha_{i,m}+1} T^i_{m} \ket{w} & \text{if $\mathsf{W}^i_m \in S'$,}\\
			0 & \text{if $\mathsf{W}^i_m \in S''$,}
		\end{cases}
	\]
	for any constants $T^i_m \in \mathbb{C}$ and $\alpha_{i,m} \in \frac{1}{2} \mathbb{Z}_{\geq 0}$. Moreover, this vector is explicitly constructed recursively as the partition function of the shifted Airy structure. We call such a vector a \emph{Whittaker-like vector} for the vertex algebra.
\end{cor}

\begin{rem}
If $S=S_+$ is the subset of positive modes, and $S' \subset S$ consists only of one-modes, then $\ket{w}$ is a  Whittaker vector for the vertex algebra, in the usual sense (see Definition \ref{d:whittaker}). However, the construction is more general here. First, the extraneous subset $S'$ may contain more than just one-modes. Second, the starting set of modes $S$ may be larger or smaller than the subset $S_+$ of positive modes, as long as it satisfies the subalgebra property. This raises the possibility that Airy structures may be useful in constructing much more general Whittaker-like vectors.
\end{rem}

The strategy to construct Whittaker vectors using Airy structures then becomes clear:
\begin{enumerate}
	\item Find subsets of modes $S$ of the strong generators of a vertex algebra that contain an extraneous subset $S' \subset S$.
	\item Find a differential representation for the modes in $S$ that satisfies the degree condition in Definition \ref{d:Airy}, and hence is an Airy structure.
\end{enumerate}
If these two steps are achieved, the resulting shifted Airy structure uniquely constructs a Whittaker-like vector annihilated by the modes in $S'' = S \setminus S'$, and such that the modes in $S'$ act as constants.

As the first step is purely algebraic, it can be formulated as the following problem in the theory of vertex algebras:

\begin{prob}\label{theprobeopen}
	Consider a vertex algebra satisfying Assumption \ref{assumption}. Classify all subsets of modes $S$ of the strong generators that contain extraneous subsets $S' \subset S$, according to Definition \ref{d:extra}. 
\end{prob}

An example of this was already achieved in the previous section. In Proposition \ref{thecoth}, it was shown that for any principal $\mathcal{W}^{\mathsf{k}}(\mathfrak{g})$-algebra, the subset $S_+$ of positive modes of the strong generators has an extraneous subset $S_\ell \subset S_+$ consisting of the one-mode with maximal conformal weight $S_\ell = \{ \mathsf{W}^\ell_1 \}$. This means that if we can find an Airy structure representation for the positive modes of the strong generators of a principal $\mathcal{W}^{\mathsf{k}}(\mathfrak{g})$-algebra, the partition function of the shifted Airy structure will construct the Whittaker vector $\ket{w}$ from Proposition \ref{thecoth}. We will do exactly that for the simple Lie algebras $\mathfrak{g} = \mathfrak{sl}_r$ (or rather $\mathfrak{gl}_r$) in Section \ref{Sec:AiryA} and $\mathfrak{g}=\mathfrak{so}_{2r}$ in Section \ref{Sec:AiryCD}.

Going back to Problem \ref{theprobeopen}, for $\mathcal{W}^{\mathsf{k}}(\mathfrak{gl}_2)$ and $\mathcal{W}^{\mathsf{k}}(\mathfrak{gl}_3)$ it can be solved explicitly by examining the known commutation relations between the modes of the strong generators. In these two cases we obtain a full classification of the subsets of modes of the strong generators that contain extraneous subsets. We expect that it is possible to construct Airy structure based on each of these subsets, thus proving existence and uniqueness of the corresponding Whittaker-like vectors. Whittaker-like vectors have been investigated  in physics in the context of Argyres--Douglas theories, see e.g. \cite{Gaiotto:2012sf, Gai, Nishinaka:2019nuy, Tanzini,Xie}, and we plan to study whether they can be produced from Airy structures in future work.
%%%%%%%%%%%%%%%%%%%%%%%%%%%%%%%%%%%%%%%%%%%%%%%%%%%%%%%%%%%%%%%

\section{Airy structures for Gaiotto vectors of type A}
\label{Sec:AiryA}

Recall from Section \ref{SecInstanton} that the Gaiotto vector for a simply-laced Lie group $G$ with Lie algebra $ \mathfrak g $, whose norm calculates the instanton partition function for pure $\mathcal{N}=2$ four-dimensional supersymmetric gauge theory, is a Whittaker vector for a Verma module of a principal $\mathcal{W}^{\mathsf{k}}(\mathfrak{g})$-algebra at shifted level $\kappa = -\frac{\epsilon_2}{\epsilon_1}$. More precisely, the Gaiotto vector $\ket{\mathfrak G}$ is the unique Whittaker vector satisfying
\begin{equation}\label{eq:goalwhitt}
	\forall i \in [\ell], \quad \forall m>0, \qquad \mathsf{W}_m^i \ket{\mathfrak G} = \Lambda^{h^\vee} \delta_{i,\ell} \delta_{m,1} \ket{\mathfrak G},
\end{equation}
where $\ell$ is the rank of $\mathfrak{g}$, $h^\vee$ is the dual Coxeter number, $\Lambda$ is an arbitrary constant, and the $\mathsf{W}^i_m$s are a certain representation (to be constructed below) of the rescaled modes (see \eqref{eq:rescaledW} --- those correspond to the modes in $\cA^{\hbar}$ according to Lemma \ref{l:Ahbar}) of the strong generators of the vertex algebra.

In this section we focus on Lie groups $G$ of type A. We follow the strategy outlined in Section \ref{s:WhitAiry} to construct this Whittaker vector using Airy structures:
\begin{enumerate}
	\item From Proposition \ref{thecoth}, we know that the subset $S_+$ of positive modes of the strong generators of the principal $\mathcal{W}$-algebra has an extraneous subset $S' \subset S_+$ consisting of the one-mode of the generator of highest conformal weight: $S' = \{ \mathsf{W}^\ell_1 \}$.
	\item We construct a differential representation for the positive modes in $S_+$ that satisfies the degree condition in Definition \ref{d:Airy}, and hence is an Airy structure.
\end{enumerate}
Then the partition function of the shifted Airy structure (where only the one-mode $\mathsf{W}^\ell_1$ is shifted) is precisely the Whittaker vector in \eqref{eq:goalwhitt}.

The parameter $\hbar$ which was introduced in Section \ref{SecBackground} plays a crucial role in the recursive construction of the partition function. In fact, it is related to the parameters $ \epsilon_1 $ and $ \epsilon_2 $ introduced in Section~\ref{Sec:Uhlenbeck} as follows:
\begin{equation}
\label{epsdef} \epsilon_1 = \hbar^{\frac{1}{2}}\kappa^{-\frac{1}{2}}\quad {\rm and}\quad \epsilon_2 = - \hbar^{\frac{1}{2}}\kappa^{\frac{1}{2}},
\end{equation}
which implies
\[
\kappa = -\frac{\epsilon_2}{\epsilon_1},\quad {\rm and}\quad  \hbar = - \epsilon_1\epsilon_2.
\] For  notational convenience, we also define 
\[
	 \alpha_0 := \hbar^{-\frac{1}{2}} (\epsilon_1+\epsilon_2) = \kappa^{- \frac{1}{2}} - \kappa^{\frac{1}{2}}.
\]
We note that at the self-dual level (see Section \ref{ss:selfdual}) we have   $\kappa = 1$, which implies that $ \alpha_0 = 0 $ and $\epsilon_1 = - \epsilon_2$. The grading with $\deg(\hbar) = 2$ can be induced by setting
\begin{equation}\label{eq:grading2}
\deg(\epsilon_1) = \deg(\epsilon_2) = 1, \qquad \text{so that} \qquad \deg(\kappa)=\deg(\alpha_0) = 0.
\end{equation}

Before we start, we remark that in type A, instead of working with the simple Lie algebra $\mathfrak{sl}_r$, it is easier to use the semi-simple Lie algebra $\mathfrak{gl}_r$. All statements regarding $\mathcal{W}$-algebra actions can easily be reformulated from $\mathfrak{sl}_r$ to $\mathfrak{gl}_r$, or deduced by reduction from $\mathfrak{gl}_r$ to $\mathfrak{sl}_r$. Thus we now focus on the principal $\mathcal{W}^{\mathsf{k}}(\mathfrak{gl}_r)$-algebras.

\subsection{Heisenberg algebra and \texorpdfstring{$\mathcal{W}$}{W}-generators}

Our goal is to construct a differential representation for the positive modes of the strong generators of $\mathcal{W}^{\mathsf{k}}(\mathfrak{gl}_r)$ that takes the form of an Airy structure. To achieve this, we will use the natural realization of $\mathcal{W}^{\mathsf{k}}(\mathfrak{gl}_r)$ as a subalgebra of the Heisenberg VOA $\mathcal{H}(\mathfrak{gl}_r)$. In this section, we describe  briefly this relation and fix our notation.

For $\mathfrak{g} = \mathfrak{gl}_r$, we have $\ell = h^\vee = r$. We equip the  Cartan subalgebra  $\mathfrak{h} = \mathbb{C}^r$ with the orthogonal canonical basis $(\chi^a)_{a = 1}^{r}$, such that
\[
	\forall a,b \in [r], \qquad \langle \chi^a,\chi^{b} \rangle = \delta_{a,b}.
\]  The Heisenberg VOA $\mathcal{H}(\mathfrak{gl}_r)$ is the VOA freely generated by $\chi^a_{-1} \ket{0}$ where  $\ket{0}$ is the vacuum vector. We define the fields $ J^a(z) $, and the corresponding modes $ J^a_m $ as
\[
J^a(z) := Y\big( \chi^a_{-1} \ket{0} , z\big) =: \sum_{m \in \mathbb Z} J^a_m z^{-m-1}, \qquad a \in [r].
\]
We introduce a parameter $\hbar$ as in Section \ref{sec:filtrations}, and define the corresponding rescaled fields $\mathsf{J}^a(z)$ and modes $\mathsf{J}^a_m \in \cA^{\hbar}$. Those satisfy the commutation relations
\[
\forall m,n \in \mathbb{Z}, \quad \forall a,b \in [r], \qquad	[\mathsf{J}^a_{m}, \mathsf{J}^{b}_{n}] = \hbar\, m \, \delta_{a,b} \delta_{m+n, 0} . 
\]

The $ \mathcal W^{\mathsf{k}}(\mathfrak{gl}_r) $-algebra is  strongly freely generated, and a set of  rescaled generators $ \mathsf{W}^i(z) $ (with modes in $\cA^{\hbar}$), where $ i \in [r] $, is obtained from the quantum Miura transform \cite{Arakawa2017}:
\begin{equation} \label{Miura}
  	 \sum_{i=0}^{r}  \mathsf{W}^{i}(z) \big(\hbar^{\frac{1}{2}} \alpha_0\big)^{r - i} \partial^{r -i}_z	 = \big( \hbar^{\frac{1}{2}} \alpha_0 \partial_z + \mathsf{J}^1(z) \big) \big( \hbar^{\frac{1}{2}}\alpha_0 \partial_z+ \mathsf{J}^2(z) \big) \,\cdots\, \big( \hbar^{\frac{1}{2}} \alpha_0 \partial_z + \mathsf{J}^r(z) \big)\,, 
\end{equation}
with the convention that $\mathsf{W}^0(z) = 1$.  In the $\mathfrak{sl}_r$ case, $\mathsf{W}^1(z)$ is identically $ 0 $, but for $\mathfrak{gl}_r$ it is non-zero. In the $\mathfrak{gl}_r$ case the Virasoro field is $\mathsf{W}^2(z) = \mathsf{L}(z)$.
 
More explicitly, the right-hand side of the Miura transform can be written as:
\[
\begin{split}
\,\,\,\,& \sum_{j = 0}^r (\hbar^{\frac{1}{2}}\alpha_0)^{r - j}  \!\!\!\!\!\! \sum_{1 \leq a_1 < \cdots < a_{j} \leq r} \!\!\! \partial_z^{a_1 - 1} \mathsf{J}^{a_1}(z) \partial_z^{a_2 - a_1 - 1} \mathsf{J}^{a_2}(z) \,\cdots\, \partial_z^{a_{j} - a_{j - 1} - 1} \mathsf{J}^{a_{j}}(z) \partial_z^{r - a_{j}} \\
& = \sum_{j = 0}^r  (\hbar^{\frac{1}{2}}\alpha_0)^{r - j} \sum_{\substack{1 \leq a_1 < \cdots < a_{j} \leq r \\ b_1,\ldots,b_{j} \in \mathbb{Z}_{\geq 0}}}N_{\mathbf{a},\mathbf{b}}\,\,\prod_{l = 1}^{j} (\partial_z^{b_j} \mathsf{J}^{a_j}(z))\,\, \partial_z^{r - j - (b_1 + \cdots + b_{j})},
\end{split}
\]
where
$$
N_{\mathbf{a},\mathbf{b}} = \sum_{\substack{c \in \mathcal{T}_{j}^+ \\ \sum_{u = v}^j c_{u,v} = b_v}} \prod_{u = 1}^{j} \frac{(a_u - a_{u - 1} - 1)!}{(a_u - a_{u - 1} - 1 - \sum_{v = u}^j c_{u,v})!\,\prod_{v = u}^{j} c_{u,v}!},
$$
is a weighted sum over the set $\mathcal{T}_{j}^{+}$ of upper triangular matrices with nonnegative entries, where $a_0 = 0$ by convention. This can be computed explicitly, via the generating series
$$
\sum_{b_1,\ldots,b_j \in \mathbb{Z}_{\geq0}} N_{\mathbf{a},\mathbf{b}}\prod_{v = 1}^j x_v^{b_v} = \prod_{u = 1}^j \bigg(1 + \sum_{v \geq u} x_v\bigg)^{a_u - a_{u - 1} - 1} ,
$$
which gives (see also \cite[Proposition 2.1]{BelEyn0})
$$
N_{\mathbf a, \mathbf b} = \prod_{u=1}^j \binom{a_u - \sum_{v=1}^{u-1} b_v -u }{b_u}.
$$
We deduce that for $i \in [r]$, we get the expression
\begin{equation} 
\label{Miuraexplicit} \mathsf{W}^i(z) = \sum_{j = 1}^{i}  \big(\hbar^{\frac{1}{2}}\alpha_0\big)^{i - j} \sum_{\substack{1 \leq a_1 < \cdots < a_j \leq r \\ b_1,\ldots,b_j \in \mathbb{Z}_{\geq0} \\ b_1 + \cdots + b_j = i - j}} N_{\mathbf{a},\mathbf{b}}\,\,\prod_{l = 1}^j (\partial_{z}^{b_l}\mathsf{J}^{a_l}(z)).
\end{equation} 
The  field $ \mathsf{W}^i(z) $ has conformal weight $i$, and its mode expansion is denoted
\begin{equation}\label{e:Miuramodes}
\mathsf{W}^i(z) = \sum_{m \in \mathbb{Z}} \mathsf{W}^i_m z^{-m-i}, \qquad i \in [r].
\end{equation}

At the self-dual level $\kappa=1$ (see Section \ref{ss:selfdual}), the form of the strong generators simplifies dramatically: 
\begin{equation}\label{Wcritlev}
W^i(z) = Y( e_i(\chi^1_{-1},\ldots,\chi^r_{-1}) \ket{0}, z) =  \sum_{1 \leq a_1 < \cdots < a_j \leq r }\,\,\prod_{l = 1}^i J^{a_l}(z), \qquad i \in [r],
\end{equation}
where $ e_i $ denotes the $ i $-th elementary symmetric polynomial. Of course, the same expression holds for the rescaled fields $\mathsf{W}^i(z)$ and $\mathsf{J}^a(z)$.

\subsection{The Airy structure}\label{Sec:AiryA2}

Using the description of the $\mathcal{W}^{\mathsf{k}}(\mathfrak{gl}_r)$-algebra as an (explicit) subalgebra of the Heisenberg algebra $\mathcal{H}(\mathfrak{gl}_r)$ from the previous section, we shall obtain a  representation for the modes of the former as differential operators, using the following representation for the modes of the latter.

 As in Section \ref{SecAiryTR}, we introduce the graded algebra of formal functions on $\mathfrak{h}$
\[
\mathcal{M}^{\hbar} : =\mathbb C \big[\!\big[(x^a_m)_{a \in [r],\,\,m > 0}\big]\!\big][\![\hbar^{1/2}]\!],
\]
with the grading \eqref{eq:grading} and \eqref{eq:grading2}. Let $ Q_1,\ldots,Q_r \in \mathbb{C}$, which have degree $0$ in the algebra. The algebra $\mathcal{M}^{\hbar}$ carries a representation of $\mathcal{H}(\mathfrak{gl}_r)$ by differential operators
\begin{equation}
\label{eq:hrep}
\begin{split}
\mathsf{J}^a_{m} = \left\{
\begin{array}{lr}
\hbar \partial_{x^a_m}& m > 0,\\
-m x^{a}_{-m}& m < 0, \\
Q_a - (\epsilon_1 + \epsilon_2) = Q_a - \hbar^{\frac{1}{2}}\alpha_0 &  m = 0. 
\end{array}
\right.
\end{split}
\end{equation}

Note that $\mathsf{J}_m^a$ for $m \neq 0$ is an operator of degree $1$, but $\mathsf{J}_0^a$ is not homogeneous unless $\kappa = 1$ (equivalently $ \alpha_0 = 0 $) --- in the language of \cite{BBCCN18}, this is because one can think of the degree $0$ term $Q_a$ as a dilaton shift of the zero-modes. We define $\ket{\lambda}$ to be the element $1 \in M$, where we recall that $ \lambda  $ is the highest weight  defined in Theorem~\ref{ThmBFN}.

We now show that this representation of the $\mathcal{W}^{\mathsf{k}}(\mathfrak{gl}_r)$-algebra can be used  to construct an Airy structure.

\begin{thm}
	\label{Whit1prop}Assume that $r \geq 2$,  $Q_1, \ldots, Q_r \in \mathbb{C}$ are  pairwise distinct, and that $T \in \mathbb{C}$. Let $\mathsf{W}^i_m$ be the differential operators associated to the rescaled modes of the strong generators $\mathsf{W}^i(z)$ of the principal $\mathcal{W}^{\mathsf{k}}(\mathfrak{gl}_r)$-algebra, obtained  via equation \eqref{eq:hrep}. Then the family of differential operators 
\[
	\forall  i \in [r],\, m \in \mathbb{Z}_{>0}, \qquad \overline{\mathsf{W}}^i_m := \mathsf{W}^i_m -\hbar^{\frac{r}{2}} T \delta_{i,r} \delta_{m, 1}
	\]
forms an Airy structure.
\end{thm}

Before we prove the theorem, let us state an immediate corollary, which follows directly from the theory of Airy structures discussed in Section \ref{SecAiryTR}. The partition function $\mathcal{Z}$ associated to the Airy structure\footnote{As in Lemma \ref{l:AiryW}, $\mathcal{Z}$ is the partition function associated to the shifted Airy structure $\overline{\mathsf{W}}^i_m$, not the unshifted one $\mathsf{W}^i_m$.} $\overline{\mathsf{W}}^i_m$ (see Theorem \ref{t:KS}) is a solution to the differential constraints
\[
	\mathsf{W}^i_m \mathcal{Z} = \hbar^{\frac{r}{2}} T \delta_{i,r} \delta_{m,1} \mathcal{Z}.
\]
In other words, the partition function $\mathcal{Z}$ can be identified with a Whittaker vector $\ket{w}$ in a Verma module $M^{\mathsf{k}}(\lambda)$ --- whose highest weight $\lambda$ is uniquely determined by the action of the zero-modes in \eqref{eq:hrep} --- satisfying
\[
	\mathsf{W}^i_m \ket{w} = \hbar^{\frac{r}{2}} T \delta_{i,r} \delta_{m,1} \ket{w}.
\]
As the theory of Airy structures guarantees the existence and uniqueness of $\mathcal{Z} = \ket{w}$, it is none other than the Gaiotto vector $ \ket{\mathfrak G} $ of type A from Section \ref{SecInstanton} (see \eqref{eq:gaiotto}), with the identification of parameters
\[
	\hbar^{\frac{r}{2}}T = \Lambda^r.
\]

\begin{cor}\label{cor:AiryG} Under the above identification of parameters, the unique partition function $\mathcal{Z}$ associated to the Airy structure $\overline{\mathsf{W}}^i_m$ coincides with the Gaiotto vector $ \ket{\mathfrak G} $ of type A defined in \eqref{Gaiodef}. \qed
	\end{cor}

\begin{proof}[Proof of Theorem~\ref{Whit1prop}]
	To prove that $\{\overline{\mathsf{W}}^i_m\,\,|\,\, i \in [r],\,\,m \in \mathbb{Z}_{>0}\}$ forms an Airy structure, we can apply Lemma~\ref{l:AiryW}, which requires checking that
	\begin{enumerate}
		\item $\{\mathsf{W}^i_m\,\,|\,\,i \in [r],\,\,m \in \mathbb{Z}_{>0}\}$ forms an Airy structure;
		\item the set $S_+ = \{\mathsf{W}^i_m\,\,|\,\,i \in [r],\,\,m \in \mathbb{Z}_{>0}\}$ of positive modes contains the extraneous subset $S' \subset S$ (see Definition \ref{d:extra}) with $S' = (\mathsf{W}^r_1)$. 
	\end{enumerate}

	Fortunately, Proposition~\ref{thecoth} for the principal $\mathcal{W}^{\mathsf{k}}(\mathfrak{g})$-algebra already gives (2).

	For (1), according to Definition \ref{d:Airy}, we need to show that the degree and subalgebra conditions are satisfied. The subalgebra condition is necessarily satisfied, since the existence of the extraneous subset $S' \subset S_+$ implies it (see Remark \ref{r:extrasub}). 

	All that remains is to verify that the degree condition is satisfied for the $\{\mathsf{W}^i_m\,\,|\,\, i \in [r],\,\,m \in \mathbb{Z}_{>0}\}$, or equivalently for the $\{\overline{\mathsf{W}}^i_m\,\,|\,\,i \in [r],\,\,m \in \mathbb{Z}_{>0}\}$ since the shift has degree $\geq 2$. Our starting point is \eqref{Miuraexplicit}, expanded in $z$ to extract the modes $\mathsf{W}^i_m$ as in \eqref{e:Miuramodes}, with the $\mathsf{J}^a_m$ understood as the differential operators in \eqref{eq:hrep}.

	First, we must check that the $\mathsf{W}^i_m$ have no degree zero term according to the grading \eqref{eq:grading} on the algebra of differential operators. Looking at \eqref{eq:hrep}, the only degree zero terms can come from products of zero-modes $\mathsf{J}_0^a$. It is easy to see, by direct inspection of \eqref{Miuraexplicit} that terms that only involve products of zero-modes $\mathsf{J}_0^a$ never occur for $\mathsf{W}^i_m$ with $m > 0$.

	Next, we need to check that the degree one condition is satisfied. Let $\pi_1$ the operator projecting  differential operators to their degree $1$ component according to the grading \eqref{eq:grading}.

	Let us  first consider the self-dual level $\kappa=1$. The degree $1$ projection of the $\overline{\mathsf{W}}^i_m$ takes the form:
\begin{equation}
	\label{pi1critlev}\pi_1(\overline{\mathsf{W}}^i_m) = \sum_{1 \leq a_1 < \cdots < a_i \leq r} \sum_{j = 1}^i \bigg(\prod_{\substack{1 \leq l \leq i \\ l \neq j}} Q_{a_l}\bigg) \mathsf{J}_m^{a_j}.
\end{equation}
This is better handled in generating series form
\begin{equation}
\label{indin} 
\sum_{i = 1}^r  t^{r - i}\,\pi_1(\overline{\mathsf{W}}^i_m)\, = \sum_{a = 1}^r \bigg(\prod_{b \neq a} (t + Q_b)\bigg) \mathsf{J}^a_m  .
\end{equation}
When $Q_1,\ldots,Q_r$ are  pairwise distinct, we can extract the Heisenberg modes $ \mathsf{J}^a_m $ for each $a \in [r]$ from \eqref{indin} as follows,
\[
	\mathsf{J}^a_m = \Res_{t = -Q_a} \dd t \,\frac{\sum_{i = 1}^r t^{r - i}\,\pi_1(\overline{\mathsf{W}}^i_m)}{\prod_{b = 1}^r (t + Q_b)} .
\]
Evaluating the right-hand side gives
\begin{equation}
\label{deg1W}
\frac{\sum_{i = 1}^r (-Q_a)^{r - i}\,\pi_1(\overline{\mathsf{W}}^i_m)}{\prod_{b \neq a} (Q_b - Q_a)} = \mathsf{J}_m^a.
\end{equation}
Therefore, the operators
\begin{equation}
\label{Ham}
\mathsf{H}_m^a = \frac{\sum_{i = 1}^r (-Q_a)^{r - i} \overline{\mathsf{W}}^i_m}{\prod_{b \neq a} (Q_b - Q_a)},\qquad a \in [r],\quad m \in \mathbb{Z}_{>0},
\end{equation}
give an Airy structure in normal form, and hence $\{\overline{\mathsf{W}}^i_m\,\,|\,\,i \in [r],\,\,m \in \mathbb{Z}_{>0}\}$ is an Airy structure.

For general $\kappa$, the terms in equation~\eqref{Miuraexplicit} when $ i = j $ contribute to the degree $ 1 $ part of the differential operators $\mathsf{W}^i_m $, and match with the ones computed in equation~\eqref{pi1critlev} for $ \kappa = 1 $. For degree reasons, the only possible additional contributions in degree $1$ come from the $ j = i-1 $ terms of equation~\eqref{Miuraexplicit}. These contributions are from  $\pi_0(\mathsf{J}^a(z)) = \frac{Q_a}{z}$, which yields
$$
\pi_1(\mathsf{W}^i(z)) = \pi_1(\mathsf{W}^i(z))|_{\kappa = 1} - \hbar^{\frac{1}{2}}\alpha_0 \sum_{1 \leq a_1 < \cdots < a_{i - 1} \leq r} \bigg(\prod_{l = 1}^{i - 1} Q_{a_l}\bigg) \bigg( \sum_{v = 1}^{i - 1} N_{\mathbf{a},\boldsymbol{\delta}_v}\bigg) z^{-i},
$$
where $\boldsymbol{\delta}_v$ is the sequence of length $i - 1$ with a single non-zero entry equal to $1$ at position $v$. Since $\mathsf{W}^i(z)$ has conformal weight $i$, the term $z^{-i}$ appears in the zero-mode, and therefore the degree $1$ terms of positive modes agree with the $\kappa = 1$ case.
\end{proof}

\begin{rem}\label{r:difference}
There are important differences between the Airy structures constructed in Theorem~\ref{Whit1prop} as $\mathcal{W}^{\mathsf{k}}(\mathfrak{gl}_r)$-modules and those constructed in \cite{BBCCN18}:
\begin{enumerate}
\item In contrast to \cite{BBCCN18}, we do not restrict ourselves to the self-dual level, but construct Airy structures for arbitrary levels.
\item The module of the Heisenberg algebra that we start with in the construction above is untwisted, in contrast to \cite{BBCCN18} where the construction started with a module of the Heisenberg algebra twisted by an element of the Weyl group (the case considered here would correspond to the untwisted case  in \cite{BBCCN18}).
\item We consider a smaller subset of modes than considered in \cite{BBCCN18}, namely only positive modes of the strong generators. Because of this, there is an extraneous mode (the one-mode $\mathsf{W}^r_1$), and we can shift it by a constant $-\hbar^{\frac{r}{2}} T $. As a result, the partition functions of our Airy structures in Theorem~\ref{Whit1prop} give rise to Whittaker vectors (as they are annihilated by all positive modes except $\mathsf{W}^r_1$), while in \cite{BBCCN18} the modes are not shifted and hence the partition functions correspond to highest-weight vectors.
\end{enumerate}
\end{rem}

\begin{rem}\label{r:shift}
We note that instead of shifting $\mathsf{W}^r_1$ by $-\hbar^{\frac{r}{2}} T $, we could have shifted it by $\sum_{j \geq 2} \hbar^{\frac{j}{2}} T_j$ for arbitrary $T_2, T_3, \ldots \in \mathbb{C}$. Since these terms all have degree $\geq 2$ according to the grading \eqref{eq:grading}, we would still obtain an Airy structure. However, the recursive structure to compute the partition function $\mathcal{Z}$ is more efficient (more information gained for the same number of recursion steps) if the shift is placed in degree 2, \emph{i.e.} choosing $T_j=0$ for $j \geq 3$. Nevertheless, in gauge theory, one would like to place the shift in degree 0, corresponding to the parameter $\Lambda$ in Theorem~\ref{ThmBFN}. At first sight, this goes out of the framework of Airy structures, and consequently, the $\hbar$-expansion of the partition function (or the $\hbar$-expansion of the Whittaker vector, which we will explore in \eqref{Phidef}) is not computed right-away by a recursion on $2g-2+n$, but rather via Corollary~\ref{corresum}. For this reason, for gauge-theoretic applications in this paper we have to consider $\Lambda$ as a formal variable near $0$.

If we want finite values of $\Lambda$, non-zero contributions of $(g,n) = (0,1)$ and $(0,2)$ make their apparition in the partition function (see Remark~\ref{rem02s}). We expect that these can be absorbed by resummation, which gives rise to the usual topological recursion on the spectral curve $\prod_{a = 1}^r \big(y - \tfrac{Q_a}{x}\big) = (-1)^{r+1}\frac{\Lambda}{x^{r + 1}}$ instead of $\prod_{a = 1}^r \big(y - \tfrac{Q_a}{x}\big) = 0$ appearing in Section~\ref{SelfdualTRspec}. The situation is similar to the case of the 1-hermitian matrix model: its Virasoro constraints do not form Airy structures, but they can be transformed by non-trivial resummations and analytic continuation on the spectral curve of the model, so as to arrive to an Airy structure and thus Chekhov--Eynard--Orantin topological recursion for the correlators, see e.g. \cite{BEO}. We shall return to this in a future article.
\end{rem}

\subsection{Structural properties of the Gaiotto vector}\label{Sec:stG}

We now use the theory of Airy structures to extract some structural properties of the Gaiotto vector, which we will use in the next section in order to compute the instanton partition function.

Let $\mathcal{Z}$ be the partition function of the Airy structure from Theorem~\ref{Whit1prop}, which is identified with the Gaiotto vector of type A. By Section \ref{SecAiryTR}, we know that $\mathcal{Z} = e^F$ with $F \in \hbar^{-1} \mathcal{M}^{\hbar}$ containing only terms of positive degree. $F$ has an explicit expansion:
\begin{equation}\label{eq:Fexp}
	F = \sum_{\substack{g \in \frac{1}{2} \mathbb{Z}_{\geq 0}, n \in \mathbb{Z}_{>0} \\ 2g-2+n >0}} \frac{\hbar^{g-1}}{n!} \sum_{\substack{a_1, \ldots, a_n \in [r] \\ k_1, \ldots, k_n \in \mathbb{Z}_{>0}}} F_{g,n}\big[\begin{smallmatrix}a_1 & \cdots & a_n \\ k_1 & \cdots & k_n \end{smallmatrix}\big] \prod_{i=1}^n x^{a_i}_{k_i} .
	\end{equation}
	We now extract some properties of the coefficients $F_{g,n}$.

\begin{lem}
\label{lem:prophbar}Assume that $r \geq 2$ and $2g - 2 + n > 0$. Then the $ F_{g,n} $ satisfy the following properties 
\begin{enumerate}
	\item The  coefficient  $ F_{g,n}\big[\begin{smallmatrix} a_1 & \cdots & a_n \\ k_1 & \cdots & k_n \end{smallmatrix}\big] $ depends polynomially on $ T $ and $ \alpha_0 $, and  whenever it is non-zero, it factors as $T^{k_1 + \ldots + k_n} P(\alpha_0)$ with $P(\alpha_0)$ a polynomial in $\alpha_0$ of degree at most $2g$. More precisely, 
	\[
		\frac{F_{g,n}\big[\begin{smallmatrix} a_1 & \cdots & a_n \\ k_1 & \cdots & k_n \end{smallmatrix}\big]} { T^{k_1 + \cdots + k_n}} \in   \mathbb{Q}(Q_1,\ldots,Q_r)\big[\alpha_0 \big]
	\]
	is a polynomial in $\alpha_0$ of degree at most $2g$.
	\item If $ (k_1 + \cdots + k_n)r > 2g $, we have $ F_{g,n}\big[\begin{smallmatrix} a_1 & \cdots & a_n \\ k_1 & \cdots & k_n \end{smallmatrix}\big] = 0 $.   In particular, this condition implies that $F_{g,n} = 0$ for $n > \frac{2g}{r}$.
	\item The coefficient $ F_{g,1} \big[\begin{smallmatrix} a \\ 1 \end{smallmatrix}\big]  = \frac{T}{\prod_{b \neq a} (Q_b - Q_a)} \delta_{g, \frac{r}{2}} $.
\end{enumerate}  
\end{lem}

\begin{proof} 
	
	Recall the definition of $\mathsf{H}_m^a$ in \eqref{Ham}, and set 
$$
T_a = \frac{T}{\prod_{b \neq a} (Q_b - Q_a)}.
$$
According to Theorem~\ref{Whit1prop} and more specifically  \eqref{deg1W}, $\mathsf{H}_m^a - \hbar^{\frac{r}{2}}\delta_{m,1}T_a$ indexed by $m > 0$ and $a \in [r]$ forms an Airy structure in normal form: $\mathsf{H}_m^a = \hbar \partial_{x_m^a} + O(2)$. Extracting the coefficient of a monomial $\hbar^g\,x_{k_2}^{a_2}\,\cdots \,x_{k_n}^{a_n}$ in $\mathcal{Z}^{-1} \mathsf{H}_{k_1}^{a_1} \mathcal{Z}$ for $k_1 > 0$ and $a_1 \in [r]$ then gives a formula for $F_{g,n}\big[\begin{smallmatrix} a_1 & \cdots & a_n \\ k_1 & \ldots & k_n \end{smallmatrix}\big]$ in terms of $F_{g',n'}$ with $2g' - 2 + n' < 2g - 2 + n$ and the coefficients of the normally ordered monomials in $\mathsf{J}$ in $\mathsf{H}_{k_1}^{a_1}$.

\emph{$ \bullet \, \, $Proof of conditions (1) and (2) for $ \kappa = 1 $:}
 
For $\kappa = 1$, collecting the coefficient of $z^{-(k_1 + i)}$ in $\mathsf{W}^i(z) = e_i(\mathsf{J}^1(z),\ldots,\mathsf{J}^r(z))$ yields a linear combination over $\mathbb{Q}(Q_1,\ldots,Q_r)$ of terms of the form $ \prod_{s \in U_-} \mathsf{J}^{b_s}_{-m_s} \prod_{s \in U_+} \mathsf{J}^{c_s}_{p_s}$ with  $b_s,c_s \in [r]$, and $m_s,p_s > 0$ such that
\begin{equation}
\label{ksum}
k_1 = \sum_{s \in U_+} p_s - \sum_{s \in U_-} m_s,\qquad |U_+| + |U_-| \leq r.
\end{equation}
The same is true for $\mathsf{H}_{k_1}^{a_1}$, and apart from the degree $1$ term $\mathsf{J}_{k_1}^{a_1}$, we have $|U_-| + |U_+| \geq 2$. Since $k_1 > 0$, \eqref{ksum} imposes $U_+ \neq \emptyset$. Then, $F_{g,n}\big[\begin{smallmatrix} a_1 & \cdots & a_n \\ k_1 & \ldots & k_n \end{smallmatrix}\big] - \delta_{g,\frac{r}{2}}\delta_{n,1}\delta_{k_1,1}T_{a_1}$ is a linear combination over $\mathbb{Q}(Q_1,\ldots,Q_r)$  of terms of the form
\begin{equation}
\label{productF}
\prod_{L \in \mathbf{L}} F_{g_L,|L| + |N_L|}\big[\begin{smallmatrix} \mathbf{c}_L & \mathbf{a}_{N_L}  \\ \mathbf{p}_L & \mathbf{k}_{N_L} \end{smallmatrix}\big],
\end{equation}
where:
\begin{itemize}
	\item[(i)] $\mathbf{L}$ is a partition of $U_+$ and for each $L \in \mathbf{L}$, we have $g_L \in \frac{1}{2}\mathbb{Z}_{\geq 0}$, $\mathbf{p}_L \in \mathbb{Z}_{>0}^L$ and $\mathbf{c}_L \in [r]^L$.
\item[(ii)] $(N_L)_{L \in \mathbf{L}}$ is a sequence of pairwise disjoint, possibly empty subsets of $\{2,\ldots,n\}$. We denote $\mathbf{k}_{N_L} = (k_l)_{l \in N_L}$ and $\mathbf{a}_{N_L} = (a_l)_{l \in N_L}$.
\item[(iii)] For each $L \in \mathbf{L}$, we have $2g_L - 2 + |L| + |N_L| > 0$. 
\item[(iv)] The complement of $\sqcup_{L \in \mathbf{L}} N_L$ is in bijection with $U_-$. 
\item[(v)] $|U_+| + |U_-| \geq 2$.
\item[(vi)] $k_1 = \sum_{s \in U_+} p_s - \sum_{s \in U_-} k_s$. 
\item[(vii)] $|U_+| + \sum_{L \in \mathbf{L}} (g_L - 1)  = g$.
\end{itemize}
Taking into account these constraints, we see that
$$
|U_+| + |U_-| - 1 + \sum_{L \in \mathbf{L}} (2g_L - 2 + |L| + |N_L|) = 2g - 2 + n.
$$
Since the  sum above is non-empty, all the terms occurring in the sum are positive, and $|U_+| + |U_-| - 1 > 0$, it is clear that $ F_{g,n} $ is constructed recursively on $2g - 2 + n > 0$. Condition (vii) ensures that all $ g_L \leq g $ and hence we get that $F_{g,n} = 0$ for $g < \frac{r}{2}$. When $g = \frac{r}{2}$ and $n = 1$ we get
\begin{equation}
\label{F1r3r}F_{\frac{r}{2},1}\big[\begin{smallmatrix} a_1 \\ k_1 \end{smallmatrix}\big] = T_{a_1}\delta_{k_1,1}.
\end{equation}

We will prove the statement of the Lemma by induction on $2g - 2 + n > 0$. When $2g - 2 + n \leq r - 1$, we have just proved that the only non-zero $ F_{g,n} $ is $ F_{\frac{r}{2},1}\big[\begin{smallmatrix} a \\ 1 \end{smallmatrix}\big]  = T_a $.  Condition (1) of the statement of the Lemma is then satisfied (as $\alpha_0=0$ at self-dual level); condition (2) is also satisfied trivially.

Now, we move on to the induction step. Assume that $g \in \frac{1}{2}\mathbb{Z}_{\geq 0}$ and $n \in \mathbb{Z}_{>0}$ such that $2g - 2 + n > r - 1$, and assume that conditions (1) and (2) hold for all $F_{g',n'}$ such that $2g' - 2 + n' < 2g - 2 +n$. Then, the non-zero terms arising from \eqref{productF} have a power of $T$ equal to 
$$
\sum_{s \in U_+} p_s + \sum_{l \notin U_-} k_l = k_1 + \sum_{l \in U_-} k_l + \sum_{l \notin U_-} k_l = \sum_{l = 1}^n k_l,
$$ which proves the condition (1). For every $  L $ in $ \mathbf{L} $, we have the inequality
\begin{equation}
\label{consufnf}
 \sum_{s \in L} p_s + \sum_{s \in N_L} k_s \leq \frac{2g_L}{r}.
\end{equation}
Summing this up over all $ L $ in $ \mathbf L $ and using (vii) yields $\sum_{l = 1}^n k_l \leq \frac{2}{r}(g - |U_+| + |\mathbf{L}|) \leq \frac{2g}{r}$, and this is exactly the condition (2).

\emph{$ \bullet \, \,$Proof of conditions (1) and (2) for arbitrary $ \kappa $:}
 
Now we assume that $\kappa$ is arbitrary. Again, we see that $\mathsf{W}_m^i$ is a linear combination over $\mathbb{Q}(Q_1,\ldots,Q_r)$ of terms of the form $\hbar^{\frac{d}{2}}\alpha_0^d \prod_{s \in U_-} \mathsf{J}^{b_s}_{-m_s} \prod_{s \in U_+} \mathsf{J}^{c_s}_{p_s}$ with only notable difference that $|U_+| + |U_-| \leq r - d$ for some $d \in \{0,\ldots,r -1 \}$. Notice that the upper bound for $|U_+| + |U_-|$ was not used in the previous argument. Hence, repeating the same argument yields that $F_{g,n}\big[\begin{smallmatrix} a_1 & \cdots & a_n \\ k_1 & \cdots & k_n \end{smallmatrix}\big]$ is a  linear combinations over $\mathbb{Q}(Q_1,\ldots,Q_r)$ of terms of the form
\begin{equation}
\label{productF2}
\alpha_0^{d} \prod_{L \in \mathbf{L}} F_{g_L,|L| + |N_L|}\big[\begin{smallmatrix} \mathbf{c}_L & \mathbf{a}_{N_L}  \\ \mathbf{p}_L & \mathbf{k}_{N_L} \end{smallmatrix}\big],
\end{equation}
with the same constraints (i)-(vi), and (vii) modified to 
\begin{itemize}
\item[(vii')] $|U_+| + \sum_{L \in \mathbf{L}} (g_L - 1) = g - d$ and $0 \leq d \leq r - 1$. 
\end{itemize}
Summing up the constraints \eqref{consufnf} over all $ L $ in $ \mathbf L $  gives $\sum_{l = 1}^n k_l \leq \frac{2}{r}(g - d - |S_+| + |\mathbf{L}|) \leq \frac{2g}{r}$, thus proving condition (2) of the lemma for arbitrary $ \kappa $.

The proof of the dependence in $T$ by induction is the same as before. As for the maximal degree in $\alpha_0$, we prove it by induction on $2g-2+n$ as well. The only non-vanishing $F_{g,n}$ with $2g-2+n \leq r-1$ is $F_{\frac{r}{2},1}\big[\begin{smallmatrix} a \\ 1 \end{smallmatrix}\big]  = T_a $, which is certainly a polynomial in $\alpha_0$ of degree $\leq 2g =2 (r/2) = r$ in this case. Now assume that the condition holds for all $F_{g',n'}$ with $2g'-2+n' < 2g-2+n$. $F_{g,n}\big[\begin{smallmatrix} a_1 & \cdots & a_n \\ k_1 & \cdots & k_n \end{smallmatrix}\big]$ is a  linear combination over $\mathbb{Q}(Q_1,\ldots,Q_r)$ of terms of the form \eqref{productF2}. By the induction hypothesis, the maximal degree in $\alpha_0$ is thus
\[
d + \sum_{L \in \mathbf{L}} 2 g_L =  2 g - d  - 2 |U_+|  + 2 |\mathbf{L}| \leq 2g,
\]
where we used (vii') and the fact that $|\mathbf{L}| \leq |U_+|$.

\emph{$ \bullet \, \, $Proof of condition (3):}

 We directly consider the case when $ \kappa  $ is arbitrary. Let us consider the terms \eqref{productF2} that contribute to  $  F_{g,1} \big[\begin{smallmatrix} a \\ 1 \end{smallmatrix}\big] - \delta_{g, \frac{r}{2}}  T_a  $. Using the previous notation, we should consider $ k_1 = 1 $ and $ U_{-} = \emptyset $. Then,  constraint (vi) states  that $ 1 = \sum_{s \in U_+} p_s $, which forces $ |U_+| = 1 $. However, this violates  constraint  (v), and hence there are no terms that contribute to  $  F_{g,1} \big[\begin{smallmatrix} a \\ 1 \end{smallmatrix}\big] - \delta_{g, \frac{r}{2}} T_a  $. This proves condition (3).

\end{proof}

\begin{rem}
\label{refinite}	The proof of Lemma~\eqref{lem:prophbar} is an explicit way to see that the $ F_{g,n} $s of the Airy structure of Theorem~\ref{Whit1prop} are constructed as finite sums, using $F_{g',n'}$s with $ 2g'  -2 + n' \leq 2g -2+n $. This justifies that it is well-defined even if the Airy structure is infinite-dimensional (see Remark \ref{r:infinite} and the discussion in Section 2.1.2 of \cite{BKS}).
\end{rem}

Now, we use the  identification $\Lambda^r = \hbar^{\frac{r}{2}}T$, in order to rewrite $ F  = \log \mathcal Z $ of the Airy structure in terms of $ \Lambda $. This is required for the comparison with the instanton partition function.

\begin{cor}
	\label{corresum}There is a unique sequence of $\Phi_{h,n}\big[\begin{smallmatrix} \mathbf{a} \\ \mathbf{k} \end{smallmatrix}\big] \in \mathbb{Q}(Q_1,\ldots,Q_r)\big[\alpha_0 \big]$ indexed by $h \in \frac{1}{2}\mathbb{Z}_{\geq 0}$, $n > 0$, $\mathbf{a} \in [r]^n$ and $\mathbf{k} \in \mathbb{Z}_{> 0}^{n}$, 
	such that
	\begin{equation} 
	\label{modifg}\sum_{\substack{h \in \frac{1}{2}\mathbb{Z}_{\geq 0} \\ n \in \mathbb{Z}_{>0}}} \sum_{\substack{\mathbf{a} \in [r]^n \\ \mathbf{k} \in \mathbb{Z}_{> 0}^n}} \frac{\hbar^{h - 1}}{n!} \,\Lambda^{r(k_1 + \cdots + k_n)}\, \Phi_{h,n}\big[\begin{smallmatrix} \mathbf{a} \\ \mathbf{k} \end{smallmatrix}\big] \prod_{l = 1}^n x_{k_l}^{a_l} = \sum_{\substack{g \in \frac{1}{2}\mathbb{Z}_{\geq 0} \\ n \in \mathbb{Z}_{>0} \\ 2g - 2 + n > 0}} \sum_{\substack{\mathbf{a} \in [r]^n \\ \mathbf{k} \in \mathbb{Z}_{> 0}^n}}  \frac{\hbar^{g - 1}}{n!}\,F_{g,n}\big[\begin{smallmatrix} \mathbf{a} \\ \mathbf{k} \end{smallmatrix}\big] \prod_{l = 1}^n x_{k_l}^{a_l},
	\end{equation}
	under the identification $\Lambda^r = \hbar^{\frac{r}{2}}T$. This is an identity between formal series in $\hbar^{\frac{1}{2}}$ and $\Lambda^r$ (or formal series in $\hbar^{\frac{1}{2}}$ and $T$). More precisely, we have
	\begin{equation}
	\label{Phihn}\Phi_{h,n}\big[\begin{smallmatrix} \mathbf{a} \\ \mathbf{k} \end{smallmatrix}\big] = F_{h + \frac{r}{2}(k_1 + \cdots + k_n),n}\big[\begin{smallmatrix} \mathbf{a} \\ \mathbf{k} \end{smallmatrix}\big]\big|_{T = 1}.
	\end{equation}
\end{cor}
\begin{rem}
	\label{rem02s}We call the left-hand side of \eqref{modifg} the \textit{modified genus expansion} of the partition function of the Airy structure. Note that although $F_{0,1} = 0$ and $F_{0,2} = 0$, \eqref{Phihn} gives non-zero $\Phi_{0,1}$ and $\Phi_{0,2}$, and the $\Phi_{g,n}$ are not \emph{a priori} given by a recursion on $2g - 2 + n > 0$. 
\end{rem}

\begin{rem} As a sanity check, we show that the $ d = 1 $ part of the Gaiotto vector $ \ket{\mathfrak G} $ of type A  matches the  result obtained in \cite[Proposition 8.9.5]{Braverman:2014xca}. Using Corollary~\ref{corresum}, we get
	\[
	\left.\frac{\partial \mathcal{Z}}{\partial \Lambda^r} \right|_{\Lambda^r = 0} = \sum_{h \in \frac{1}{2}\mathbb{Z}_{\geq 0}}\sum_{a \in [r]} \hbar^{h-1} \Phi_{h,1}\big[\begin{smallmatrix} a \\ 1 \end{smallmatrix}\big] x^a_1 = \sum_{a \in [r]} \frac{\hbar^{-1} x^a_1}{\prod_{b \neq a}(Q_b-Q_a)},
	\] where we use the fact that only $ h = 0  $ contributes to the sum, thanks to condition (3) of Lemma~\ref{lem:prophbar}, and this contribution is precisely $ \Phi_{0,1}\big[\begin{smallmatrix} a \\ 1 \end{smallmatrix}\big] = F_{\frac{r}{2},1}\big[\begin{smallmatrix} a \\ 1 \end{smallmatrix}\big]\big|_{T=1} = \frac{1}{{\prod_{b \neq a}(Q_b-Q_a)}}. $

\end{rem}

\subsection{The Nekrasov instanton partition function}

Recall from Section~\ref{Sec:instanton} that the instanton partition function $ \mathfrak{Z} $ is identified with the norm squared of the Gaiotto vector using the pairing $  \braket{\ \cdot\ |\  \cdot\  }    $. Using the identification of the Gaiotto vector as the partition function $ \mathcal Z$ of the Airy structure  in Corollary~\ref{cor:AiryG} and the structural properties proved in Section~\ref{Sec:stG}, we will prove some properties of the instanton partition function.

Let us first recall the bilinear pairing $  \braket{\ \cdot\ |\  \cdot\  }  $  that we discussed in Section~\ref{Sec:instanton}. The pairing is uniquely defined by requiring that $\langle  \lambda |\lambda  \rangle  = 1$  and for any rescaled mode $\mathsf{X} $ of $ \mathcal W^k(\mathfrak g) $,
\[
	  \langle u |\mathsf{X} v  \rangle   =  \langle \iota(\mathsf{X}) u | v  \rangle .
\]
The adjoint of $\mathsf{J}_{m}^{a}$, as in equation~\eqref{invol}, is
\[
\iota(\mathsf{J}_m^a) = -\mathsf{J}_{-m}^a - 2(\epsilon_1 + \epsilon_2)\delta_{m,0}  =- \mathsf{J}_{-m}^a - 2\hbar^{\frac{1}{2}}\alpha_0\delta_{m,0}, 
\] 
and in particular, this gives us
\[
\langle x^a_k | x^b_l \rangle  = \frac{\hbar}{k} \delta_{k,l}\delta_{a,b}.
\]

Now, we are ready to prove some properties of  the instanton partition function $\mathfrak{Z}$. By \eqref{eq:Fexp} and Corollary~\ref{corresum}, identifying the partition function $\mathcal{Z}$ with the Whittaker vector $\ket{w}$, we can write $\ket{w} = e^{\Phi}\ket{\lambda}$ with
\begin{equation}
\label{Phidef}
\Phi = \sum_{\substack{h \in \frac{1}{2} \mathbb{Z}_{\geq 0} \\ n \in \mathbb{Z}_{>0}}} \frac{\hbar^{h - 1}}{n!} \sum_{\substack{\mathbf{a} \in [r]^n \\ \mathbf{k} \in \mathbb{Z}_{>0}^n}} \Lambda^{r(\sum_{i = 1}^n k_i)} \Phi_{h,n}\big[\begin{smallmatrix} a_1 & \ldots & a_n \\ k_1 & \ldots & k_n \end{smallmatrix}\big] \prod_{i = 1}^n \frac{\mathsf{J}_{-k_i}^{a_i}}{k_i},
\end{equation}
and $\ket{\lambda}$ the highest-weight state. We now show that we can reconstruct the instanton partition function recursively using Airy structures and topological recursion.

The result will be expressed as a sum over decorated two-level graphs.

\begin{defin}
A two-level graph (see Figure~\ref{Fig2le}) is a finite connected graph with two types of vertices (called in and out) and edges that can only connect an in-vertex to an out-vertex. We require that there exists at least one vertex of each type. We denote $V$ (resp $V_{\rm in}$,  $V_{\rm out}$) the set of vertices (resp. in, out) and $E(v)$ the set of edges incident to a vertex $v$.

A decoration on a two-level graph is the assignment of a half-integer $h(v) \in \frac{1}{2}\mathbb{Z}$ to each vertex $v$ and of an ordered pair $\gamma_e = (a_{e},k_e) \in [r] \times \mathbb{Z}_{> 0}$ to each edge. We denote $\gamma(v) = (\gamma_e)_{e \in E(v)}$ the tuple of labels carried by edges incident to $v$. 

The weight of a decorated two-level graph $\Gamma$ is
\begin{equation}
\label{wgnung}w(\Gamma) = \frac{1}{s(\Gamma)}\,\prod_{v \in V} \hbar^{h(v) - 1} \Lambda^{r(\sum_{e \in E(v)} k_{e})} \Phi_{h(v),|E(v)|}[\gamma(v)] \prod_{e \in E} \frac{\hbar}{k_e},
\end{equation}
including a symmetry factor taking into account the fact that vertices and edges are undistinguishable.
$$
s(\Gamma) = |V_-|!\,|V_+|!\,\prod_{\substack{v \in V_{\rm in} \\ v' \in V_{\rm out}}} |E(v) \cap E(v')|!.
$$

A possibly disconnected decorated two-level graph $\Gamma$ is a finite disjoint union of two-level graphs $\Gamma_1 \cup \cdots \cup \Gamma_n$, and we define its weight as
\begin{equation}
\label{wgnungnun}
w(\Gamma) = \frac{1}{n!} \prod_{i = 1}^n w(\Gamma_i).
\end{equation}
We denote $\mathbf{G}$ (resp. $\widehat{\mathbf{G}}$) the set of (possibly disconnected) decorated two-level graphs.
\end{defin}

\begin{rem}
\label{remon}Notice that $w(\Gamma)$ depends on $\hbar$ only through the factor $\hbar^{-1 + b_1(\Gamma) + \sum_{v \in V} h(v)}$, where $b_1(\Gamma)$ is the first Betti number of $\Gamma$. In particular, the power of $\hbar$ is $\geq -1$.
\end{rem}

\begin{figure}[h!]
\begin{center}
\includegraphics[width = 0.45\textwidth]{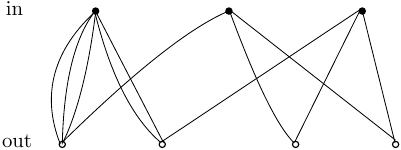}
\caption{\label{Fig2le} A two-level graph with $b_1 = 5$.}
\end{center}
\end{figure}

\begin{prop}\label{instton}
The instanton partition function $\mathfrak{Z} = \braket{ \mathfrak G| \mathfrak G } $ takes the form
\begin{equation}\label{eq:instgraph}
\mathfrak{Z} = \exp\bigg( \sum_{G \in \mathbf G } w(G)\bigg) = \exp\bigg(\sum_{h \in \frac{1}{2}\mathbb{Z}_{\geq0}} \hbar^{h - 1} \mathfrak{F}_h\bigg)
\end{equation}
for some $\mathfrak{F}_h \in \mathbb Q (Q_1,\ldots, Q_r)[\alpha_0][\![\Lambda^r]\!]$. The term $\mathfrak{F}_0$ is called the prepotential, and it is given by the weighted sum over the set $\mathbf{G}_0 \subset \mathbf{G}$ of trees in which all vertices have $h(v) = 0$. 
\end{prop}

\begin{proof}
We represent $\Phi$ of the form \eqref{Phidef} as a weighted sum over graphs consisting of a single vertex, carrying an index $h \in \frac{1}{2}\mathbb{Z}_{\geq 0}$, with $n$ outgoing leaves carrying labels $(a_i,k_i) \in [r] \times \mathbb{Z}_{> 0}$. The weight of such a graph is
	\begin{equation}
	\label{weivhgzn}
	\frac{\hbar^{h-1}}{n!} \Lambda^{r(\sum_{i=1}^n k_i)} \Phi_{h,n}\big[\begin{smallmatrix} a_1 & \ldots & a_n \\ k_1 & \ldots & k_n \end{smallmatrix}\big] \prod_{i = 1}^n \frac{\mathsf{J}_{-k_i}^{a_i}}{k_i}.
	\end{equation}
and the symmetry factor expresses the fact that  the leaves are considered indistinguishable. Since, for $k \neq 0$, the adjoint of $\mathsf{J}_{k}^a$ with respect to the pairing is $\mathsf{J}_{-k}^a$,  we have
\begin{equation}
\label{Oibnznfs}	\Phi^{\dagger} =  \sum_{\substack{h' \in \frac{1}{2}\mathbb{Z}_{\geq 0} \\ n' \in \mathbb{Z}_{>0}}} \frac{\hbar^{h' - 1}}{n'!} \sum_{\substack{\mathbf{a'} \in [r]^{n'} \\ \mathbf{k} \in \mathbb{Z}_{>0}^{n'}}} \Lambda^{r(\sum_{i = 1}^{n'} k'_i)} \Phi_{h',n'}\big[\begin{smallmatrix} a'_1 & \ldots & a'_{n'} \\ k'_1 & \ldots & k'_n \end{smallmatrix}\big] \prod_{i = 1}^{n'} \frac{\mathsf{J}_{k'_i}^{a'_i}}{k'_i}.
\end{equation}
Likewise, we represent it as a weighted sum over graphs consisting a single vertex, carrying an index $h' \in \frac{1}{2}\mathbb{Z}_{\geq 0}$, with $n'$ incoming leaves carrying labels $(a_i',k_i') \in [r] \times \mathbb{Z}_{> 0}$, with a weight:
\begin{equation}
\label{Oibnznf}
\frac{\hbar^{h'-1}}{n'!} \Lambda^{r(\sum_{i=1}^{n'} k'_i)} \Phi_{h',n'}\big[\begin{smallmatrix} a'_1 & \ldots & a'_n \\ k'_1 & \ldots & k'_n \end{smallmatrix}\big] \prod_{i = 1}^{n'} \frac{\mathsf{J}_{k_i'}^{a_i'}}{k_i'}.
\end{equation}
 Then, $ e^{\Phi} $ (resp.  $ e^{\Phi^\dagger} $) can be expressed as a weighted sum over possibly disconnected graphs whose connected components are of the aforementioned kind, and the weight is the product of the weights of its connected components divided by the order of the automorphism group permuting the connected components. 

The product $e^{\Phi^{\dagger}} e^{\Phi}$ can then be rewritten canonically as a sum of normally ordered monomials using the commutation relations in the Heisenberg algebra. Graphically, this is realised by summing all ways of pairings a subset of outgoing leaves in $\Phi$ with ingoing leaves in $\Phi^{\dagger}$. The pairing of an outgoing leaf with label $(a,k)$ and an ingoing leaf with label $(a',k')$ forms an edge and corresponds to a commutator $[\mathsf{J}_{a',k'},\mathsf{J}^{a}_{-k}] = \hbar k \delta_{a,a'}\delta_{k,k'}$. Therefore, we can assume that only leaves carrying the same label can be paired, and assign this label $(a,k)$ to the resulting edge. The weight of such a graph is the product of
\begin{itemize}
\item the scalar part of the weights from \eqref{Oibnznfs} and \eqref{Oibnznf};
\item an updated symmetry factor taking into account the connected components, vertices and edges of the resulting graph are indistinguishable;
\item the normal ordered product of the negative $J$s from unpaired outgoing leaves and of the positive $J$s from unpaired ingoing leaves
\end{itemize}

The instanton partition function is obtained as $\mathfrak{Z} = \bra{\lambda}e^{\Phi^{\dagger}} e^{\Phi} \ket{\lambda}$. As $\ket{\lambda}$ is killed by positive $J$s acting from its left, $\bra{\lambda}$ is killed by negative $J$s to its right, and $\langle \lambda | \lambda \rangle = 1$, we deduce that the effect of $\bra{\lambda} \cdots \ket{\lambda}$ is to keep in this sum only the graphs in which all leaves have been paired. Such graphs are exactly the possibly disconnected two-level graphs, where the in vertices are those coming from $\Phi^{\dagger}$ and the out vertices are those coming from $\Phi$. The description of their weights  in $\mathfrak{Z}$ also matches \eqref{wgnung}-\eqref{wgnungnun}. Therefore:
$$
\mathfrak{Z} = \sum_{\Gamma \in \widehat{\mathbf{G}}} w(\Gamma).
$$
From the multiplicativity of the weight, we deduce that taking the logarithm amounts to summing over connected graphs only, i.e
$$
\mathfrak{Z} = \exp\bigg(\sum_{\Gamma \in \mathbf{G}} w(\Gamma)\bigg).
$$
By definition, $\mathfrak{F}_0$ is the coefficient of $\hbar^{-1}$ in $\ln \mathfrak{Z}$. We see from Remark~\ref{remon} that only the graphs with $h(v) = 0$ for all vertices $v$ and $b_1(\Gamma) = 0$ (i.e. trees) contribute.
\end{proof}

\begin{rem}
	We note that trees in $\mathbf{G}$ are exactly the graphs in which there is at most one edge between any given pair of vertices.
\end{rem}
 
%
%
%\begin{figure}[h!] \nc{Picture to be fixed}
%\includegraphics[width=0.55\textwidth]{Graphinst.pdf}
%\caption{\label{Figraph}Graphs contributing to $\mathfrak{F}_h$. The total genus $h$ should be the sum of genera coming from the vertices and extra genera $(n_j - 1)$ when there are $n_j$ strands connecting the $j$-th vertex to the main vertex. Each strands carries an index $(a,k) \in \mathbb[r] \times \mathbb{Z}_{>0}$ which should be summed over. The weight of a graph is the product of weights attached to the vertices and include a symmetry factor (the strands and the grey vertices are not ordered).}
%\end{figure}

%%%%%%%%%%%%%%%%%%%%%%%%%%%%%%%%%%%%%%%%%%%%%%%%%5

\section{Spectral curve topological recursion for type A}
\label{Sec:TRA}
The purpose of this section is to convert the differential constraints for the partition function $\mathcal{Z}$ of the Airy structure of the previous section into a topological recursion \`a la Chekhov--Eynard--Orantin, which is based on residue analysis on a spectral curve. This conversion is done in three steps:
\begin{enumerate}
\item We start with the partition function $\mathcal{Z}$ uniquely specified by an Airy structure. To the partition function, we associate an infinite sequence of \emph{correlators} $\omega_{g,n}$, with $g \in \frac{1}{2}\mathbb{Z}_{\geq 0}$, $n \in \mathbb{Z}_{>0}$, and $2g-2+n > 0$, which are symmetric differentials living on a \emph{spectral curve}.
\item We recast the differential constraints satisfied by $\mathcal{Z}$ as a system of constraints for the correlators $\omega_{g,n}$, known as \emph{abstract loop equations}.
\item We show that the abstract loop equations can be uniquely solved recursively using residue analysis -- the result is the \emph{spectral curve topological recursion}.
\end{enumerate}

\subsection{Topological recursion at self-dual level \texorpdfstring{$\kappa = 1$}{kappa=1}}
\label{SelfdualTRspec}

We start with type A at self-dual level $\kappa=1$, as the calculations are much simpler. In fact, in this case a similar conversion was carried out in \cite{BKS}. Nevertheless, we redo the calculation here in a slightly different way, to lay the foundations for the arbitrary level case in Section \ref{s:TRarbit}.

\subsubsection{Spectral curve and correlators}

Following the three-step approach outlined above, the starting point is the definition of correlators and the geometry of a spectral curve. The spectral curve for the case at hand is a particular case of what was studied in Section 4 of \cite{BKS} (it corresponds to the unramified case in standard polarization). It is given by the quintuple $(C,x,\omega_{0,1},\omega_{\frac{1}{2},1},\omega_{0,2})$, where:
\begin{itemize}
	\item $C$ is the union of $r$ copies of a formal disk, $C=\bigsqcup_{a=1}^r C^a$, with each $C^a$ a formal disk: $C^a = \operatorname{Spec} \mathbb C[\![\zeta]\!]$. We will denote points on $C$ by $z = \left(\begin{smallmatrix} a \\ \zeta \end{smallmatrix}\right) $, to specify in which copy $C^a$ of the formal disk it lives. We also write $\mathfrak{c}(z) = a$ to denote the component that the point $z \in C$ lives in.
	\item $x\,:\;C \to C_0 = \operatorname{Spec} \mathbb  C[\![\zeta]\!]  $ is the forgetful morphism $  \left(\begin{smallmatrix} a \\ \zeta \end{smallmatrix}\right) \mapsto \zeta$. It is unramified, and we denote by $\mathfrak{a} = x^{-1}(0)$ the central fiber, which in the present case is identified with $[r]$. If $\big( \begin{smallmatrix}a\\ \zeta \end{smallmatrix} \big ) \in C$, we write $\mathfrak{f}\big( \begin{smallmatrix}a\\ \zeta \end{smallmatrix} \big )  = x^{-1}(x\big( \begin{smallmatrix}a\\ \zeta \end{smallmatrix} \big ) ) = x^{-1}(\zeta) \subset C$ for the set of preimages of $\zeta \in C_0$, and $\mathfrak{f}'\big( \begin{smallmatrix}a\\ \zeta \end{smallmatrix} \big )  = \mathfrak{f} \big( \begin{smallmatrix}a\\ \zeta \end{smallmatrix} \big ) \setminus \big\{\begin{smallmatrix}a\\ \zeta \end{smallmatrix} \big\} $ for the set of preimages of $\zeta  \in C_0$ with the original point $\big(\begin{smallmatrix} a \\ \zeta \end{smallmatrix}\big) \in C$ excluded.
	\item For $k \in \mathbb{Z}$ and $a \in [r]$ we define differentials $\dd\xi^a_{k},\dd\xi^*_k \in H^0(C,K_C(* \mathfrak{a}))$ by
\begin{equation}
\label{defdifferentials}
\dd\xi^a_k\big(\begin{smallmatrix} b \\ \zeta \end{smallmatrix}\big) = \delta_{a,b} \zeta^{k - 1}\dd \zeta,\qquad \dd\xi^*_k \big(\begin{smallmatrix} b\\ \zeta \end{smallmatrix}\big) := \sum_{a=1}^r \dd\xi^a_{k}\big(\begin{smallmatrix} b\\ \zeta \end{smallmatrix}\big) = \zeta^{k-1} \dd \zeta.
\end{equation}
Then the $1$-form $\omega_{0,1}$ is given by:
\begin{equation}\label{eq:01}
	\omega_{0,1} \big (\begin{smallmatrix} b\\ \zeta \end{smallmatrix} \big) = \sum_{a \in [r]} F_{0,1}\big[ \begin{smallmatrix}a \\ 0 \end{smallmatrix} \big] \dd \xi^a_0\big (\begin{smallmatrix} b\\ \zeta \end{smallmatrix} \big)  =  \frac{Q_b}{\zeta} \dd \zeta,
\end{equation}
with $Q_1,\ldots,Q_r \in \mathbb{C}$. That is, $F_{0,1}\big[ \begin{smallmatrix}a \\ 0 \end{smallmatrix} \big] = Q_a$.
For future use we define the following function on $C$:
\begin{equation}\label{eq:defy}
	y \big (\begin{smallmatrix} b\\ \zeta \end{smallmatrix} \big) = \frac{Q_b}{\zeta}, \qquad \text{so that} \qquad
	\omega_{0,1} = y\,\dd x . 
\end{equation}
\item The $1$-form $\omega_{\frac{1}{2},1}$ is given by:
	\begin{equation}\label{eq:121}
	\omega_{\frac{1}{2},1} \big(\begin{smallmatrix}b \\ \zeta \end{smallmatrix}\big) = - \alpha_0  \dd \xi^*_0 \big(\begin{smallmatrix}b \\ \zeta \end{smallmatrix}\big) = - \frac{\alpha_0}{\zeta} \dd \zeta = 0,
	\end{equation}
	since $\alpha_0 = \kappa^{-\frac{1}{2}} - \kappa^{\frac{1}{2}} = 0$ at the self-dual level.
\item The bidifferential $\omega_{0,2}$ is given by:
	\begin{equation}\label{def51}
\omega_{0,2}\big(\begin{smallmatrix} b_1 & b_2 \\ \zeta_1 & \zeta_2 \end{smallmatrix}\big)  = \frac{\delta_{b_1,b_2}\dd \zeta_1\dd \zeta_2}{(\zeta_1 - \zeta_2)^2}.
	\end{equation}
	We note that as for $|\zeta_1| > |\zeta_2|$
	\begin{equation}\label{eq:exp02}
	\omega_{0,2}(z_1,z_2) = \sum_{\substack{a \in [r] \\ k \in \mathbb{Z}_{>0}}}k \dd \xi^a_{-k}(z_1) \dd \xi^a_{k}(z_2).
	\end{equation}
\end{itemize}

Next we define the correlators $\omega_{g,n}$ from the partition function $\mathcal{Z}$ of Theorem~\ref{Whit1prop} at self-dual level. To the partition function
\begin{equation}
\label{ZdefTR}
\mathcal{Z} = \exp\bigg(\sum_{\substack{g \in \frac{1}{2}\mathbb{Z}_{\geq 0},\,\,n \in \mathbb{Z}_{>0} \\ 2g - 2 + n >0}} \frac{\hbar^{g - 1}}{n!} \sum_{\substack{a_1,\ldots,a_n \in [r] \\ k_1,\ldots,k_n \in \mathbb{Z}_>0}} F_{g,n}\big[\begin{smallmatrix} a_1 & \cdots & a_n \\ k_1 & \cdots & k_n \end{smallmatrix}\big]\,\prod_{l = 1}^n x_{k_l}^{a_l}\bigg),
\end{equation}
we associate the correlators
\begin{equation}
\label{cordefTR}
\omega_{g,n}(z_1,\ldots,z_n) = \sum_{\substack{a_1,\ldots,a_n \in [r] \\ k_1,\ldots,k_n \in \mathbb{Z}_{>0}}} F_{g,n}\big[\begin{smallmatrix} a_1 & \cdots & a_n \\ k_1 & \cdots & k_n \end{smallmatrix}\big]\,\prod_{l = 1}^n \dd \xi_{-k_l}^{a_l}(z_l).
\end{equation}
The correlators $ \omega_{g,n} $ are symmetric meromorphic $ n $-differentials on $ C $, in other words, they are elements of $H^0\big(C^n,K_C(*\mathfrak{a})^{\boxtimes n}\big)^{\mathfrak{S}_n}$. We stress again that each $z_i$ denote a point on a copy of $C$, hence corresponds to a bi-coordinate $z_i = \big(\begin{smallmatrix} b_i \\ \zeta_i \end{smallmatrix} \big)$.

\subsubsection{Abstract loop equations}\label{sec:aleA}

The next step is to turn the differential constraints satisfied by the partition function $\mathcal{Z}$ into the so-called abstract loop equations for the correlators.

For this purpose, we have to introduce a few objects. We first define an operator that takes a series $f \in \mathcal{M}^{\hbar} =\mathbb C \big[\!\big[(x^a_m)_{a \in [r],\,\,m > 0}\big]\!\big][\![\hbar^{1/2}]\!]$ and transforms it into a symmetric differential form on the spectral curve $C$.

\begin{defin}\label{d:adjop}
Let $f \in \mathcal{M}^{\hbar}$. We define
\begin{equation}
\ad_{g,n}(f) = [\hbar^g] \sum_{\substack{a_1,\ldots,a_n \in [r] \\ k_1, \ldots, k_n \in \mathbb{Z}_{>0}}} \left( \frac{\partial}{\partial x^{a_1}_{k_1}} \, \cdots \,\frac{\partial}{\partial x^{a_n}_{k_n}} f \right)_{x^{a_i}_{k_i} = 0} \dd \xi^{a_1}_{-k_1}(w_1)\, \cdots\,\, \dd \xi^{a_n}_{-k_n}(w_n).
\end{equation}
In other words, it picks the terms of order $\hbar^g$ that are homogeneous of degree $n$ in the variables $x^{a_i}_{k_i}$, and replaces these variables by the corresponding $1$-forms $\dd \xi^{a_i}_{-k_i}(w_i)$ on $C$ in variables $w_1,\ldots,w_n$.
\end{defin}

Our goal is to study the differential constraints satisfied by the partition function $\mathcal{Z}$ of Theorem~\ref{Whit1prop} at self-dual level. Those constraints are obtained through the action of the Heisenberg modes $\mathsf{J}^a_k$ on $\mathcal{Z}$. So we first study this action on its own. In fact, it is convenient to work with the Heisenberg fields directly $\mathsf{J}^a(\zeta)$ --- we replace the variable $z$ from Section \ref{Sec:AiryA} by $\zeta$ here as it will become the variable on the formal disks in the spectral curve $C$ and reserve the letter $z$ to refer to points on $C$. We go further and turn these Heisenberg fields into differential forms on $C$. 

\begin{defin}\label{d:H1f}
We define the \emph{Heisenberg $1$-form} $\mathcal{J}(z)$ on $C$ by:
\[
\mathcal{J}(z)  = \sum_{\substack{a \in [r] \\ k \in \mathbb{Z}}} \mathsf{J}^a_k \dd \xi^a_{-k}(z), \qquad \text{so that} \qquad \mathcal{J}\big( \begin{smallmatrix} a \\ \zeta \end{smallmatrix}\big) = \mathsf{J}^a(\zeta) \dd \zeta.
\]
It will be convenient to separate the zero-modes from the positive and negative modes. Thus we define
\[
\mathcal{J}^*(z) =  \sum_{\substack{a \in [r] \\ k \in \mathbb{Z}_{\neq 0}}} \mathsf{J}^a_k \dd \xi^a_{-k}(z),\qquad
 \mathcal{J}_0(z) = \sum_{a \in [r]} \mathsf{J}^a_0 \dd \xi^a_0(z) = \omega_{0,1}(z),
\]
with $\mathcal{J}(z) = \mathcal{J}^*(z) + \omega_{0,1}(z)$.
We call $\mathcal{J}^*(z)$ the \emph{non-zero Heisenberg $1$-form}, and $\mathcal{J}_0(z)$ the \emph{zero Heisenberg $1$-form}.
\end{defin}

We are interested in the action of the Heisenberg modes $\mathsf{J}^a_m$ on $\mathcal{Z}$, using the differential representation \eqref{eq:hrep}, or equivalently in the action of the Heisenberg $1$-form. We first look at the action of the positive and negative modes only:
\begin{defin}\label{d:defomega}
For any $i \in [r]$, we define
\begin{equation}\label{e:defomega}
\Omega'_{g,i,n}(z_{[i]}; w_{[n]}) := \ad_{g,n} \left(\mathcal{Z}^{-1}\, \prod_{k=1}^i \mathcal{J}^*(z_k)\, \mathcal{Z} \right).
\end{equation}
In other words, we act on the partition function with $i$ copies of the non-zero Heisenberg $1$-form, and then use the operator in Definition \ref{d:adjop} to turn the result into a symmetric differential on $C^n$ in the variables $w_{1}, \ldots, w_{n}$.
\end{defin}

This object will play a key role in the following, and we can express it directly in terms of the correlators $\omega_{g,n}$ introduced in \eqref{cordefTR}.

\begin{lem}\label{l:omeg}
For any $i \in [r]$,
\[
\Omega'_{g,i,n}(z_{[i]}; w_{[n]}) =\sum^{\text{no $\omega_{0,1}$}}_{\substack{\mathbf{L} \vdash [i] \\ \sqcup_{L \in \mathbf{L}} N_L = [n] \\ i + \sum_L (g_L - 1) = g}} \prod_{L \in \mathbf{L}} \omega_{g_L,|L| + |N_L|}(z_L,w_{N_L}),
\]
where we discard all summands containing $\omega_{0,1}$, see Section \ref{s:notation} for an explanation of the notation above.
\end{lem}

\begin{proof}
We omit the proof as these combinatorics have been worked out in a number of papers by now, for instance  \cite[Section 2]{BBCCN18} or \cite[Section 4]{BKS}.
\end{proof}

\begin{rem}\label{r:recursive}
We note that
\[
\begin{split}
\Omega'_{g,1,n}(z;w_{[n]}) & = \omega_{g,n+1}(z,w_{[n]}), \\
\Omega'_{g,2,n}(z_1,z_2;w_{[n]}) & = \omega_{g -1,n + 2}(z_1,z_2,w_{[n]}) + \sum_{\substack{J \sqcup J' = [n] \\ h + h' = g}}^{{\rm no}\,\,\omega_{0,1}} \omega_{h,1+|J|}(z_1,w_{J}) \omega_{h,1 + |J'|}(z_2,w_{J'}),
\end{split}
\]
and so on. More generally, for $i \geq 2$, $\Omega'_{g,i,n}(z_{[i]},w_{[n]})$ only involves $\omega_{g',n'}$ with indices satisfying $2g'-2+n' < 2g-2+(n+1)$.
\end{rem}

We will also be interested in the action of the zero-modes, but in a particular way. Recall that $\mathfrak{f}(z)$ is the set of $r$ preimages $x^{-1}(x(z)) \subset C$. We define the following operator, which calculates the action of $i$ non-zero Heisenberg $1$-forms at $i$ of these preimages, and $r-i$ zero Heisenberg $1$-forms at the other preimages minus $\omega_{0,1}$. More precisely:
\begin{defin}\label{d:Qhatsd}
For $i \in [r]$, we define:
\[
\left[ \mathcal{Q}^{(i)} \Omega'_{g,i,n}(\cdot\,; w_{[n]}) \right](z) := \ad_{g,n} \left(\sum_{\substack{Z \subseteq \mathfrak{f}(z) \\ |Z| = i \\ \{z\} \subseteq Z}}  \mathcal{Z}^{-1} \,\prod_{z' \in Z} \mathcal{J}^*(z') \prod_{z'' \in \mathfrak{f}(z) \setminus Z} (\mathcal{J}_0(z'') - \omega_{0,1}(z) )\,\mathcal{Z} \right).
\]
We note that the first sum is over subsets of preimages of cardinality $i$ that include the point $z \in C$. This is in fact implied by the summation, since for any subset that does not include $z$, the expression vanishes, as it involves $\mathcal{J}_0(z) - \omega_{0,1}(z) =0$. 
\end{defin}

These operators are very easy to evaluate, since the zero-modes $J^a_0$ simply act by multiplication, and the action of the non-zero Heisenberg $1$-forms was calculated in Lemma \ref{l:omeg}. We get:
\begin{lem}\label{l:Qi}
For any $i \in [r]$, we have:
\begin{equation}\label{eq:Qi}
\left[ \mathcal{Q}^{(i)} \Omega'_{g,i,n}(\cdot\,; w_{[n]}) \right](z) = \sum_{\substack{Z \subseteq \mathfrak{f}(z) \\ |Z| = i \\ \{z\} \subseteq Z}} \Omega'_{g,i,n}(Z; w_{[n]}) \prod_{z'' \in \mathfrak{f}(z) \setminus Z} (\omega_{0,1}(z'') - \omega_{0,1}(z) ).
\end{equation}
\end{lem}

\begin{rem}\label{r:why}
We can now explain why the fact that only subsets $Z$ that include $\{z \}$ contribute to the summation is key. It implies that the operator $\left[ \mathcal{Q}^{(1)} \Omega'_{g,1,n}(\cdot\,; w_{[n]}) \right](z) $ is ``diagonal'', \emph{i.e.} it involves the evaluation of $\Omega'_{g,1,n}(z;w_{[n]})$ only at the point $z$ and not at other points in the fiber $\mathfrak{f}(z)$. This facilitates the inversion of this operator, which is necessary to obtain the spectral curve topological recursion (see the proof of Theorem~\ref{lemTRA}).
\end{rem}

We are now ready to turn the differential constraints for $\mathcal{Z}$ into a system of constraints for the correlators:

\begin{lem}\label{l:ALE}
The differential constraints $\mathsf{W}^i_m \mathcal{Z} = \hbar^{\frac{r}{2}} T \delta_{i,r} \delta_{m,1} \mathcal{Z}$, for $i \in [r]$ and $m \in \mathbb{Z}_{>0}$, satisfied by the partition function from  Theorem~\ref{Whit1prop} at self-dual level, imply that, for any  $g \in \frac{1}{2} \mathbb{Z}_{\geq 0}$, $n \in \mathbb{Z}_{\geq 0}$, and $ 2g-2+n \geq 0$, the correlators satisfy the system of equations
\begin{equation}\label{eq:TPALE}
\sum_{i=1}^r \left[ \mathcal{Q}^{(i)} \Omega'_{g,i,n}(\cdot\,; w_{[n]}) \right](z)  - \delta_{g,\frac{r}{2}} \delta_{n,0} \frac{T}{\zeta^{r+1}} (\dd \zeta)^r = O\left( \frac{\dd\zeta}{\zeta} \right)^r \qquad \text{as $\zeta = x(z) \to 0$,}
\end{equation}
with the left-hand-side evaluated in \eqref{eq:Qi}. We call this system of constraints the \emph{abstract loop equations}.
\end{lem}

\begin{proof}
Our starting point is the quantum Miura transform for the strong generators of the $\mathcal{W}^{\mathsf{k}}(\mathfrak{gl}_r)$-algebra, \eqref{Miura}:
\[
  	 \sum_{i=0}^{r}  \mathsf{W}^{i}(\zeta) \nabla_\zeta^{r - i}  = \big( \mathsf{J}^1(\zeta)+\nabla_\zeta \big) \big( \mathsf{J}^2(\zeta)+\nabla_\zeta \big) \,\cdots\, \big( \mathsf{J}^r(\zeta)+ \nabla_\zeta \big), 
\]
where we defined $\nabla_\zeta = \hbar^{\frac{1}{2}} \alpha_0 \partial_\zeta$ for convenience. At self-dual level, this simplifies since we set $\alpha_0$ to zero after commuting all derivatives to the right. In other words, at self-dual level we can think of $\nabla_\zeta$ as a formal variable that commutes with everything; the result is the expression for the $\mathsf{W}^i(\zeta)$ in terms of elementary symmetric polynomials \eqref{Wcritlev}. In fact, at self-dual level we can replace $\nabla_\zeta$ by any commuting object in the equation above. We choose to set $\nabla_\zeta \longrightarrow - y(z)$, where $y(z)$ was defined in \eqref{eq:defy}. We get:
\[
 \sum_{i=0}^{r} (-1)^{r-i} \mathsf{W}^{i}(\zeta) y(z)^{r - i}  = \big( \mathsf{J}^1(\zeta)- y(z) \big) \big( \mathsf{J}^2(\zeta)- y(z) \big) \,\cdots\, \big( \mathsf{J}^r(\zeta)- y(z) \big)\,.
\]
We multiply both sides by $(\dd \zeta)^r$ to turns this equation into a differential on $C$, using the definition of the Heisenberg $1$-forms:
\begin{equation}\label{eq:step}
 \sum_{i=0}^{r} (-1)^{r-i} \mathsf{W}^{i}(\zeta)\, \omega_{0,1}(z)^{r - i}  (\dd \zeta)^i = \prod_{z' \in \mathfrak{f}(z)} \big( \mathcal{J}^*(z') + \mathcal{J}_0(z')- \omega_{0,1}(z) \big)\,,
\end{equation}
where on the right-hand-side we separated the zero-modes.

The partition function $\mathcal{Z}$ satisfies the constraints $\mathsf{W}^i_m \mathcal{Z} = \hbar^{\frac{r}{2}} T \delta_{i,r} \delta_{m,1} \mathcal{Z}$ for $i \in [r]$ and $m \in \mathbb{Z}_{>0}$. Those can be recast as
\begin{equation}\label{eq:ZZconst}
\mathcal{Z}^{-1} \mathsf{W}^i(\zeta) \mathcal{Z} - \hbar^{\frac{r}{2}}T  \delta_{i,r} \zeta^{-r-1} = O\left(\zeta^{-i}\right) \qquad \text{as $\zeta \to 0$}.
\end{equation}
Since $\omega_{0,1}(z) = O\left( \frac{\dd \zeta}{\zeta} \right)$ as $\zeta \to 0$, we obtain that
\[
\mathcal{Z}^{-1}\left( \sum_{i=0}^{r} (-1)^{r-i} \mathsf{W}^{i}(\zeta) \omega_{0,1}(z)^{r - i}  (\dd \zeta)^i \right) \mathcal{Z} - \hbar^{\frac{r}{2}}T  \frac{(\dd \zeta)^r}{\zeta^{r+1}} = O\left( \frac{\dd \zeta}{\zeta}\right)^r \qquad \text{as $\zeta \to 0$.}
\]
By \eqref{eq:step}, this means that
\[
\mathcal{Z}^{-1}\, \prod_{z' \in \mathfrak{f}(z)} \big( \mathcal{J}^*(z') + \mathcal{J}_0(z')- \omega_{0,1}(z) \big) \,\mathcal{Z} - \hbar^{\frac{r}{2}}T  \frac{(\dd \zeta)^r}{\zeta^{r+1}}  = O\left( \frac{\dd \zeta}{\zeta}\right)^r \qquad \text{as $\zeta \to 0$.}
\]
Expanding the left-hand-side in terms with $i$ contributions of the non-zero Heisenberg modes, taking  $\ad_{g,n}$, and using Definition \ref{d:Qhatsd}, we obtain \eqref{eq:TPALE}.
\end{proof}

\subsubsection{Spectral curve topological recursion}

We now solve the abstract loop equations in Lemma \ref{l:ALE} recursively, leading to the spectral curve topological recursion.

\begin{thm}
\label{lemTRA} For $2g - 2 + n \geq 0$, the correlators associated to the partition function of Theorem~\ref{Whit1prop} at self-dual level satisfy the topological recursion:
\begin{equation}
\label{TRA}\begin{split}
\omega_{g,n + 1}(z_0,z_{[n]}) & = \sum_{o \in \mathfrak{a}} \Res_{z = o} \left(\sum_{i=2}^r K(z_0,z) \left[ \mathcal{Q}^{(i)} \Omega'_{g,i,n}(\cdot\,; z_{[n]}) \right](z) \right) \\
& \quad + T\delta_{g,\frac{r}{2}}\delta_{n,0} \sum_{a = 1}^r \frac{\dd \xi_{-1}^a(z_0)}{\prod_{b \neq a} (Q_b - Q_a)},
\end{split}
\end{equation}
with the recursion kernel
\begin{equation}
\label{OKOK} 
K(z_0,z) := \frac{-\int_o^{z} \omega_{0,2}(z_0,\cdot)}{\prod_{z' \in \mathfrak{f}'(z)}\left( \omega_{0,1}(z') - \omega_{0,1}(z) \right)},
\end{equation}
where  $o$ is the point in $\mathfrak{a}$ in the same component of $C$ as $z$, and we recall that $\mathfrak{f}'(z) = \mathfrak{f}(z) \setminus \{ z \}$.
\end{thm}
Note that $K(z_0,z)$ is defined for all $z \in C$, but the formula defining it depends on the connected component to which $z$ belongs; namely, we must choose the base point $o$ to be the unique point of $\mathfrak{a}$ belonging to the same component as $z$.
\begin{proof}
By Lemma \ref{l:ALE}, we know that the correlators satisfy the abstract loop equations
\begin{equation}\label{eq:ll}
\sum_{i=1}^r \left[ \mathcal{Q}^{(i)} \Omega'_{g,i,n}(\cdot\,; w_{[n]}) \right](z)  - \delta_{g,\frac{r}{2}} \delta_{n,0} \frac{T}{\zeta^{r+1}} (\dd \zeta)^r = O\left( \frac{\dd\zeta}{\zeta} \right)^r \qquad \text{as $\zeta \to 0$.}
\end{equation}
We extract from the left-hand-side the term with $i=1$:
\[
 \left[ \mathcal{Q}^{(1)} \Omega'_{g,1,n}(\cdot\,; w_{[n]}) \right](z).
\]
Since $\Omega'_{g,1,n}(z; w_{[n]}) = \omega_{g,n+1}(z,w_{[n]})$, to solve the loop equations we simply need to invert the operator $\mathcal{Q}^{(1)}$. This can be done easily using residue analysis. From Lemma \ref{l:Qi}, we know that
\[
\left[ \mathcal{Q}^{(1)} \Omega'_{g,i,n}(\cdot\,; w_{[n]}) \right](z) = \Omega'_{g,1,n}(z; w_{[n]}) \prod_{z' \in \mathfrak{f}'(z)} (\omega_{0,1}(z') - \omega_{0,1}(z) ),
\]
where $\mathfrak{f}'(z) = \mathfrak{f}(z) \setminus \{ z \}$. All we have to do to invert this operator is divide by the product term and use Cauchy formula to obtain:
\begin{equation}\label{eq:inversion}
 \Omega'_{g,1,n}(z_0; w_{[n]}) =  -\sum_{o \in \mathfrak{a}} \Res_{z = o} K(z_0;z) \left[ \mathcal{Q}^{(1)} \Omega'_{g,i,n}(\cdot\,; w_{[n]}) \right](z) ,
\end{equation}
with $K(z_0;z)$ defined in \eqref{OKOK}. The application of Cauchy formula relies on the fact that $\omega_{0,2}(z,z_0)$ has a pole on the diagonal $z = z_0$ only, so $\int_{o}^{z} \omega_{0,2}(\cdot,z_0) \sim \frac{\dd \zeta_0}{\zeta_0 - \zeta}$ when $\zeta \rightarrow \zeta_0$, \emph{i.e.} it has a simple pole. 

We then substitute back in the abstract loop equations \eqref{eq:ll}. Noting that $K(z_0,z)$ has a zero of degree $r$ at $ o \in \mathfrak{a}$, the residue of $K(z_0,z)$ times the right-hand-side of \eqref{eq:ll} at $o \in \mathfrak{a}$ vanishes. We thus obtain:
\begin{equation}
\label{eq530} \Omega'_{g,1,n}(z_0; w_{[n]})  = \sum_{o \in \mathfrak{a}} \Res_{z = o} K(z_0,z) \left( \sum_{i=2}^r \left[ \mathcal{Q}^{(i)} \Omega'_{g,i,n}(\cdot\,; w_{[n]}) \right](z)  -  \delta_{g,\frac{r}{2}} \delta_{n,0} \frac{T}{\zeta^{r+1}} (\dd \zeta)^r \right).
\end{equation}
We evaluate the term proportional to $T$:
\[
\begin{split}
- \sum_{o \in \mathfrak{a}} \Res_{z = o} K(z_0,z) \frac{T}{\zeta^{r+1}} (\dd \zeta)^r & = \sum_{a=1}^r \Res_{\zeta = 0} \frac{\zeta \dd \xi^a_{1}(z_0)}{\zeta_0 (\zeta_0 - \zeta)} \frac{T}{\prod_{b \neq a}(Q_b - Q_a)} \frac{\zeta^{r-1}( \dd \zeta)^r}{\zeta^{r+1} (\dd \zeta)^{r-1}} \\
 & = T \sum_{a=1}^r \frac{ \dd \xi^a_{-1}(z_0).}{\prod_{b \neq a}(Q_b - Q_a)}
\end{split}
\]
to conclude the proof.
\end{proof}

\begin{rem}
	We remark that, if $u(z_1,\ldots,z_i)$ is a $1$-form with respect to each of its $i$ variables, we have
	$$
	K(z_0,z) \big[\mathcal{Q}^{(i)} u\big](z)  = \sum_{\substack{\{z\} \subset Z \subseteq \mathfrak{f}(z) \\ |Z| = i}} \frac{-\int_o^z \omega_{0,2}(\cdot,z_0)}{\prod_{z' \in \mathfrak{f}(z) \setminus Z} (\omega_{0,1}(z') - \omega_{0,1}(z))}\,u(Z).
	$$
	This property can be used to rewrite the topological recursion in the form presented in \cite{BKS}.
	\end{rem}

\begin{rem}
Remark \ref{r:recursive} ensures that \eqref{TRA} is a recursion on $2g-2+n > 0$, while $\omega_{0,1}$ and $\omega_{0,2}$ can be thought as initial data. The fact that we started from an Airy structure guarantees the symmetry of the $\omega_{g,n}$ produced by this recursion, despite the seemingly special role of the variable $z_0$ in the recursive formula. 
\end{rem}

\begin{rem}
	Using the properties for the $ F_{g,n} $ proved in Lemma~\ref{lem:prophbar}, it is straightforward to bound the order of poles in the correlators $ \omega_{g,n} $. We leave this as an exercise for  the interested reader.
\end{rem}

\subsection{General dilaton shift and change of polarization}
\label{Generald}
From a given Airy structure $(\mathsf{H}_m^a)_{a,m}$, we can obtain a class of new ones $(\tilde{\mathsf{H}}_m^a)_{a,m}$, based on an isomorphic algebra of modes, by  changes of polarization and (when it is well-defined) dilaton shifts \cite{BBCCN18} (see also \cite{KO10} for earlier work). Such an Airy structure takes the form $(\tilde{\mathsf{H}}_m^a)_{a,m}=(\hat{\mathcal{U}}\mathsf{H}_m^a \hat{\mathcal{U}}^{-1})_{a,m}$ and the new partition function is $\tilde{\mathcal{Z}} =\hat{\mathcal{U}}\mathcal{Z}$, where
$$  
\hat{\mathcal{U}} = \exp\left(\sum_{\substack{a \in [r] \\ k \in \mathbb{Z}_{>0}}} \big(F_{0,1}\big[\begin{smallmatrix} a \\ - k \end{smallmatrix}\big] + \hbar^{\frac{1}{2}}\,F_{\frac{1}{2},1}\big[\begin{smallmatrix} a \\ -k \end{smallmatrix}\big]\big)\,\,\frac{\mathsf{J}_k^a}{\hbar\,k} + \sum_{\substack{a,b \in [r] \\ k,l \in \mathbb{Z}_{> 0}}} F_{0,2}\big[\begin{smallmatrix} a & b \\ -k & -l \end{smallmatrix}\big]\,\frac{\mathsf{J}_{k}^{a}\mathsf{J}_l^b}{2\hbar\,kl}\right).
$$ 
When we act on the Airy structure of Theorem~\ref{Whit1prop}, this action on the partition function is well-defined for arbitrary choices of scalars $F_{0,2}\big[\begin{smallmatrix} a & b \\ -k & -l \end{smallmatrix}\big] = F_{0,2}\big[\begin{smallmatrix} b & a \\ -l & -k \end{smallmatrix}\big]$, $F_{0,1}\big[\begin{smallmatrix} a \\ -k \end{smallmatrix}\big]$ and $F_{\frac{1}{2},1}\big[\begin{smallmatrix} a \\ -k \end{smallmatrix}\big]$.

The effect of $F_{0,2}$ is called a change of polarization and expresses the new $F_{g,n}$s as a sum over stable graphs with vertex weight given by the old $F_{g,n}$s, while the effect of $F_{0,1}$ and $F_{\frac{1}{2},1}$ is called a dilaton shift, as it indeed amounts to shifting 
$$ 
x_{k}^a \rightarrow x_k^a + \frac{1}{k}\big(F_{0,1}\big[\begin{smallmatrix} a \\ -k \end{smallmatrix}\big] + \hbar^{\frac{1}{2}} F_{\frac{1}{2},1}\big[\begin{smallmatrix} a \\ -k \end{smallmatrix}\big]\big).
$$

\subsubsection{Spectral curve description}

We can also provide a spectral curve formulation of topological recursion corresponding to this more general class of Airy structures. 

\begin{itemize}
\item The underlying  curve itself does not change: it is still $C = \bigsqcup_{a = 1}^r C^a$ with $C^a = {\rm Spec}\,\mathbb{C}[\![\zeta]\!]$ equipped with the forgetful (unramified) map $x\,:\,C \rightarrow C_0 :=  {\rm Spec}\,\mathbb{C}[\![\zeta]\!]$ simply given by $x\big(\begin{smallmatrix} a \\ \zeta \end{smallmatrix}) = \zeta$. 
\item The dilaton shift amounts to modifying the $1$-form $\omega_{0,1}$ as:
\begin{equation}\label{eq:w01shift}
\omega_{0,1} \big (\begin{smallmatrix} b\\ \zeta \end{smallmatrix} \big) = Q_b \dd \xi^*_{0}  \big (\begin{smallmatrix} b\\ \zeta \end{smallmatrix} \big) + \sum_{\substack{a \in [r] \\ k \in \mathbb{Z}_{> 0} }} F_{0,1} \big[ \begin{smallmatrix} a \\ -k \end{smallmatrix} \big] \dd \xi_{k}^a \big( \begin{smallmatrix} b \\ \zeta \end{smallmatrix} \big),
\end{equation}
and the $1$-form $\omega_{\frac{1}{2},1}$ as:
\begin{equation}\label{eq:w12shift}
\omega_{\frac{1}{2},1}\big(\begin{smallmatrix} b \\ \zeta \end{smallmatrix}\big)  = \widetilde{Q}_a  \dd \xi^*_{0}  \big (\begin{smallmatrix} b\\ \zeta \end{smallmatrix} \big) + \sum_{\substack{a \in [r] \\  k \in \mathbb{Z}_{>0}}} F_{\frac{1}{2},1}\big[\begin{smallmatrix} a \\ -k \end{smallmatrix}\big]  \dd \xi_k^a\big( \begin{smallmatrix} b \\ \zeta \end{smallmatrix} \big), 
\end{equation}
with $\widetilde{Q}_a =-  \alpha_0 =  \kappa^{\frac{1}{2}} - \kappa^{-\frac{1}{2}}= 0$  for all $a \in [r]$. In fact, we still get an Airy structure if we choose $(\widetilde{Q}_a)_{a = 1}^r$  to be arbitrary complex numbers;  clearly,   the proof of Theorem~\ref{Whit1prop} is unchanged if we choose the  zero-modes  to be $\mathsf{J}_0^a = Q_a + \hbar^{\frac{1}{2}}\widetilde{Q}_a$ instead of \eqref{eq:hrep}.
\item The change of polarization modifies the standard bidifferential  $\omega_{0,2}$ of \eqref{def51} to the most general  symmetric bidifferential on the curve $C \times C$ whose only singularity is a double pole on the diagonal with biresidue $1$:
\begin{equation}\label{eq:genw02}
\omega_{0,2}\big(\begin{smallmatrix} b_1 & b_2 \\ \zeta_1 & \zeta_2 \end{smallmatrix}\big) = \frac{\delta_{b_1,b_2}\,\dd \zeta_1 \dd \zeta_2}{(\zeta_1 - \zeta_2)^2} + \sum_{\substack{a_1, a_2 \in [r] \\ k_1, k_2 \in \mathbb{Z}_{>0}}} F_{0,2}\big[\begin{smallmatrix} a_1 & a_2 \\ - k_1 & - k_2 \end{smallmatrix}\big]\, \dd \xi^{a_1}_{k_1} \big(\begin{smallmatrix} b_1  \\ \zeta_1 \end{smallmatrix}\big)  \dd \xi^{a_2}_{k_2} \big(\begin{smallmatrix}  b_2 \\ \zeta_2 \end{smallmatrix}\big) .
\end{equation}
The change of polarization also affects the basis of differentials. For any $a \in [r]$, for $k \geq 0$ the $\dd \xi^a_k \big(\begin{smallmatrix} b \\ \zeta \end{smallmatrix}) = \delta_{a,b} \zeta^{k - 1}\dd \zeta$ stay the same, but the remaining differentials become:
\begin{equation}\label{eq:newbasis}
\dd \xi_{-k}^a\big(\begin{smallmatrix} b \\ \zeta \end{smallmatrix}) = \frac{\delta_{a,b}\dd \zeta}{\zeta^{k + 1}} + \sum_{l \in \mathbb{Z}_{>0}} \frac{F_{0,2}\big[\begin{smallmatrix} a & b \\ -k & -l \end{smallmatrix}\big]}{k}\,\zeta^{l - 1}\dd \zeta, \qquad k > 0.
\end{equation}
The basis of differentials is modified such that
\begin{equation}
\label{526}\omega_{0,2}(z_1,z_2) \,\,\,\mathop{=}_{|x(z_1)| > |x(z_2)|} \,\,\,\sum_{\substack{a \in [r] \\ k \in \mathbb{Z}_{>0}}} k \dd \xi_{-k}^a(z_1)\,\dd \xi_{k}^a(z_2),
\end{equation}
In fact, the choice of $\omega_{0,2}$ is equivalent to a choice  of polarization for the symplectic vector space
$$ 
\bigcap_{o \in \mathfrak{a}} {\rm Ker}\big[\mathop{{\rm Res}}_{o}\,:\, H^0(C,K_C(*\mathfrak{a})) \longrightarrow \mathbb{C}\big]	 \cong \bigoplus_{a = 1}^r	\big\{ \varpi \in   \mathbb C(\!(\zeta)\!) \dd \zeta\,\, \big|\,\, \Res_{\zeta = 0} \varpi = 0  \big\}.
$$
equipped with the symplectic form 
$$
(\varpi_1,\varpi_2) = \sum_{o \in \mathfrak{a}} \mathop{{\rm Res}}_{z = o} \bigg(\int^{z} \varpi_1(\cdot) \bigg) \varpi_2(z).
$$
See \cite[Section 5.1]{BBCCN18} for more details. 
\end{itemize}

At the self-dual level $\kappa = 1$,  the constraints $\hat{\mathcal{U}}(\mathsf{W}_m^i - T\hbar^{\frac{r}{2}}\delta_{m,1}\delta_{i,r})\hat{\mathcal{U}}^{-1} \tilde{\mathcal{Z}}= 0$ for $a \in [r]$ and $m \in \mathbb{Z}_{> 0}$ can be converted into a spectral curve topological recursion just as in Theorem \ref{lemTRA}, following along the same steps with slightly modified definitions.  The end result is that the correlators
$$
\omega_{g,n}(z_1,\ldots,z_n) = \sum_{\substack{a_1,\ldots,a_n \in [r] \\ k_1,\ldots,k_n > 0}} F_{g,n}\big[\begin{smallmatrix} a_1 & \cdots & a_n \\ k_1 & \cdots & k_n \end{smallmatrix}\big] \prod_{i = 1}^n \dd \xi_{-k_i}^{a_i}(z_i),\qquad 2g - 2 + n > 0,
$$ 
from the partition function $\tilde{\mathcal{Z}}$
can be calculated using the formula~\eqref{TRA}, with the spectral curve $(C,x,\omega_{0,1},\omega_{\frac{1}{2},1},\omega_{0,2})$ described above in \eqref{eq:w01shift}, \eqref{eq:w12shift}, and \eqref{eq:genw02}. Again, the Airy structure property  guarantees that the $\omega_{g,n}$ produced by this topological recursion are symmetric in their $n$ variables.

\begin{rem}
As the forgetful map $x: C \to C_0$ in the spectral curve considered here is unramified, the topological recursion that we obtain is different from the original topological recursion of Chekhov--Eynard--Orantin and its higher generalizations, which are based on ramified coverings. In fact, for $r=2$ what we obtain is the topological recursion without branched covers found in \cite[Section 10]{ABCD}; the present work therefore gives it a gauge-theoretic interpretation. From the point of view of Airy structures, the fact that we consider unramified coverings is translated into the statement that our Airy structures were constructed from untwisted modules of the Heisenberg algebra --- see Remark \ref{r:difference} and \cite{BCHORS}. 

We note as well that we could consider more general spectral curves, where the map $x: C \to C_0$ is ramified for some components and unramified for others. From the point of view of Airy structures, this would correspond to constructing Airy structures as modules of direct sums of $\mathcal{W}^{\mathsf{k}}(\mathfrak{gl}_r)$-algebras, with some factors corresponding to modules of the type of Theorem~\ref{Whit1prop}, and others based on the twisted modules studied in \cite{BBCCN18} and \cite{BM,BKS}. Admissible dilaton shifts and changes of polarization could then be applied to mix the summands. You can mix and match at your whim!
\end{rem}

\begin{rem}
We noted in Remark \ref{r:shift} that we could have shifted the one-mode $\mathsf{W}^r_1$ by $\sum_{j \geq 2} \hbar^{\frac{j}{2}} T_j$ for arbitrary $T_2,T_3, \ldots \in \mathbb{C}$ and still obtain Airy structures. In the spectral curve topological recursion it would merely change the second line of \eqref{TRA} to
$$
\sum_{j \geq 2} T_j \delta_{g,\frac{j}{2}}\delta_{n,1} \sum_{a  = 1}^r \frac{\dd \xi_{-1}^a(z_0)}{\prod_{b \neq a} (Q_b - Q_a)}.
$$
However, the structural properties and vanishing properties of the correlators (for instance Lemma~\ref{lem:prophbar}) would be different.  
\end{rem}

\subsection{Topological recursion at arbitrary level \texorpdfstring{$ \kappa $}{kappa}}

\label{s:TRarbit}

In this section we obtain a spectral curve topological recursion, as in Theorem~\ref{lemTRA}, but corresponding to the Airy structure of Theorem~\ref{Whit1prop} for arbitrary level $\kappa$. It turns out that the calculation is more subtle and in fact gives rise to a non-commutative formulation of topological recursion, due to the non-commutative nature of the quantum Miura transform.

We follow again the three-step approach outlined in the beginning of Section \ref{Sec:TRA} to turn an Airy structure into a spectral curve topological recursion, but steps (2) and (3) are significantly more involved now that $\kappa$ is arbitrary.

\subsubsection{Nekrasov--Shatashvili regime}
\label{NSregime}
Before we implement these steps, it is beneficial to rearrange the expansion of the partition function. Recall that the partition function of the Airy structure in Theorem~\ref{Whit1prop} reads $\mathcal{Z} = e^F$, where the free energy $F$ has an expansion of the form of \eqref{eq:Fexp}, reproduced here for convenience:
\begin{equation}\label{eq:fee}
F=\sum_{\substack{g \in \frac{1}{2}\mathbb{Z}_{\geq 0},\,\,n \in \mathbb{Z}_{>0} \\ 2g - 2 + n >0}} \frac{\hbar^{g - 1}}{n!} \sum_{\substack{a_1,\ldots,a_n \in [r] \\ k_1,\ldots,k_n \in \mathbb{Z}_>0}} F_{g,n}\big[\begin{smallmatrix} a_1 & \cdots & a_n \\ k_1 & \cdots & k_n \end{smallmatrix}\big]\,\prod_{l = 1}^n x_{k_l}^{a_l}.
\end{equation}
As proven in Lemma \ref{lem:prophbar}, at arbitrary level the coefficients satisfy the property:
$$
\frac{F_{g,n}\big[\begin{smallmatrix} a_1 & \cdots & a_n \\ k_1 & \cdots & k_n \end{smallmatrix}\big]} { T^{k_1 + \cdots + k_n}} \in   \mathbb{Q}(Q_1,\ldots,Q_r)\big[\alpha_0 \big].
$$
Recall that $\alpha_0 = \kappa^{-\frac{1}{2}} - \kappa^{\frac{1}{2}}$. Let us now define
\begin{equation}\label{boxedalpha}
\alpha = \epsilon_1 + \epsilon_2 = \hbar^{\frac{1}{2}} \alpha_0,
\end{equation}
which has degree $1$ according to the grading \eqref{eq:grading} and \eqref{eq:grading2}.
We would like to rearrange the expansion of the free energy using the variables $(\hbar, \alpha)$ instead of $(\hbar, \alpha_0)$. The expansion then becomes adapted to study the Nekrasov-Shatashvili regime, where $\epsilon_2 \to 0$ and $\epsilon_1$ is set to a finite value, or equivalently $\hbar \to 0$ and $\alpha$ is set to a finite value.

\begin{lem}
\label{lemNS} There exists a sequence of coefficients
\[
F_{g,n}^\alpha\big[\begin{smallmatrix} a_1 & \cdots & a_n \\ k_1 & \cdots & k_n \end{smallmatrix}\big] \in  \mathbb{Q}(Q_1,\ldots,Q_r)[T][\![\alpha]\!],
\]
with $g \in \frac{1}{2} \mathbb{Z}_{\geq 0}$, $n \in \mathbb{Z}_{>0}$, and $2g-2+n>0$, such that the free energy \eqref{eq:fee} for the partition function of Theorem~\ref{Whit1prop} takes the form
\[
F=\sum_{\substack{g \in \frac{1}{2}\mathbb{Z}_{\geq 0},\,\,n \in \mathbb{Z}_{>0} \\ 2g - 2 + n >0}} \frac{\hbar^{g - 1}}{n!} \sum_{\substack{a_1,\ldots,a_n \in [r] \\ k_1,\ldots,k_n \in \mathbb{Z}_>0}} F_{g,n}^\alpha\big[\begin{smallmatrix} a_1 & \cdots & a_n \\ k_1 & \cdots & k_n \end{smallmatrix}\big]\,\prod_{l = 1}^n x_{k_l}^{a_l}
\]
after the identification \eqref{boxedalpha}.
\end{lem}
\begin{proof}
We know from Lemma \ref{lem:prophbar} that:
\[
		\frac{F_{g,n}\big[\begin{smallmatrix} a_1 & \cdots & a_n \\ k_1 & \cdots & k_n \end{smallmatrix}\big]} { T^{k_1 + \cdots + k_n}} \in   \mathbb{Q}(Q_1,\ldots,Q_r)\big[\alpha_0 \big],
	\]
	and that it is a polynomial in $\alpha_0$ of degree at most $2g$.  We can then write
	\[
	F_{g,n}\big[\begin{smallmatrix} a_1 & \cdots & a_n \\ k_1 & \cdots & k_n \end{smallmatrix}\big] = \sum_{d=0}^{2g} \alpha_0^d F_{g,n,d}\big[\begin{smallmatrix} a_1 & \cdots & a_n \\ k_1 & \cdots & k_n \end{smallmatrix}\big],
	\]
and using the redefinition $\alpha_0 = \hbar^{-\frac{1}{2}} \alpha$, the solution to our problem is
\[
F_{g,n}^{\alpha} \big[\begin{smallmatrix} a_1 & \cdots & a_n \\ k_1 & \cdots & k_n \end{smallmatrix}\big]  = \sum_{d \in \mathbb{Z}_{\geq 0}} \alpha^{d}\,F_{g + \frac{d}{2},n,d} \big[\begin{smallmatrix} a_1 & \cdots & a_n \\ k_1 & \cdots & k_n \end{smallmatrix}\big].
\]
\end{proof}

\noindent From now on we will work with this expansion in $(\hbar, \alpha)$.

\subsubsection{The spectral curve and the correlators}

The first step is to define the geometry. The spectral curve is almost the same as in Section~\ref{SelfdualTRspec}, and we keep notations from there. Our spectral curve is given by the quintuple $(C,x,\omega^\alpha_{0,1},\omega^\alpha_{\frac{1}{2},1},\omega^\alpha_{0,2})$, with:
\begin{equation}
\label{eq:w02arbit}
\begin{split}
C & = \bigsqcup_{a=1}^r C^a,\\
x & : ~C \to C_0 \text{ the forgetful map,}\\
\omega^\alpha_{0,1} \big( \begin{smallmatrix} a \\ \zeta \end{smallmatrix}\big) & = y^\alpha \big( \begin{smallmatrix} a \\ \zeta \end{smallmatrix}\big) \dd \zeta = \frac{Q_a - \alpha}{\zeta} \dd \zeta,\\
\omega^\alpha_{\frac{1}{2},1} \big( \begin{smallmatrix} a \\ \zeta \end{smallmatrix}\big) & = 0,\\
\omega^\alpha_{0,2}\big(\begin{smallmatrix} a_1 & a_2 \\ \zeta_1 & \zeta_2 \end{smallmatrix}\big)  & = \frac{\delta_{a_1,a_2}\dd \zeta_1\dd \zeta_2}{(\zeta_1 - \zeta_2)^2}.
\end{split}
\end{equation}
The only difference is in the $1$-form $\omega^\alpha_{0,1} $, which is due to the fact that we are now working with $(\hbar, \alpha)$ instead of $(\hbar, \alpha_0)$. We recall that we denote points on $C$ by $z = \left(\begin{smallmatrix} a \\ \zeta \end{smallmatrix}\right) $, and that we use $\mathfrak{c}(z) = a$ to denote the component that the point $z \in C$ lives in.

To the partition function $\mathcal{Z} = e^F$ uniquely defined by the Airy structure of Theorem~\ref{Whit1prop}, whose free energy has an expansion (in $(\hbar,\alpha)$) of the form
\[
F=\sum_{\substack{g \in \frac{1}{2}\mathbb{Z}_{\geq 0},\,\,n \in \mathbb{Z}_{>0} \\ 2g - 2 + n >0}} \frac{\hbar^{g - 1}}{n!} \sum_{\substack{a_1,\ldots,a_n \in [r] \\ k_1,\ldots,k_n \in \mathbb{Z}_>0}} F_{g,n}^\alpha\big[\begin{smallmatrix} a_1 & \cdots & a_n \\ k_1 & \cdots & k_n \end{smallmatrix}\big]\,\prod_{l = 1}^n x_{k_l}^{a_l},
\]
we associate \emph{Nekrasov-Shatashvili (NS) correlators} by
\[
\label{NScor} \omega_{g,n}^{\alpha}(z_1,\ldots,z_n) := \sum_{\substack{a_1,\ldots,a_n \in [r] \\ k_1,\ldots,k_n \in \mathbb{Z}_{>0}}} F_{g,n}^{\alpha}\big[\begin{smallmatrix} a_1 & \cdots & a_n \\ k_1 & \cdots & k_n \end{smallmatrix}\big] \prod_{i = 1}^n \dd \xi_{-k_i}^{a_i}(z_i),
\]
where the $\dd\xi_{-k}^{a}$ are the differentials defined in \eqref{defdifferentials}. Our goal is to formulate a topological recursion satisfied by the NS correlators.

\subsubsection{Abstract loop equations}

The next step is to recast the differential constraints satisfied by the partition function $\mathcal{Z}$ as a set of abstract loop equations for the NS correlators. It turns out that the result is structurally the same as for self-dual level -- but the devil is in the details. Let us first state the result, and then delve into its intricacies.

We proceed along the same path as in Section \ref{SelfdualTRspec}. We define the map $\ad_{g,n}$ as in Definition \ref{d:adjop} and the Heisenberg $1$-form as in Definition \ref{d:H1f}. We define the action of the non-zero Heisenberg $1$-forms on $\mathcal{Z}$ as in Definition \ref{d:defomega}, which can be calculated as in Lemma \ref{l:omeg}, but now for the NS correlators. To highlight the distinction with the former correlators, we use the notation
$\Omega'^{\alpha}_{g,i,n}(z_{[i]}; w_{[n]})$. 

We obtain the following lemma.

\begin{lem}\label{l:ALEarbit}
The differential constraints $\mathsf{W}^i_m \mathcal{Z} = \hbar^{\frac{r}{2}} T \delta_{i,r} \delta_{m,1} \mathcal{Z}$, for $i \in [r]$ and $m \in \mathbb{Z}_{>0}$, satisfied by the partition function from Theorem~\ref{Whit1prop} at arbitrary level, imply that, for any  $g \in \frac{1}{2} \mathbb{Z}_{\geq 0}$, $n \in \mathbb{Z}_{\geq 0}$, and $ 2g-2+n \geq 0$, the NS correlators satisfy the system of equations
\begin{equation}\label{eq:ALEarbit}
\sum_{i=1}^r \left[ \widehat{\mathcal{Q}}^{(i)} \Omega'^{\alpha}_{g,i,n}(\cdot\,; w_{[n]}) \right](z)  - \delta_{g,\frac{r}{2}} \delta_{n,0} \frac{T}{\zeta^{r+1}} (\dd \zeta)^r = O\left( \frac{\dd\zeta}{\zeta} \right)^r \qquad \text{as $x(z) = \zeta \to 0$.}
\end{equation}
We call this system of constraints the \emph{abstract loop equations}.
\end{lem}

The complexity here is in the definition of the operator $\left[\widehat{\mathcal{Q}}^{(i)} \Omega'^{\alpha}_{g,i,n}(\cdot\,; w_{[n]}) \right](z)$ --- it is not the same as in Definition \ref{d:Qhatsd}. In spirit, it achieves the same goal, that is, it incorporates the action of the zero-modes. However, recall that the starting point of the proof of Lemma \ref{l:ALE} was the quantum Miura transform, and the fact that at self-dual level we can think of $\nabla_\zeta$ as a formal symbol that commutes with everything. This is certainly not the case anymore at arbitrary level. This non-commutativity gives rise to a much more intricate definition for the operator $\left[ \widehat{\mathcal{Q}}^{(i)} \Omega'^{\alpha}_{g,i,n}(\cdot\,; w_{[n]}) \right](z)$, to which we now turn to.

We first need to define a notion of ordering for points in the spectral curve $C$, and for products of functions or differentials on $C$. Recall that the spectral curve $C = \bigsqcup_{a=1}^r C^a$ has a natural notion of ordering, given by indexing the components with $a\in [r]$.

\begin{defin}
Pick any two points $z_1, z_2 \in C$. We say that $z_1 \preceq z_2$ if $\mathfrak{c}(z_1) \leq \mathfrak{c}(z_2)$, and call this ordering \emph{the spectral ordering}. 
We define the \emph{spectrally ordered product of functions (or differentials) on $C$},
\[
\overset{\nearrow}{\mathcal{SP}}\left( f_1(z_1)\, \cdots\, f_n(z_n) \right),
\]
for any set of distinct points $z_1, \ldots, z_n \in \mathfrak{f}(z)$,
as taking the product where the order of the factors from left to right is in increasing spectral order of the points $z_k \in C$. We denote the reverse ordering, where the order of the factors from left to right is in decreasing spectral order, by
\[
\overset{\searrow}{\mathcal{SP}}\left( f_1(z_1) \,\cdots \,f_n(z_n) \right),
\]
\end{defin}

For example, given functions $f,g,h$ on $C$, we have:
\[
\begin{split}
\overset{\nearrow}{\mathcal{SP}} \left( f\big( \begin{smallmatrix} 1 \\ \zeta \end{smallmatrix}\big) g\big( \begin{smallmatrix} 4 \\ \zeta \end{smallmatrix}\big) h\big( \begin{smallmatrix} 2 \\ \zeta \end{smallmatrix}\big)  \right) & = f\big( \begin{smallmatrix} 1 \\ \zeta \end{smallmatrix}\big)  h\big( \begin{smallmatrix} 2 \\ \zeta \end{smallmatrix}\big) g\big( \begin{smallmatrix} 4 \\ \zeta \end{smallmatrix}\big), \\
\overset{\searrow}{\mathcal{SP}} \left( f\big( \begin{smallmatrix} 1 \\ \zeta \end{smallmatrix}\big) g\big( \begin{smallmatrix} 4 \\ \zeta \end{smallmatrix}\big) h\big( \begin{smallmatrix} 2 \\ \zeta \end{smallmatrix}\big)  \right) & = g\big( \begin{smallmatrix} 4 \\ \zeta \end{smallmatrix}\big)  h\big( \begin{smallmatrix} 2 \\ \zeta \end{smallmatrix}\big) f\big( \begin{smallmatrix} 1 \\ \zeta \end{smallmatrix}\big).
\end{split}
\]

\begin{rem}
The spectrally ordered product should not be confused with normal ordering. Normal ordering is about ordering of the modes of the fields while spectral ordering is about ordering of the dependence in the points in $C$.
\end{rem}

Now we would like to define the operator $\widehat{\mathcal{Q}}^{(i)}$ as in Definition \ref{d:Qhatsd}, but with $-\omega_{0,1}(z)$ replaced by $\nabla_\zeta$, where\footnote{More precisely, $\nabla_{\zeta}$ applied to a $k$-form $u$ in the variable $\zeta$ should be understood as $$\nabla_{\zeta} = (\dd \zeta)^{k + 1} \partial_{\zeta}\bigg(\frac{u(\zeta)}{(\dd \zeta)^{k}}\bigg).$$} $\nabla_\zeta =\alpha\,\dd_{\zeta}$. Indeed, this is precisely the combination that will appear in the quantum Miura transform \eqref{Miura} at arbitrary level. However, since $\nabla_\zeta$ certainly does not commute with Heisenberg $1$-forms, one must be careful with ordering of the variables, which is why we introduced the spectrally ordered product. 

Naively, we would like to define the operators as
\[
\left[ \widehat{\mathcal{Q}}^{(i)} \Omega'^{\alpha}_{g,i,n}(\cdot\,; w_{[n]}) \right](z) = \ad_{g,n} \left(\sum_{\substack{Z \subseteq \mathfrak{f}(z) \\ |Z| = i }}  \mathcal{Z}^{-1}\, \overset{\nearrow}{\mathcal{SP}}\left( \prod_{z' \in Z} \mathcal{J}^*(z') \prod_{z'' \in \mathfrak{f}(z) \setminus Z} (\mathcal{J}_0(z'') +\nabla_\zeta ) \right)\,\mathcal{Z} \right),
\]
with the spectral ordering of the factors done with respect to the variables in $\mathfrak{f}(z)$.
However, this is too naive. Recall that in Definition \ref{d:Qhatsd}, the sum was naturally restricted to subsets $Z \subseteq \mathfrak{f}(z)$ that contain $\{z \} \subseteq Z$, because of the factors $(\mathcal{J}_0(z'') - \omega_{0,1}(z))$. As highlighted in Remark \ref{r:why}, this property was crucial to find the inverse of the operator $\mathcal{Q}^{(1)}$ --- which is required to get topological recursion --- as it became diagonal. Unfortunately, in the present case there is no such natural restriction. Instead, we can make sure that $Z$ always includes at least one point of spectral order greater or equal than $z$.  While this is not as nice as having $\{z \} \subseteq Z$, it will be sufficient to find the inverse of $\widehat{\mathcal{Q}}^{(1)}$, as it becomes triangular.

To achieve this, we multiply on the right by a particular function $\Psi(z)$ on $C$, whose restrictions $\psi_b(\zeta) = \Psi\big(\begin{smallmatrix} b \\ \zeta \end{smallmatrix} \big)$ to the various components of $C$ give a basis of solutions of
$$
\hat{\mathcal{Y}} \psi_b(\zeta) =  \overset{\nearrow}{\mathcal{SP}}\bigg(\prod_{\substack{z' \in \mathfrak{f}(z)}}   (\mathcal{J}_0(z') +\nabla_\zeta )  \bigg)  \psi_b(\zeta) = (\dd \zeta)^{r} \prod_{a = 1}^r \big((\epsilon_1 + \epsilon_2)\partial_\zeta + \tfrac{Q_a - (\epsilon_1 + \epsilon_2)}{\zeta}\big) \psi_b(\zeta) = 0.
$$
%For our convenience, $\Psi(z)$ is actually not a function: instead, for each $a \in [r]$ we make $\Psi\big(\begin{smallmatrix} a \\ \zeta \end{smallmatrix}\big)$ a $(r - a)$-differential form.

\begin{defin}\label{d:psi}
Let $\Psi(z)$ be a function on $C$ such that:
\begin{equation*}
\begin{split}
 \overset{\nearrow}{\mathcal{SP}}\bigg(\prod_{\substack{z' \in \mathfrak{f}(z) \\ z' \succeq z}}   (\mathcal{J}_0(z') +\nabla_\zeta )  \bigg) \Psi(z) & = 0,\\
 \overset{\nearrow}{\mathcal{SP}}\bigg(\prod_{\substack{z' \in \mathfrak{f}(z) \\ z' \succ z}}   (\mathcal{J}_0(z') +\nabla_\zeta )  \bigg) \Psi(z) & \neq 0.
\end{split}
\end{equation*}
In other words, $\Psi \big( \begin{smallmatrix} a \\ \zeta \end{smallmatrix}\big)$ is killed by the spectral ordered product of all factors $ (\mathcal{J}_0\big( \begin{smallmatrix} a' \\ \zeta \end{smallmatrix}\big)+\nabla_\zeta ) $ with  $a' \geq a$, but it is not killed by any product with a smaller number of factors. 
For any $z_0 \in \mathfrak{f}(z)$, we also define the differential of degree $r-\mathfrak{c}(z_0)$:
\[
\Psi_{z_0}(z) = \overset{\nearrow}{\mathcal{SP}}\bigg(\prod_{\substack{z' \in \mathfrak{f}(z) \\ z' \succ z_0}}   (\mathcal{J}_0(z') +\nabla_\zeta )  \bigg) \Psi(z). 
\]
In particular, $\Psi_{z}(z) = 0$ if $z_0 \prec z$, and we set $\Psi_{z_0}(z) = \Psi(z)$ if $\mathfrak{c}(z_0) = r$.
\end{defin}

\begin{lem}
\label{lem520} Recall that $\mathcal{J}_0\big(\begin{smallmatrix}a \\ \zeta \end{smallmatrix}\big)= \frac{Q_a - \alpha}{\zeta} \dd \zeta$, and $\nabla_\zeta = \alpha\, \dd_{\zeta}$. Then the choice
\begin{equation}\label{eq:psi}
\Psi(z) = \zeta^{1 + r-\mathfrak{c}(z) - \frac{Q_{\mathfrak{c}(z)}}{\alpha}}
\end{equation}
matches Definition \ref{d:psi} (for generic constants $Q_a$s). For this choice of $\Psi(z)$, we get:
\[
\Psi_{z_0} (z) =(\dd \zeta)^{r-\mathfrak{c}(z_0)} \left(\prod_{j=\mathfrak{c}(z_0)+1}^r \left(Q_j-Q_{\mathfrak{c}(z)} + \alpha(j-\mathfrak{c}(z)) \right) \right) \zeta^{1+\mathfrak{c}(z_0)-\mathfrak{c}(z)-\frac{Q_{\mathfrak{c}(z)}}{\alpha}}.
\]
\end{lem}

\begin{proof}
Direct calculation.
\end{proof}

\begin{rem}
We remark that Definition~\ref{d:psi} does not specify uniquely $\Psi(z)$. Any two such functions  $\Psi_1$ and $\Psi_2$ --- thought of as column vectors with components $\Psi_i\big(\begin{smallmatrix} a \\ \zeta \end{smallmatrix}\big)$, $a \in [r]$, $i=1,2$ --- must be related by $\Psi_1 = U \Psi_2$ with $U$ a constant, upper triangular, invertible matrix. In any case, from now on we choose $\Psi(z)$ as in \eqref{eq:psi}.
\end{rem}

\begin{defin}\label{d:QII}
For $i \in [r]$, we define:
\begin{equation*}
\begin{split}
& \quad \left[ \widehat{\mathcal{Q}}^{(i)} \Omega'^{\alpha}_{g,i,n}(\cdot\,; w_{[n]}) \right](z) \\
& = \ad_{g,n} \left(\sum_{\substack{Z \subseteq \mathfrak{f}(z) \\ |Z| = i }}  \mathcal{Z}^{-1} \Psi^{-1}(z) \, \overset{\nearrow}{\mathcal{SP}}\left( \prod_{z' \in Z} \mathcal{J}^*(z') \prod_{z'' \in \mathfrak{f}(z) \setminus Z} (\mathcal{J}_0(z'') +\nabla_\zeta ) \right)\, \Psi (z) \mathcal{Z} \right),
\end{split}
\end{equation*}
where we let the differential operator in $\zeta$ act on $\Psi(z)$. Due to the defining properties of $\Psi(z)$, the only non-vanishing terms in the sum over $Z \subseteq \mathfrak{f}(z)$ are for subsets $Z$ that include at least one $z' \in Z$ such that $z' \succeq z$ according to spectral order.
\end{defin}

With this definition of the operator, we can prove Lemma \ref{l:ALEarbit}.

\begin{proof}[Proof of Lemma \ref{l:ALEarbit}]
The starting point is the quantum Miura transform \eqref{Miura}. We multiply right away by $(\dd \zeta)^r$ and separate the zero-modes to get:
\begin{equation} \label{eq:miurat}
  	\sum_{i=0}^{r}  \mathsf{W}^{i}(\zeta) \nabla_\zeta^{r - i}  (\dd \zeta)^i  = \overset{\nearrow}{\mathcal{SP}}\left(\prod_{z' \in \mathfrak{f}(z)} (\mathcal{J}^*(z') + \mathcal{J}_0(z') + \nabla_\zeta) \right)  .
\end{equation}
The partition function $\mathcal{Z}$ satisfies the constraints $\mathsf{W}^i_m \mathcal{Z} = \hbar^{\frac{r}{2}} T \delta_{i,r} \delta_{m,1} \mathcal{Z}$ for $i \in [r]$ and $m \in \mathbb{Z}_{>0}$. Those can be recast as the statement that
\begin{equation}\label{eq:ZZconst2}
\mathcal{Z}^{-1} \mathsf{W}^i(\zeta) \mathcal{Z} - \hbar^{\frac{r}{2}}T  \delta_{i,r} \zeta^{-r-1} = O\left(\zeta^{-i}\right) \qquad \text{as $\zeta \to 0$}.
\end{equation}
We would like to turn those into constraints for
\begin{equation}\label{eq:trt}
\mathcal{Z}^{-1} \Psi^{-1}(z) \left( \sum_{i=0}^{r}  \mathsf{W}^{i}(\zeta) \nabla_\zeta^{r - i}  (\dd \zeta)^i \right) \Psi(z) \mathcal{Z} - \hbar^{\frac{r}{2}}T  \delta_{i,r} \zeta^{-r-1}.
\end{equation}
Since $\Psi\big(\begin{smallmatrix} a\\ \zeta \end{smallmatrix}\big) = \zeta^{r-a+1 - \frac{Q_a}{\alpha}}$,  we have $\Psi^{-1}(z) \nabla_\zeta^{r-i} \Psi(z)= N \zeta^{i-r}$ for some irrelevant but non-zero constant $N$. By \eqref{eq:ZZconst2}, this means that the expression in \eqref{eq:trt} behaves like $\left(\frac{\dd \zeta}{\zeta} \right)^r$ as $\zeta \to 0$. By \eqref{eq:miurat}, this means that
\[
\mathcal{Z}^{-1} \Psi^{-1}(z)\, \overset{\nearrow}{\mathcal{SP}}\left(\prod_{z' \in \mathfrak{f}(z)} (\mathcal{J}^*(z') + \mathcal{J}_0(z') + \nabla_\zeta) \right)\, \Psi(z) \mathcal{Z} - \frac{\hbar^{\frac{r}{2}}T  \delta_{i,r}}{\zeta^{r+1}} = O\left(\frac{\dd \zeta}{\zeta} \right)^r
\]
as $\zeta \to 0$. Expanding the left-hand-side in terms with $i$ contributions of the non-zero Heisenberg modes,  taking $\ad_{g,n}$ and using Definition \ref{d:QII}, we get the abstract loop equations \eqref{eq:ALEarbit}.
\end{proof}

Before we proceed, we remark that it would be nice to have an explicit formula for the $\left[ \widehat{\mathcal{Q}}^{(i)} \Omega'^{\alpha}_{g,i,n}(\cdot\,; w_{[n]}) \right](z)$ as in Lemma \ref{l:Qi}. Unfortunately, this is not so easy to write down, because of the non-commutativity of the $\nabla_\zeta$. However, there is one simple case, which is when $i=r$: in this case only non-zero Heisenberg modes appear, therefore $\widehat{\mathcal{Q}}^{(r)}$ acts as the identity, that is
\[
\left[ \widehat{\mathcal{Q}}^{(r)} \Omega'^{\alpha}_{g,r,n}(\cdot\,; w_{[n]}) \right](z) =  \Omega'^{\alpha}_{g,r,n}(z; w_{[n]}),
\]
 just like at self-dual level.

\subsubsection{The topological recursion}

Lemma \ref{l:ALEarbit} establishes the abstract loop equations for arbitrary $\kappa$. We can solve just as we did at self-dual level in Theorem~\ref{lemTRA}. In fact, the result is again structurally the same:

\begin{thm}
\label{lemTRAarbit} For $2g - 2 + n \geq 0$, the NS correlators associated to the partition function of Theorem~\ref{Whit1prop} at arbitrary level satisfy the topological recursion:
\begin{equation}
\label{TRAarbit}\begin{split}
\omega_{g,n + 1}^{\alpha}(z_0,z_{[n]}) & = \sum_{o \in \mathfrak{a}} \Res_{z = o} \left( \sum_{i=2}^r K^{\alpha}(z_0,z) \left[ \widehat{\mathcal{Q}}^{(i)} \Omega'^{\alpha}_{g,i,n}(\cdot\,; z_{[n]}) \right](z) \right) \\
& \quad + T\delta_{g,\frac{r}{2}}\delta_{n,0} \sum_{a = 1}^r \frac{\dd \xi_{-1}^a(z_0)}{\prod_{b \neq a} (Q_b - Q_a)}.
\end{split}
\end{equation}
\end{thm}

As usual, the complexity lies in the details --- in this case, in the definition of the recursion kernel $K^{\alpha}(z_0,z)$, which is not given by \eqref{OKOK}. However, it achieves precisely the same goal, namely the inversion of the operator $ \left[ \widehat{\mathcal{Q}}^{(1)} \Omega'^{\alpha}_{g,1,n}(\cdot\,; w_{[n]}) \right](z)$. But given the complexity of the $\widehat{\mathcal{Q}}$ operators at arbitrary level, it is significantly more difficult to obtain the recursion kernel.

Let us first define what the kernel is supposed to do, assuming its existence. We want the recursion kernel $K^{\alpha}(z_0,z)$ to extract $\Omega'^{\alpha}_{g,1,n}(z; w_{[n]})$ from the knowledge of the quantity $\big[\widehat{\mathcal{Q}}^{(1)}\Omega'^{\alpha}_{g,1,n}(\cdot\,; w_{[n]})\big](z)$. To clarify, let us abstract the definition of the $\widehat{\mathcal{Q}}^{(1)}$ operator from the particular context of $\Omega'^{\alpha}_{g,1,n}$.
\begin{defin}\label{d:Q1arbit}
We define the operator $\widehat{\mathcal{Q}}^{(1)}$, which takes as input a meromorphic $1$-form $u$ on $C$, and outputs a meromorphic differential of degree $r$ on $C$ by the formula:
\begin{equation}
\label{eq:Q1abstract}
\begin{split}
[\widehat{\mathcal{Q}}^{(1)} u](z) & = \sum_{\substack{z' \in \mathfrak{f}(z) \\ z' \succeq z}} \Psi^{-1}(z)   \overset{\nearrow}{\mathcal{SP} }\left( u(z') \prod_{z'' \in \mathfrak{f}(z) \setminus \{ z' \}} (\mathcal{J}_0(z'') + \nabla_\zeta) \right) \Psi(z) \\
& = \sum_{\substack{z' \in \mathfrak{f}(z) \\ z' \succeq z}} \Psi^{-1}(z)  \overset{\nearrow}{\mathcal{SP}} \left( \prod_{\substack{z'' \in \mathfrak{f}(z) \\  z'' \prec z'}}  (\mathcal{J}_0(z'') + \nabla_\zeta) \right) u(z') \Psi_{z'}(z),
\end{split}
\end{equation}
where $\Psi_{z'}(z)$ was introduced in Definition \ref{d:psi}. Then, $\big[\widehat{\mathcal{Q}}^{(1)}\Omega'^{\alpha}_{g,1,n}(\cdot\,; w_{[n]})\big](z)$ from Definition~\ref{d:QII} is just an instance of application of the operator $\widehat{\mathcal{Q}}^{(1)}$ to the $1$-form $\Omega'^{\alpha}_{g,1,n}(z; w_{[n]}) $ in the variable $z$.
\end{defin}

We can now present our wish list for the kernel $K^{\alpha}$. 
\begin{itemize}
\item[(i)] For any meromorphic $1$-form $u$ on $C$ that is a linear combination of $(\dd \xi_{-k}^a)_{k,a}$, we have
\[
u(z_0) = -\sum_{o \in \mathfrak{a}} \Res_{z = o} K^{\alpha}(z_0,z) \big[\widehat{\mathcal{Q}}^{(1)}u\big](z).
\]
This assumption on $u$ is equivalent to
\begin{equation}
\label{uusa}u(z_0) = \sum_{o \in \mathfrak{a}} \Res_{z = o} \bigg(\int_{o}^z \omega_{0,2}(\cdot,z_0)\bigg) u(z),
\end{equation}
and for the spectral curve \eqref{eq:w02arbit} it means concretely that $u\big(\begin{smallmatrix} a \\ \zeta\end{smallmatrix}\big) \in \zeta^{-2}\mathbb{C}[\zeta^{-1}]\dd \zeta$ is assumed for any $a \in [r]$.
\item[(ii)] $K^{\alpha}(z_0,z)$ has a zero of degree $r$ at all $z=o \in \mathfrak{a}$ (as in the self-dual case).
\end{itemize} 
It is not clear \emph{a priori} that such a kernel exists, but we will construct a canonical one explicitly in Section \ref{s:kernel} for generic $(Q_a)_{a = 1}^r$, and thus justify its existence. 

\begin{proof}[Proof of Theorem~\ref{lemTRAarbit}]
Accepting the existence of such a $K^{\alpha}(z_0,z)$, the proof of the spectral curve topological recursion \eqref{TRAarbit} from the abstract loop equations in Lemma \ref{l:ALEarbit} is then precisely the same as in the self-dual case, Theorem~\ref{lemTRA}, so we do not repeat it here. The only small difference is to use property (i) for $u(z) = \Omega'^{\alpha}_{g,1,n}(z,w_{[n]})$ instead of \eqref{eq:inversion}, and the step above \eqref{eq530} remains valid now due to property (ii).
\end{proof}

\subsubsection{The recursion kernel}
\label{s:kernel}

We now turn to the existence of the recursion kernel $K^{\alpha}(z_0,z)$ satisfying the conditions (i) and (ii) above. We need the following auxiliary operator.

\begin{defin}\label{d:Rhat}
We define an operator $\widehat{\mathcal{R}}$ which takes as input a meromorphic $1$-form $u$ on $C$ and a preimage $z' \in \mathfrak{f}(z)$, and outputs a meromorphic $1$-form on $C$. For $z' \succ z$,
\[
\left[ \widehat{\mathcal{R}}_{z' }u \right](z) =  \frac{1}{\Psi_z(z)}\overset{\nearrow}{\mathcal{SP}} \left( \prod_{\substack{z'' \in \mathfrak{f}(z) \\  z\preceq  z'' \prec z'}}(\mathcal{J}_0(z'')+ \nabla_\zeta) \right) u(z') \Psi_{z'}(z),
\]
while if $z' \preceq z$, we set $\left[ \widehat{\mathcal{R}}_{z' }u \right](z) = 0$.
\end{defin}

\begin{lem}\label{l:kernel1}
Assume that we have $R(z_0,z)$ which is a $1$-form in $z_0$ and a differential of degree $1 -r $ in $z$, which has zeroes of order $r$ when $z \rightarrow \mathfrak{a}$ and satisfying
\begin{equation}\label{eq:Uz0}
\overset{\searrow}{\mathcal{SP}} \left( \prod_{\substack{z' \in \mathfrak{f}(z) \\  z' \prec z}}  (\mathcal{J}_0(z')- \nabla_\zeta) \right) R(z_0,z) \Psi^{-1}(z) =  \frac{\int_o^z \omega_{0,2}(\cdot, z_0)}{\Psi_z(z)}.
\end{equation}
Here, $o$ is the point in $\mathfrak{a}$ in the same component of $C$ as $z$, $\omega_{0,2}(z_1,z_2)$ was defined in \eqref{eq:w02arbit}, and for  $\mathfrak{c}(z)=1$ we understand the operator in bracket on the left-hand side as being the identity. Then, the following formula defines a kernel with property (i) and (ii):
\begin{equation}\label{eq:recker}
K^{\alpha}(z_0,z) = -\left(R(z_0,z) + \sum_{i = 1}^{r} (-1)^i \sum_{\substack{Z _0\subseteq \mathfrak{f}_{\succ}(z_0) \\ |Z_0| = i}} \left[ \overset{\searrow}{\mathcal{SP}}\left( \prod_{z_0' \in Z_0} \widehat{\mathcal{R}}_{z_0'} \right) R(\cdot, z) \right](z_0)\right),
\end{equation}
where $\mathfrak{f}_{\succ}(z) = \{z' \in \mathfrak{f}(z) \,\,|\,\,z' \succ z\}$. Furthermore, $K^{\alpha}(z_0,z)$ has a zero of degree $r$ when $z$ approaches a point in $\mathfrak{a}$.
\end{lem}
\begin{proof}
Let us start by considering an arbitrary $R(z_0,z)$, and assume that $u$ is a meromorphic $1$-form. We can write:
\[
\begin{split}
 \sum_{o \in \mathfrak{a}}& \Res_{z = o} R(z_0,z)[\widehat{\mathcal{Q}}^{(1)} u](z) \\
 & =  \sum_{o \in \mathfrak{a}} \Res_{z = o}  \sum_{\substack{z' \in \mathfrak{f}(z) \\ z' \succeq z}} R(z_0,z) \Psi^{-1}(z)  \overset{\nearrow}{\mathcal{SP}} \left(  \prod_{\substack{z'' \in \mathfrak{f}(z) \\  z'' \prec z'}}(\mathcal{J}_0(z'') + \nabla_\zeta) \right) u(z') \Psi_{z'}(z) \\
 & = \sum_{o \in \mathfrak{a}} \Res_{z = o}  \sum_{\substack{z' \in \mathfrak{f}(z) \\ z' \succeq z}}  u(z') \Psi_{z'}(z)  \overset{\searrow}{\mathcal{SP}} \left( \prod_{\substack{z'' \in \mathfrak{f}(z) \\  z'' \prec z'}}  (\mathcal{J}_0(z'')- \nabla_\zeta) \right) R(z_0,z) \Psi^{-1}(z) ,
\end{split}
\]
where in the last line we used integration by parts. We adopted the convention that the operator in bracket is the identity when $\mathfrak{c}(z') = 1$. We define
\begin{equation}
U_{z'}(z_0,z) := \overset{\searrow}{\mathcal{SP}} \left( \prod_{\substack{z'' \in \mathfrak{f}(z) \\  z'' \prec z'}}  (\mathcal{J}_0(z'')- \nabla_\zeta) \right) R(z_0,z) \Psi^{-1}(z),
\end{equation}
so that
\begin{equation}\label{eq:sss}
 \sum_{o \in \mathfrak{a}} \Res_{z = o} R(z_0,z)[\widehat{\mathcal{Q}}^{(1)} u](z) =\sum_{o \in \mathfrak{a}} \Res_{z = o}  \sum_{\substack{z' \in \mathfrak{f}(z) \\ z' \succeq z}}  u(z') \Psi_{z'}(z)  U_{z'}(z_0,z).
\end{equation}
Imagine we can choose $R(z_0,z)$ such that
\begin{equation}\label{eq:ass}
U_z(z_0,z) = \frac{\int_o^z \omega_{0,2}(\cdot, z_0)}{\Psi_z(z)},
\end{equation}
in other words it satisfies \eqref{eq:Uz0}. Then, for any $u$ satisfying \eqref{uusa}, the $z=z'$ term on the right-hand-side of \eqref{eq:sss} is evaluated directly to $u(z_0)$, and we obtain:
\[
\begin{split}
& \sum_{o \in \mathfrak{a}} \Res_{z = o} R(z_0,z)[\widehat{\mathcal{Q}}^{(1)} u](z)\\
 &= u(z_0) + \sum_{o \in \mathfrak{a}} \Res_{z = o}  \sum_{\substack{z' \in \mathfrak{f}(z) \\ z' \succ z}}  u(z') \Psi_{z'}(z)  U_{z'}(z_0,z) \\
&=u(z_0) + \sum_{o \in \mathfrak{a}} \Res_{z = o}  \sum_{\substack{z' \in \mathfrak{f}(z) \\ z' \succ z}}  u(z') \Psi_{z'}(z)  \overset{\searrow}{\mathcal{SP}} \left( \prod_{\substack{z'' \in \mathfrak{f}(z) \\  z \preceq  z'' \prec z'}}  (\mathcal{J}_0(z'')- \nabla_\zeta) \right) U_z(z_0,z) \\
& =u(z_0) + \sum_{o \in \mathfrak{a}} \Res_{z = o}  \sum_{\substack{z' \in \mathfrak{f}(z) \\ z' \succ z}} U_z(z_0,z ) \overset{\nearrow}{\mathcal{SP}} \left( \prod_{\substack{z'' \in \mathfrak{f}(z) \\  z \preceq  z'' \prec z'}}(\mathcal{J}_0(z'')+ \nabla_\zeta) \right) u(z') \Psi_{z'}(z),
\end{split}
\]
where in the last line we integrated by parts. The residue can be evaluated because of \eqref{eq:ass}, and comparing with Definition~\ref{d:Rhat} we get:
\begin{equation}\label{eq:tte}
\sum_{o \in \mathfrak{a}} \Res_{z = o} R(z_0,z)[\widehat{\mathcal{Q}}^{(1)} u](z) 
= u(z_0) +\sum_{\substack{z_0' \in \mathfrak{f}(z_0)}} \left[ \widehat{\mathcal{R}}_{z_0'} u \right](z_0).
\end{equation}
We have obtained a triangular system, as the only non-zero terms in the sum are for $z_0' \succ z_0$). This can be inverted to get a recursion kernel $K^{\alpha}(z_0,z)$ such that we have property (i), namely
\[
-\sum_{o \in \mathfrak{a}} \Res_{z = o} K^{\alpha}(z_0,z)[\widehat{\mathcal{Q}}^{(1)} u](z) 
= u(z_0) .
\]
The explicit inversion of the triangular system leads to \eqref{eq:recker}.

Now assume that $R(z_0,z)$ has a zero of degree $r$ when $z$ approaches a point $ o \in \mathfrak{a}$. Since the operator $\widehat{\mathcal{R}}$ does not change the order of vanishing at $\zeta \to 0$, we conclude that $K^{\alpha}(z_0,z)$ has a zero of degree $r$ at $z=o$, as required in property (ii). 
\end{proof}

Lemma \ref{l:kernel1} gives a general formula for the recursion kernel in terms of $R(z_0,z)$ satisfying \eqref{eq:Uz0}, but we still have to justify the existence of $R(z_0,z)$. We now construct it explicitly, which will also result in an explicit formula for the recursion kernel.

\begin{lem}\label{l:Ro}
Let $z =  \big( \begin{smallmatrix} a \\ \zeta \end{smallmatrix} \big)$ and assume $Q_a - Q_b \notin \alpha \mathbb{Z}$ for $a \neq b$. Then, there is a unique $R(z_0,z)$ satisfying \eqref{eq:Uz0} and vanishing at order $r$ when $z$ approaches points in $\mathfrak{a}$. It is given by:
\begin{equation}
\label{reformula} R(z_0,z) =\frac{\dd \zeta_0}{(\dd \zeta)^{r - 1}} \delta_{a, \mathfrak{c}(z_0)} \sum_{k \in \mathbb{Z}_{>0}}\frac{\zeta^{k+r-1 }}{\zeta_0^{k+1}} 
\sum_{i=1}^{a-1} \frac{1}{Q_a-Q_i +\alpha(a+k-i)  } \prod_{\substack{d \in [r] \\ d \neq a,i}} \frac{1}{Q_d - Q_i +\alpha(d-i) }.
\end{equation}
For $a=1$, the summation over $i$ is understood to become $1$, with the last product becoming a product over $d \in [r]$ with $d \neq 1$.
\end{lem}

\begin{proof}
We need to solve the inhomogeneous differential equation of order $\mathfrak{c}(z)$:
\begin{equation}\label{eq:toso}
\overset{\searrow}{\mathcal{SP}} \left( \prod_{\substack{z' \in \mathfrak{f}(z) \\  z' \prec z}}  (\mathcal{J}_0(z')- \nabla_\zeta) \right) R(z_0,z) \Psi^{-1}(z) =  \frac{\int_o^z \omega_{0,2}(\cdot, z_0)}{\Psi_z(z)}.
\end{equation}
Let $z = \big( \begin{smallmatrix} a \\ \zeta \end{smallmatrix} \big)$, thus $\mathfrak{c}(z) = a$. We recall that $\nabla_\zeta = \alpha\,\dd_{\zeta}$, and according to \eqref{eq:w02arbit} and Lemma~\ref{lem520} we have
\[
\begin{split}
\mathcal{J}_0(z) & = \frac{Q_{a} - \alpha}{\zeta} \dd \zeta, \\
\Psi^{-1}(z) & = \zeta^{-r+a-1+\frac{Q_{a}}{\alpha}},\\
\Psi_z(z) & = (\dd \zeta)^{r - a} \left( \prod_{j = a+ 1}^r (Q_j - Q_a + \alpha(j-a)) \right) \zeta^{1  - \frac{Q_{a}}{\alpha}},\\
\int_o^z \omega_{0,2}(\cdot, z_0) & = \delta_{a, \mathfrak{c}(z_0)}\dd \zeta_0 \sum_{k \in \mathbb{Z}_{>0}} \frac{\zeta^k}{\zeta_0^{k+1}} .
\end{split}
\]
We therefore have a differential equation with respect to $\zeta$, while $\zeta_0$ (or $z_0$) is a spectator variable. We start by looking at the homogeneous differential equation, dropping the differentials $\dd \zeta$ in the intermediate steps as they play no role here:
\[
\overset{\searrow}{\mathcal{SP}} \left( \prod_{\substack{z' \in \mathfrak{f}(z) \\  z' \prec z}}  (\mathcal{J}_0(z')- \nabla_\zeta) \right) u(z)  = \left( \frac{Q_{a-1} - \alpha}{\zeta} - \alpha \partial_\zeta \right) \cdots \left(\frac{Q_1-\alpha}{\zeta} - \alpha \partial_\zeta \right) u \big( \begin{smallmatrix} a \\ \zeta \end{smallmatrix} \big)= 0.
\]
For any $a \in [r]$, a basis of solutions is given by
\[
u_b \big( \begin{smallmatrix} a \\ \zeta \end{smallmatrix} \big) = \zeta^{b-2+\frac{Q_b}{\alpha}}, \qquad b \in [ a-1].
\]
By the method of variation of constants, the general solutions to the inhomogeneous differential equation
\[
\left( \frac{Q_{a-1} - \alpha}{\zeta} - \alpha \partial_\zeta \right) \cdots \left(\frac{Q_1-\alpha}{\zeta} - \alpha \partial_\zeta \right) v \big( \begin{smallmatrix} a \\ \zeta \end{smallmatrix} \big)= g\big( \begin{smallmatrix} a \\ \zeta \end{smallmatrix} \big)
\]
are given by
\[
v \big( \begin{smallmatrix} a \\ \zeta \end{smallmatrix} \big) = \frac{(-1)^{a-1}}{\alpha^{a-1}} \sum_{i=1}^{a-1} u_i  \big( \begin{smallmatrix} a \\ \zeta \end{smallmatrix} \big)  \int^{\zeta}  \frac{{\rm Wr}_i \big( \begin{smallmatrix} a \\ \eta \end{smallmatrix} \big) }{{\rm Wr} \big( \begin{smallmatrix} a \\ \eta \end{smallmatrix} \big)}\dd \eta,
\]
for arbitrary choices of integration constants, where
$$
{\rm Wr} \big( \begin{smallmatrix} a \\ \zeta \end{smallmatrix} \big) = \mathop{\det}_{i,j \in [a - 1]} \left( \partial_{\zeta}^{i-1} u_j \big( \begin{smallmatrix} a \\ \zeta \end{smallmatrix} \big) \right)
$$
is the Wronskian of the system, and ${\rm Wr}_i$ is the Wronskian with the $i$th column replaced by $(0,0,\ldots,0, g\big( \begin{smallmatrix} a \\ \zeta \end{smallmatrix} \big))$. The Wronskian can be calculated:
\[
\begin{split}
{\rm Wr} \big( \begin{smallmatrix} a \\ \zeta \end{smallmatrix} \big) & =\prod_{i=1}^{a-1} \zeta^{i-2+\frac{Q_i}{\alpha}} \prod_{j=1}^{a-1} \zeta^{1-j} \prod_{1 \leq i < j \leq a-1} \left(j-i + \frac{Q_j - Q_i}{\alpha} \right) \\
& = \zeta^{1-a+ \sum_{k=1}^{a-1} \frac{Q_k}{\alpha}}  \prod_{1 \leq c < d \leq a-1} \left( d-c + \frac{Q_d - Q_c}{\alpha} \right).
\end{split}
\]
As for ${\rm Wr}_i$, it can also be calculated:
\[
\begin{split}
{\rm Wr}_i \big( \begin{smallmatrix} a \\ \zeta \end{smallmatrix} \big) & = (-1)^{i+a-1} g\big( \begin{smallmatrix} a \\ \zeta \end{smallmatrix} \big) \prod_{\substack{k=1 \\ k\neq i}}^{a-1} \zeta^{k-2+\frac{Q_k}{\alpha}} \prod_{j=1}^{a-2} \zeta^{1-j}   \prod_{\substack{1 \leq c < d \leq a-1 \\ c,d \neq i}} \left( d-c + \frac{Q_d - Q_c}{\alpha} \right) \\
& = (-1)^{i+a-1} g\big( \begin{smallmatrix} a \\ \zeta \end{smallmatrix} \big) \zeta^{-i+1+  \sum_{k=1}^{a-1} \frac{Q_k}{\alpha} - \frac{Q_i}{\alpha}}   \prod_{\substack{1 \leq c < d \leq a-1 \\ c,d \neq i}} \left( d-c + \frac{Q_d - Q_c}{\alpha} \right) .
\end{split}
\]
Thus
\[
\frac{{\rm Wr}_i \big( \begin{smallmatrix} a \\ \zeta \end{smallmatrix} \big)}{{\rm Wr} \big( \begin{smallmatrix} a \\ \zeta \end{smallmatrix} \big)} =  \frac{(-1)^{a-1}g\big( \begin{smallmatrix} a \\ \zeta \end{smallmatrix} \big) \, \zeta^{a-i-\frac{Q_i}{\alpha}}}{\prod_{\substack{1 \leq d \leq a-1 \\ d \neq i}} \left(d-i + \frac{Q_d - Q_i}{\alpha} \right)}.
\]
Now let us substitute
\[
 g\big( \begin{smallmatrix} a \\ \zeta \end{smallmatrix} \big) = \frac{\int_o^z \omega_{0,2}(\cdot, z_0)}{\Psi_z(z)} = \delta_{a, \mathfrak{c}(z_0)}\sum_{k \in \mathbb{Z}_{>0}}  \frac{\zeta^{k+\frac{Q_a}{\alpha} - 1}}{\zeta_0^{k+1}} \frac{1}{\prod_{j = a+ 1}^r \big(Q_j - Q_a + \alpha(j-a)\big)},
\]
where we dropped the differentials. We get:
\[
\begin{split}
v \big( \begin{smallmatrix} a \\ \zeta \end{smallmatrix} \big) & = \delta_{a, \mathfrak{c}(z_0)} \frac{1}{\alpha} \sum_{i=1}^{a-1}\sum_{k \in \mathbb{Z}_{>0}}  \frac{\zeta^{i-2+\frac{Q_i}{\alpha}}}{\prod_{j = a+ 1}^r \big(Q_j - Q_a + \alpha(j-a)\big) \prod_{\substack{1 \leq d \leq a-1 \\ d \neq i}} \big(Q_d - Q_i+\alpha(d-i)\big) \zeta_0^{k+1}}  \\
& \qquad \qquad \times \int^{\zeta}  \eta^{a+k-i-1+\frac{Q_a-Q_i}{\alpha}}\dd \eta \\
& = \delta_{a, \mathfrak{c}(z_0)} \frac{1}{\prod_{j = a+ 1}^r (Q_j - Q_a + \alpha(j-a)) } \sum_{k \in \mathbb{Z}_{>0}}\frac{\zeta^{a+k-2+\frac{Q_a}{\alpha} }}{\zeta_0^{k+1}} \\
&\qquad \qquad \times   \sum_{i=1}^{a-1} \frac{1}{(Q_a-Q_i +\alpha(a+k-i))\prod_{\substack{1 \leq d \leq a-1 \\ d \neq i}} \left( Q_d - Q_i +\alpha(d-i) \right)  },
\end{split}
\]
where we have made a specific choice of integration constants to get the last line. Finally, we get $R(z_0,z) = v \big( \begin{smallmatrix} a \\ \zeta \end{smallmatrix} \big)  \Psi(z)$, where $\mathfrak{c}(z) = a$ (here we reinserted the differentials for the final result):
\[
 R(z_0,z) =\frac{\dd \zeta_0}{(\dd \zeta)^{1 - r}} \delta_{a, \mathfrak{c}(z_0)} \sum_{k \in \mathbb{Z}_{>0}}\frac{\zeta^{k+r-1 }}{\zeta_0^{k+1}} 
\sum_{i=1}^{a-1} \frac{1}{Q_a-Q_i +\alpha(a+k-i)  } \prod_{\substack{d \in [r] \\ d \neq a,i}} \frac{1}{Q_d - Q_i +\alpha(d-i) }.
\]
We in fact see that the choice of integration constants made is the only one for which $R(z_0,z)$ vanishes at order $r$ when $z$ approaches points of $\mathfrak{a}$.
\end{proof}

We are now ready to obtain an explicit formula for the recursion kernel $K^{\alpha}(z_0,z)$.
\begin{prop}\label{p:kerr}
The following kernel satisfies properties (i) and (ii):
\[
K^{\alpha}\big( \begin{smallmatrix} a \\ \zeta_0\end{smallmatrix}, \begin{smallmatrix} b \\ \zeta \end{smallmatrix} \big) = - \left(R\big( \begin{smallmatrix} a \\ \zeta_0\end{smallmatrix}, \begin{smallmatrix} b \\ \zeta \end{smallmatrix} \big) + \frac{\dd \zeta_0}{(\dd \zeta)^{r - 1}} \sum_{m=1}^{b-a} \sum_{b=a_m > \ldots > a_1>a} (-1)^m  \sum_{k \in \mathbb{Z}_{>0}} A_k(\mathbf{a}) \frac{\zeta^{k+r-1}}{\zeta_0^{k+1}}\right),
\]
when $b \geq a$, and $0$ otherwise, with
\begin{equation}
\label{eq:coeffi}
\begin{split}
A_k(\mathbf{a}) & = \prod_{\ell = 0}^{m-1} \frac{\prod_{j=a_\ell}^{a_{\ell+1}-1} (Q_j - Q_{a_\ell} + \alpha(j-a_\ell-k)) }{\prod_{j=a_\ell+1}^{a_{\ell+1}} (Q_j - Q_{a_\ell} + \alpha(j-a_{\ell})} \\
& \quad�\times  \sum_{i=1}^{a_m-1} \frac{1}{(Q_{a_m}-Q_i +\alpha(a_m+k-i))}\prod_{\substack{d \in [r] \\ d \neq i, a_m}} \frac{1}{ Q_d - Q_i +\alpha(d-i) } .
\end{split}
\end{equation}
\end{prop}

\begin{proof}
We need to write the action of the operator $\widehat{\mathcal{R}}$ from Definition \ref{d:Rhat} explicitly. (We drop the differentials for clarity.) Let $z = \big(\begin{smallmatrix} a \\ \zeta \end{smallmatrix} \big)$, and $z' = \big(\begin{smallmatrix} b \\ \zeta \end{smallmatrix} \big)$ with $b>a$. Substituting explicit expressions in the definition, we get:
\begin{equation*}
\begin{split}
& \quad \left[ \widehat{\mathcal{R}}_{z' }u \right](z) \\
& =   \frac{\zeta^{\frac{Q_a}{\alpha} - 1}}{\prod_{j=a+1}^b (Q_j-Q_a + \alpha(j-a))}
\left(\frac{Q_a - \alpha}{\zeta} + \alpha \partial_\zeta \right) \cdots \left( \frac{Q_{b-1} - \alpha}{\zeta} + \alpha \partial_\zeta \right) \cdot
 u\big( \begin{smallmatrix} b \\ \zeta \end{smallmatrix}\big) \zeta^{1+b-a-\frac{Q_a}{\alpha}}.
\end{split}
\end{equation*}
We now apply this operator to $R(z_0,z)$, in the variable $z_0$, with $\mathfrak{c}(z_0) = a$ and $\mathfrak{c}(z_0') = b$ as above. In fact, the only dependence in $\zeta_0$ in $R(z_0,z)$ is in the monomial $\zeta_0^{-(k+1)}$, which is summed over $k$. So let us calculate the action on this monomial first. We get:
\[
\left[ \widehat{\mathcal{R}}_{z_0'}  \zeta_0^{-(k+1)} \right] =   \frac{ \prod_{j=a}^{b-1} (Q_{j}-Q_a + \alpha(j-a-k) )}{\prod_{j=a+1}^b (Q_j-Q_a + \alpha(j-a))}  \zeta_0^{-(k+1)}.
\]
Thus, for any given sequence $\mathbf{a} = (a_0 < a_1 < \ldots< a_m)$, we define the coefficient \eqref{eq:coeffi}.
It follows that, if $z_0^0 < z_0^1 < \ldots< z_0^m$, with $\mathfrak{c}(z_0^i) = a_i$, we have:
\[
\left[ \widehat{\mathcal{R}}_{z_0^m}\, \cdots\, \widehat{\mathcal{R}}_{z_0^1} K^{\alpha}(\cdot, z) \right](z_0^0) = -\frac{\dd \zeta_0}{(\dd \zeta)^{r - 1}} \delta_{a_m, \mathfrak{c}(z)} \sum_{k \in \mathbb{Z}_{>0}} A_k(\mathbf{a}) \frac{\zeta^{k+r-1}}{\zeta_0^{k+1}}.
\]
The result of the Proposition follows by imposing that $a_m = b$ and $a_0 = a$, and summing over all possibilities.
\end{proof}

\begin{rem}
If $\alpha$ is set to a finite non-zero value as in the Nekrasov-Shatashvili regime, one needs to assume the non-resonance condition $Q_a - Q_b \notin \alpha \mathbb{Z}$ for any $a \neq b$ for the above expressions to be well-defined. It is possible  to treat  a resonant situation (but we will not do it) by introducing a function $\Psi(z)$ containing logarithms. On the other hand, if $\kappa$ is set to a non-zero value, as $\alpha = -\hbar^{\frac{1}{2}}(\kappa^{\frac{1}{2}} - \kappa^{-\frac{1}{2}})$ for finite $\kappa$, all expressions above should be thought as power series expansion in $\alpha \rightarrow 0$ and it suffices to assume $Q_a \neq Q_b$ (as already mentioned in Theorem~\ref{Whit1prop}).
\end{rem}

\subsubsection{The $r = 2$ case}
\label{r2kappagen}
For $r = 2$, the topological recursion in Theorem~\ref{lemTRAarbit} simplifies drastically. First, the only term in the summation is $i=2$, and
\[
\left[ \widehat{\mathcal{Q}}^{(2)} \Omega'^\alpha_{g,2,n}(\cdot\,; w_{[n]}) \right](z) = \Omega'^\alpha_{g,2,n}(z;w_{[n]}).
\]
Second, the equation $(\mathsf{J}^1(\zeta) + \mathsf{J}^2(\zeta))\mathcal{Z}= O(\zeta^{-1})$ is part of the Airy structure, and it results in
$$
\omega^{\alpha}_{g,n+1}\big(\begin{smallmatrix} 1\\ \zeta \end{smallmatrix},z_{[n]}\big) + \omega^{\alpha}_{g,n+1}\big(\begin{smallmatrix}2 \\ \zeta \end{smallmatrix},z_{[n]}\big) = O\bigg(\frac{\dd \zeta}{\zeta} \bigg).
$$
Since $\omega^{\alpha}_{g,n}$ does not contain positive powers of $\zeta$ --- owing to the definition of the basis of differentials $\dd\xi_{-k}^a$ in \eqref{defdifferentials} --- this yields an equality
\begin{equation}
\label{tiurugngiu}\omega_{g,n + 1}^{\alpha}\big(\begin{smallmatrix} 2 \\ \zeta \end{smallmatrix},z_{[n]}\big) = - \omega^{\alpha}_{g,n + 1}\big(\begin{smallmatrix} 1 \\ \zeta \end{smallmatrix},z_{[n]}\big).
\end{equation}
We therefore only need to know
$$
\omega^{\alpha}_{g,n}\big(\begin{smallmatrix} 2 & \cdots & 2 \\ \zeta_1 & \cdots & \zeta_n \end{smallmatrix}\big) =: \omega_{g,n}^{\alpha,(2)}(\zeta_1,\ldots,\zeta_n).
$$
Since the recursion kernel $K^{\alpha}\big( \begin{smallmatrix}a \\ \zeta_0\end{smallmatrix}, \begin{smallmatrix} b \\ \zeta\end{smallmatrix} \big)$ is non-vanishing only if $b \geq a$, we see that we only need
\[
K^{\alpha}\big( \begin{smallmatrix}2 \\ \zeta_0\end{smallmatrix}, \begin{smallmatrix} 2 \\ \zeta\end{smallmatrix} \big)  = - R\big( \begin{smallmatrix}2 \\ \zeta_0\end{smallmatrix}, \begin{smallmatrix} 2 \\ \zeta\end{smallmatrix} \big)   = - \frac{\dd \zeta_0}{\dd \zeta} \sum_{k \in  \mathbb{Z}_{>0}} \frac{\zeta^{k+1}}{\zeta_0^{k+1}} \frac{1}{Q_2-Q_1 + \alpha(k+1)} ,
\]
which we evaluated using Proposition \ref{p:kerr} and Lemma \ref{l:Ro}. This leads to
\[
\begin{split}
& \quad \omega_{g,n + 1}^{\alpha,(2)}(\zeta_0,\zeta_{[n]}) \\
& = \mathop{{\rm Res}}_{\zeta = 0}\,\dd \zeta\, K^{\alpha}\big( \begin{smallmatrix}2 \\ \zeta_0\end{smallmatrix}, \begin{smallmatrix} 2 \\ \zeta\end{smallmatrix} \big) \bigg(\omega^{\alpha,(2)}_{g - 1,n + 1}(\zeta,\zeta,\zeta_{[n]}) + \sum_{\substack{h + h' = g \\ J \sqcup J' = [n]}}^{{\rm no}\,\,\omega_{0,1}} \omega^{\alpha,(2)}_{h,1+|J|}(\zeta,\zeta_{J})\omega^{\alpha,(2)}_{h',1+|J'|}(z,z_{J'})\bigg) \\
& \quad + \frac{T\delta_{g,\frac{r}{2}}\delta_{n,0}}{Q_1 - Q_2}\,\frac{\dd \zeta_0}{\zeta_0^2}.
\end{split} 
\]

\subsection{Non-commutative topological recursion}
 
\label{NCTR}
A key role in obtaining the abstract loop equations in Lemma~\ref{l:ALEarbit} and the topological recursion in Theorem~\ref{lemTRAarbit} is played by the $D$-module with regular singularities on the base curve $C_0$ annihilated by
$$
\hat{\mathcal{Y}} = \prod_{a = 1}^r \left( \omega_{0,1}\big(\begin{smallmatrix} a\\ \zeta \end{smallmatrix} \big) + \nabla_{\zeta} \right).
$$
This can be viewed as a non-commutative curve (sometimes called quantum curve). 

As in Section~\ref{Generald}, the topological recursion formula in the form of Theorem~\ref{lemTRAarbit} will apply verbatim for the NS correlators of the Airy structures obtained from the ones of Theorem~\ref{Whit1prop} via general dilaton shifts and changes of polarization.  Those shifts and changes of polarization correspond to modifying the spectral curve as follows:
\begin{equation}
\label{thegeneralsp}\begin{split}
\omega_{0,1}\big(\begin{smallmatrix} a \\ \zeta \end{smallmatrix}\big) & = \frac{Q_a}{\zeta} \dd \zeta + \sum_{k > 0} F_{0,1}\big[\begin{smallmatrix} a \\ -k \end{smallmatrix}\big]\,\zeta^{k - 1} \dd \zeta, \\
\omega_{0,2}\big(\begin{smallmatrix} a_1 & a_2 \\ \zeta_1 & \zeta_2 \end{smallmatrix}\big) & = \frac{\delta_{a_1,a_2}\,\dd \zeta_1 \dd \zeta_2}{(\zeta_1 - \zeta_2)^2} + \sum_{k_1,k_2 > 0} F_{0,2}\big[\begin{smallmatrix} a_1 & a_2 \\ k_1 & k_2 \end{smallmatrix}\big]\,\zeta_1^{k_1 - 1}\zeta_2^{k_2 - 1} \dd \zeta_1 \dd \zeta_2, \\
\dd \xi_{-k}^a\big(\begin{smallmatrix} b \\ \zeta \end{smallmatrix}) & = \frac{\delta_{a,b}\dd \zeta}{\zeta^{k + 1}} + \sum_{l > 0} \frac{F_{0,2}\big[\begin{smallmatrix} a & b \\ -k & -l \end{smallmatrix}\big]}{k}\,\zeta^{l - 1}\dd \zeta.
\end{split}
\end{equation}
By allowing $Q_a$ to be $\alpha$ dependent, we can  set the zero-mode $\mathsf{J}_0^a = Q_a - \alpha$ to be any value.

We can find a basis $\Psi\big(\begin{smallmatrix} a \\ \zeta \end{smallmatrix} \big)$ of the $D$-module of the form
$$
\Psi\big(\begin{smallmatrix} a \\ \zeta \end{smallmatrix} \big) = \zeta^{a - \frac{Q_a}{\alpha}} \cdot \Xi_a(\zeta),\qquad a \in [r],
$$
where $\Xi_a(\zeta) = 1 + O(\zeta)$ are formal power series in $\zeta$. Likewise, the right-hand side in \eqref{eq:psi} is now multiplied by a formal power series in $\zeta$ with constant term $1$. We can construct the Cauchy kernel
$$ 
G\big(z_0,z)= \int^{z}_o \omega_{0,2}(\cdot,z_0) \,\,\,\mathop{=}_{|x(z)| < |x(z_0)|}\,\,\, \sum_{\substack{k > 0 \\ a \in [r]}} k\,\dd \xi_{-k}^a(z_0) \dd \xi_k^a(z),
$$
where we recall that $\dd \xi_k^a \big(\begin{smallmatrix} b \\ \zeta \end{smallmatrix}\big) = \delta_{a,b} \frac{\zeta^k}{k}$. The auxiliary operator $\widehat{\mathcal{R}}$ is still given by Definition~\ref{d:Rhat}. We can still find a unique solution $R(z_0,z)$ solving \eqref{eq:Uz0} and vanishing at order $r$ when $z \rightarrow \mathfrak{a}$. It is no longer defined as \eqref{reformula}: for $z = \big(\begin{smallmatrix} a \\ \zeta \end{smallmatrix}\big)$, the factor $\frac{\dd\zeta_0}{\zeta_0^{k + 1}}$ is replaced by $\dd \xi_{-k}^{a}(z_0)$ and the $\zeta^{k + r - 1}$ is multiplied by a formal power series in $\zeta$  of the form $1 + O(\zeta)$. Then, the assumptions of Lemma~\ref{l:kernel1} are satisfied  and  hence we obtain  a recursion kernel using the formula \eqref{eq:recker}.

Thus, we can  propose \eqref{TRAarbit} as a general definition of the topological recursion taking as input not a spectral curve, but a $D$-module of degree $r$ with regular singularities on a local curve. There are two types of assumptions we can impose for this definition to make sense:
\begin{itemize}
\item If we consider $\alpha$ as a formal parameter near $0$, we should assume that $Q_a \neq Q_b + O(\alpha)$ for any $a \neq b$.
\item If we set $\alpha$ to a finite non-zero value, we should rather assume the non-resonance condition $Q_a - Q_b \notin \alpha\mathbb{Z}$.
\end{itemize}
Then, the fact that it corresponds to an Airy structure guarantees symmetry of the $\omega_{g,n}^{\alpha}$, which would be non-trivial to prove from \eqref{TRAarbit} itself.

Topological recursion for non-commutative curves first appeared in the context of $\beta$-matrix models (with identification $\kappa = \frac{\beta}{2}$ and $\hbar^{-1} = N$ the size of the matrix) \cite{CE06,BEMPf}, and was first developed as an abstract construction \cite{EMq0,CEMq1,CEMq2} for degree $2$ (non-commutative) hyperelliptic curves. It was further studied in \cite{BelEyn0,BelEyn} for more general $D$-modules on global curves. One particularity in this work is that the global geometry allowed  the authors to move the residue to surround the Stokes line, or even squeeze the contour further to the (infinitely many) zeroes of $\Psi\big(\begin{smallmatrix} a \\ \zeta\end{smallmatrix}\big)$ (called Bethe roots) --- in this form, with their assumptions,  the $i > 2$ terms in the recursion formula vanish. In comparison, our work bypasses the (interesting) global difficulties and proposes an appropriate local definition. While our formula is simpler as we only take residues at the regular singularities as opposed to extended contours or a sum of residues at infinitely many points, it is more intricate, for $ r > 2 $, due to the non-trivial differential operators $\widehat{\mathcal{Q}}^{(i)}$ from Definition~\ref{d:QII}. For $r = 2$ our formula has the same structure\footnote{Yet, among the simplifications mentioned in Section~\ref{r2kappagen}, the only one that does not apply for general change of polarization is the equality \eqref{tiurugngiu}.} as the commutative case. Only the recursion kernel changes, and it is now obtained as a solution of a differential equation which, in fact, reproduces the recursion kernel of \cite{CEMq1,CEMq2,BelEyn0}. For $r > 2$ our formula is new and can be viewed as the non-commutative generalization of the higher topological recursion of \cite{BHLMR,BE2}.

\section{Gaiotto vectors of type B}

\label{Sec:AiryB}

The connection between Gaiotto vectors and Whittaker vectors given in Theorem  \ref{ThmBFN} is only valid when $G$ is a simply-laced Lie group. However, when $G$ is of type $B_r$, \emph{i.e} $G = {\rm SO}(2r+1)$, the proposal from physics is that the Gaiotto vector should be a Whittaker vector for an appropriate twisted module of $\mathcal{W}^{\mathsf{k}}(\mathfrak{gl}_{2r})$ (see \cite{KMST} for the physics proposal). The twist in question is a folding symmetry of the Dynkin diagram of type A that gives the Dynkin diagram of type B.

We can repeat the steps of Section~\ref{Sec:AiryA} to construct Gaiotto vectors of type B as partition functions of Airy structures realized as $\mathcal{W}^{\mathsf{k}}(\mathfrak{gl}_{2r})$-modules, and the steps of Section~\ref{Sec:TRA} to obtain the corresponding spectral curve topological recursion. For simplicity, in this section we will restrict to the self-dual level $\kappa =1$.

\subsection{The Airy structure}

In this subsection, we assume that $r \geq 2$. We consider the $ \mathcal W $-algebra $ \mathcal  W^{\mathsf{k}}(\mathfrak{gl}_{2r}) \subset \mathcal{H}(\mathfrak{gl}_{2r})$ at the self-dual level $\kappa=1$. We choose the canonical orthonormal basis $( \chi^a)_{a = 1}^{2r}$ of the Cartan subalgebra $\mathfrak{h} = \mathbb{C}^{2r}$, and define the automorphism $\sigma$ which acts as follows,
$$
\forall a \in [r],\qquad \sigma(\chi^a) = - \chi^{2r + 1 - a}.
$$ 
We choose a diagonal basis for the action, namely
$$
v^{\pm,a} := \chi^a \mp \chi^{2r + 1 - a},\qquad a \in [r],
$$ such that $\sigma(v^{\pm,a}) = \pm v^{\pm,a}$. The inverse base change is given as follows:
$$
\forall a \in [r],\qquad \chi^a = \tfrac{1}{2}(v^{-,a} + v^{+,a}),\qquad \chi^{2r + 1 - a} = \tfrac{1}{2}(v^{-,a} - v^{+,a}).
$$
The diagonal basis has the property that
\begin{equation}
	\label{sclarB}
	\forall \gamma,\gamma' \in \{\pm 1\},\qquad \langle v^{\gamma,  a},v^{\gamma',a'}\rangle = 2\delta_{\gamma,\gamma'}\delta_{a,a'}.
\end{equation}
This automorphism $ \sigma $ extends to an automorphism of the Heisenberg algebra $ \mathcal{H}(\mathfrak{gl}_{2r}) $, using which we construct a $\sigma$-twisted $\mathcal{H}(\mathfrak{gl}_{2r})$-module $\mathcal{M}^{\hbar} = \mathbb C \big[\!\big[(x^a_m)_{a \in [r],\,\,m > 0}\big]\!\big][\![\hbar^{1/2}]\!] $ as follows. We  have the twisted state-field correspondence,
\begin{equation}\label{eq:vfield}
	Y_{\sigma}(v_{-1}^{+,a} \ket{0},z) = \sum_{m \in \mathbb{Z}} \frac{J_{2m}^a}{z^{m + 1}} =: v^{+,a}(z) ,\qquad Y_{\sigma}(v_{-1}^{-,a},\ket{0},z) = \sum_{m \in \mathbb{Z} + \frac{1}{2}} \frac{J_{2m}^a}{z^{m + 1}} =: v^{-,a}(z),
\end{equation}
where the representation of the rescaled Heisenberg modes $\mathsf{J}^a_m$ is\footnote{For general $\kappa$ one would add $ -\hbar^{\frac{1}{2}}\alpha_0 $ to the zero-mode, as in \eqref{eq:hrep}.}, for $a \in [r]$ and $m  \in \mathbb{Z}$,
\begin{equation}\label{eq:heis}
	\mathsf{J}^a_{m} = \left\{
	\begin{array}{lr}
		\hbar \frac{\partial}{\partial x^a_m}&  m > 0,\\
		-m x^{a}_{-m}& m < 0, \\
		Q_a& m = 0, 
	\end{array} \right.
\end{equation}
with $Q_1, \ldots, Q_r \in \mathbb{C}$.

Since $\kappa=1$, the $ \mathcal{W}^{\mathsf{k}}(\mathfrak{gl}_{2r}) $-algebra generators obtained from the quantum Miura transform have the following transformation properties,
$$
\sigma\big(e_i(\chi^1_{-1},\ldots,\chi^{2r}_{-1}) \ket{0}\big) = (-1)^i e_i(\chi^1_{-1},\ldots,\chi^{2r}_{-1}) \ket{0},
$$
and hence we get the mode expansion
\begin{equation}\label{eq:modestwist}
	i \in [2r],\qquad W^i(z) := Y_\sigma\big(e_i(\chi^1_{-1},\ldots,\chi^{2r}_{-1}) \ket{0}, z\big) = \sum_{k \in \mathbb Z + \epsilon(i)} \frac{W^i_k}{z^{k+i}},
\end{equation}
where $\epsilon(i)  = 0$ if $i$ is even and $\epsilon(i) = \frac{1}{2}$ if $i$ is odd. The twisting by $ \sigma $ changes the OPE of the Heisenberg fields, which in turn modifies the expression for the fields $W^i(z)$ in terms of the Heisenberg fields (\eqref{Wcritlev} in the untwisted case at self-dual level) in accordance with the product formula \cite[Proposition 3.2]{BakalovMilanov2}.  However, for the rescaled fields $\mathsf{W}^i(z)$, it is fairly easy to see from the product formula\footnote{Examples of explicit computations using the product formula can be found in \cite[Section 4.1.2]{BBCCN18}.} that the (quantum) ``correction terms'' come with powers of $ \hbar $.  (We recall the definition of the rescaled modes in $\mathcal{A}^{\hbar}$ from Lemma \ref{l:Ahbar}.)

Our goal is to construct a Whittaker vector for the rescaled positive modes $\mathsf{W}^i_k$ using Airy structures. More precisely, we would like to construct a state $\ket{w}$ that is annihilated by all positive modes $\mathsf{W}^i_k$, $i \in [2 r]$, $ k \in ( \mathbb{Z}_{>0}- \epsilon(i))$, except for the mode $\mathsf{W}^{2r-1}_{1/2}$ acting on $\ket{w}$ as a constant. According to \cite{KMST}, $\ket{w}$ should then correspond to the Gaiotto vector of type B.

We follow the strategy outlined in Section \ref{s:WhitAiry}:
\begin{enumerate}
	\item We need to show that $S' = \{\mathsf{W}^{2r-1}_{1/2}\}$ is an extraneous subset of the set of positive modes $S_+$.
	\item We need to construct an Airy structure differential representation for the positive modes $\mathsf{W}^i_k$. 
\end{enumerate}
It then follows that the partition function for the shifted Airy structure is the sought-for Whittaker vector $\ket{w}$.

Unfortunately, at this stage we do not have a proof of the first step. It is clear by Proposition 3.14 in \cite{BBCCN18} (see also footnote \ref{f:subalgebra}) that the set of positive modes $S_+$ satisfies the subalgebra condition in Definition \ref{d:subVOA}). However, we do not have an analog of Proposition \ref{thecoth} for the existence of the extraneous subset for the present twisted module. We thus formulate the first condition as a conjecture:

\begin{conj}
	\label{theconjB} Let $S_+ = \{\mathsf{W}^i_k\,\,|\,\,i\in [2 r],\,\, k \in \left (\mathbb{Z}_{>0}-\epsilon(i) \right )\}$ be the subset of positive modes of the strong generators in \eqref{eq:modestwist}, and $S' = \{\mathsf{W}^{2r-1}_{1/2}\} \subset S_+$. Then $S'$ is an extraneous subset, according to the terminology of  Definition \ref{d:extra}.
\end{conj}

\begin{thm}
	\label{BAiry} Assume that Conjecture~\ref{theconjB} holds. Let $T \in \mathbb{C}$ and $ Q_1, \ldots, Q_r \in \mathbb{C}^*$  such that  $ Q_a \neq \pm Q_b $ for any distinct $ a,b \in [r] $. Let $\mathsf{W}^i_m$ be the differential operators associated to the rescaled modes $\mathsf{W}^i_m$ of the strong generators in \eqref{eq:modestwist} obtained via \eqref{eq:heis}.Then, the family of differential operators: 
	\[
	\forall i \in [2r], \, m \in \left( \mathbb{Z}_{>0}-\epsilon(i) \right), \qquad \overline{\mathsf{W}}^i_m := \mathsf{W}^i_m - \hbar^{\frac{2r - 1}{2}} T\delta_{i,2r-1} \delta_{m,\frac{1}{2}} 
	\] forms an Airy structure.
\end{thm}
\begin{proof}[Proof of Theorem~\ref{BAiry}]
	By Proposition 3.14 of \cite{BBCCN18} (see also footnote \ref{f:subalgebra}), we know that the subset of positive modes $S_+$ satisfies the subalgebra condition in Definition \ref{d:Airy}. Assuming that Conjecture~\ref{theconjB} holds, the subset $S' = \{\mathsf{W}^{2r-1}_{1/2} \} \subset S_+$ is extraneous. Therefore the subalgebra condition of Definition \ref{d:Airy} continues to hold for the shifted modes $\overline{\mathsf{W}}^i_m$. 
	
	All that remains is to verify that the degree condition in Definition \ref{d:Airy} is satisfied, for either the $\mathsf{W}^i_n$s or the $\overline{\mathsf{W}}^i_n$s --- as the shift has degree $\geq 2$ according to the grading \eqref{eq:grading} and hence does not affect the degree condition.
	
	First, as we have discussed below equation~\eqref{eq:modestwist}, the correction terms in $ \mathsf{W}^i(z) $ due to the twisting come with powers of $ \hbar $. Such terms  do not affect the terms of degrees zero and one, and we can safely ignore them in the following proof. As usual, degree zero terms cannot appear in any of the positive modes.  To analyze the degree one terms, we  notice that the zero-modes only occur in $\mathsf{J}^{+,a}(z)$.  We introduce the notation $Q_{2r + 1 -b} = - Q_b$ for $b \in [r]$. For $m \in \mathbb{Z}_{\geq 0} + \frac{1}{2}$, the contribution of $\mathsf{J}_{2m}^b$ must come from a factor $v^{-}$ in $\mathsf{W}^{2i - 1}_{m}$ for some $i \in [r]$, and we find that the projection to the degree one terms of the modes $\mathsf{W}^{2i-1}_m$ is:
	$$
	\pi_1(\mathsf{W}^{2i - 1}_m) = \sum_{b = 1}^{r} \mathsf{J}_{2m}^{b} 2^{-(2i - 1)}\,\big(e_{2i - 2}(Q_1,\ldots,\widehat{Q_b},\ldots,Q_{2r}) + e_{2i - 2}(Q_1,\ldots,\widehat{Q_{2r + 1 - b}},\ldots,Q_{2r})\big).
	$$
	This can be simplified by a generating series
	$$
	\frac{\sum_{i = 1}^r t^{r - i} \pi_1(\mathsf{W}^{2i - 1}_m)}{\prod_{a = 1}^r (t - \frac{Q_a^2}{4})} = \sum_{b = 1}^{r} \frac{\mathsf{J}_{2m}^b}{t - \frac{Q_b^2}{4}}.
	$$
	Under the assumption that $Q_a \neq \pm Q_b$ for $a,b \in [r]$ distinct, taking the residue at $t = \frac{1}{4}Q_a^2$ yields
	$$
	\forall a \in [r],\qquad \mathsf{J}_{2m}^a = \sum_{i = 1}^r \frac{2^{2i -2}\,Q_a^{2(r - i)}\,\pi_1(\mathsf{W}_{2m}^{2i - 1})}{\prod_{b \neq a} (Q_a^2 - Q_b^2)}.
	$$ 
	Likewise, for $m \in \mathbb{Z}_{> 0}$, the contribution of $\mathsf{J}_{2m}^b$ must come from a factor $v^+$ in $\mathsf{W}^{2i}_m$ for some $i \in [r]$, and we have
	$$
	\pi_1(\mathsf{W}^{2i}_m) = \sum_{b = 1}^r \mathsf{J}_{2m}^b 2^{-2i} \big(e_{2i - 1}(Q_1,\ldots,\widehat{Q_b},\ldots,Q_{2r}) - e_{2i - 1}(Q_1,\ldots,\widehat{Q_{2r + 1 - b}},\ldots,Q_{2r})\big),
	$$
	and so
	$$
	\frac{\sum_{i = 1}^r t^{r - i} \pi_1(\mathsf{W}^{2i}_m)}{\prod_{a = 1}^r (t - \frac{Q_a^2}{4})} =  -\sum_{b = 1}^r \frac{Q_b}{2}\,\frac{\mathsf{J}_{2m}^b}{t - \frac{Q_b^2}{4}} .
	$$
	Taking the residue at $t = \frac{1}{4}Q_b^2$ we get
	$$
	\mathsf{J}_{2m}^b = - \sum_{i = 1}^r \frac{2^{2i - 1}Q_b^{2r - 2i - 1}\,\pi_1(\mathsf{W}_{2m}^{2i})}{\prod_{a \neq b} (Q_b^2 - Q_a^2)},
	$$ where we use the condition that $ Q_a \neq 0 $ for all $ a \in [r] $.
	Hence, the family of operators
	\begin{equation}
		\label{HnormB}\begin{split}
			\mathsf{H}_{2m}^a & = -\sum_{i = 1}^r \frac{2^{2i - 1}Q_a^{2r - 2i - 1}\,\overline{\mathsf{W}}_{2m}^{2i}}{\prod_{b \neq a} (Q_a^2 - Q_b^2)} \qquad m \in \mathbb{Z}_{> 0},\quad a \in [r],  \\
			\mathsf{H}_{2m}^a & = \sum_{i = 1}^r \frac{2^{2i - 2}Q_a^{2r - 2i}\,\overline{\mathsf{W}}_{2m}^{2i - 1}}{\prod_{b \neq a} (Q_a^2 - Q_b^2)}  \qquad m \in \mathbb{Z}_{\geq 0} + \tfrac{1}{2},\quad a \in [r],
		\end{split}
	\end{equation}
	yields an Airy structure in normal form.
\end{proof}

\begin{rem}
	The proof of Lemma~\ref{lem:prophbar} on the $T$ dependence and vanishings in the partition function applies without much modification to this Airy structure. Therefore, the result of Corollary~\ref{corresum} is valid and we can apply Proposition~\ref{instton}. It shows that upon setting $\Lambda^{r} = T\hbar^{\frac{2r - 1}{2}}$, the instanton partition function takes the form
	$$
	\mathfrak{Z} = \exp\bigg(\sum_{h \geq 0} \hbar^{h - 1}\,\mathfrak{F}_h\bigg)
	$$
	and $\mathfrak{F}_h \in \mathbb{Q}(Q_1,\ldots,Q_r)[\![\Lambda^r]\!]$ is computed by a sum of Feynman graphs (which is finite at every given order in $\Lambda^r$).
\end{rem}

\begin{rem}
	We note that the proposal of \cite{KMST} for instanton partition functions of type $   B $ is valid at any level. In this case, the usual basis of generators for the  $ \mathcal W (\mathfrak{gl}_{2r}) $-algebra given by the quantum Miura transformation is not diagonal. In order to get an Airy structure in these cases, one needs to find an appropriate diagonal basis, and then apply the construction of this section. We restricted to the self-dual level $ \kappa = 1 $ in order to avoid this additional complication.
\end{rem}

\subsection{Spectral curve description}

\label{Sec:TRB}

In this section, we will rewrite the recursion on the coefficients $F_{g,n}$ of the partition function of the Airy structure of Theorem~\ref{BAiry} as a spectral curve topological recursion for the associated correlators $ \omega_{g,n} $. Although one could carry out this computation in a very similar style to \cite[Section 4]{BKS}, for the sake of uniformity, we adopt the (slightly different) presentation and notations of Section~\ref{SelfdualTRspec}. In fact, the calculations are very similar to Section~\ref{SelfdualTRspec} and we will only stress the main differences.   

\subsubsection{ Spectral curve}

Following the recipe outlined at the start of Section~\ref{Sec:TRA}, we first define the spectral curve $ (\tilde{C},\tilde{x},\omega_{0,1},\omega_{0,2}) $. As we are at the self-dual level  $\kappa = 1$, we set  $ \omega_{\frac{1}{2},1} = 0 $ --- see \eqref{eq:121}.
\begin{itemize}
	\item  We consider $r$ copies of a formal disk $C^a = {\rm Spec }\, \mathbb{C}[\![\zeta]\!]$, and their union $C = \bigsqcup_{a = 1}^r Ca$. We also introduce $ \tilde{C} :=  \bigsqcup_{a = 1}^r \operatorname{Spec} \mathbb{C}[\![\tilde{\zeta}]\!]  $, and denote points on $\tilde{C} $ by $z = \big( \begin{smallmatrix} a \\ \tilde{\zeta} \end{smallmatrix} \big)$.
	
	\item Recall that we have the unramified covering from Section~\ref{SelfdualTRspec} $ x : C \to C_0 = \operatorname{Spec} \mathbb{C}[\![\zeta]\!] $, which is defined as $ x\big(\begin{smallmatrix} a \\ \zeta \end{smallmatrix}\big) = \zeta$.  We introduce a branched covering $ \pi : \tilde{C} \to C  $ which acts as $ \pi\big( \begin{smallmatrix} a \\ \tilde{\zeta} \end{smallmatrix} \big) = \big( \begin{smallmatrix} a \\ \tilde{\zeta}^2 \end{smallmatrix} \big) $. Finally, we define our branched covering of interest as 
	\[
		 \tilde{x} :=  x \circ \pi 	 : \tilde{C} \to C_0, 
	\] 
	 The set of ramification points is $\mathfrak{a} := \tilde{x}^{-1}(0)$ which  is identified with $[r]$. In contrast to the spectral curve of Section~\ref{SelfdualTRspec}, the map $ \tilde{x} $ in the present situation is a ramified covering.
	
	\item  For $k \in \mathbb{Z}$ and $a \in [r]$ we define differentials $\dd\xi^a_{k},\dd\xi^*_k \in H^0(\tilde{C},K_{\tilde{C}}(* \mathfrak{a}))$ by
	\begin{equation}
		\dd\xi^a_k\big(\begin{smallmatrix} b \\ \tilde{\zeta} \end{smallmatrix}\big) = \delta_{a,b} \tilde{\zeta}^{k - 1}\dd \tilde{\zeta},\qquad \dd\xi^*_k \big(\begin{smallmatrix} b\\ \tilde{\zeta} \end{smallmatrix}\big) := \sum_{a=1}^r \dd\xi^a_{k}\big(\begin{smallmatrix} b\\ \tilde{\zeta} \end{smallmatrix}\big) = \tilde{\zeta}^{k-1} \dd \tilde{\zeta}.
	\end{equation}
	Then the $1$-form $\omega_{0,1}$ is given by:
	\begin{equation}
		\omega_{0,1} \big (\begin{smallmatrix} b\\ \tilde{\zeta} \end{smallmatrix} \big) = \sum_{a \in [r]} F_{0,1}\big[ \begin{smallmatrix}a \\ 0 \end{smallmatrix} \big] \dd \xi^a_0\big (\begin{smallmatrix} b\\ \tilde{\zeta} \end{smallmatrix} \big)  =  \frac{Q_b}{\tilde{\zeta}} \dd \tilde{\zeta},
	\end{equation}
	with $Q_1,\ldots,Q_r \in \mathbb{C^*}$. 
	\item As before, the bilinear differential $\omega_{0,2}$ is given by:
	\begin{equation}\label{eq:w02B}
		\omega_{0,2}\big(\begin{smallmatrix} b_1 & b_2 \\ \tilde{\zeta}_1 & \tilde{\zeta}_2 \end{smallmatrix}\big)  = \frac{\delta_{b_1,b_2}\dd \tilde{\zeta}_1\dd \tilde{\zeta}_2}{(\tilde{\zeta}_1 - \tilde{\zeta}_2)^2}.
	\end{equation}
\end{itemize}

Let us introduce the following notation for convenience. For $z \in \tilde{C}$, we define $\mathfrak{f}(z) := \tilde{x}^{-1}(\tilde{x}(z))$ as the set of preimages of $\tilde{x}(z)$, and $\mathfrak{f}'(z) := \mathfrak{f}(z) \setminus \{z\}$ as the preimages of $\tilde{x}(z)$ with the original $z$ excluded. More precisely, we have
\begin{equation}\label{eq:fB}
	\mathfrak{f}\big(\begin{smallmatrix} a \\ \tilde{\zeta} \end{smallmatrix}\big) = \bigsqcup_{a = 1}^r \big\{\big(\begin{smallmatrix} a \\ \tilde{\zeta} \end{smallmatrix}\big),\big(\begin{smallmatrix} a \\ -\tilde{\zeta} \end{smallmatrix}\big)\big\}.
\end{equation}
We will also need to work with
\[
\mathfrak{f}_-\big(\begin{smallmatrix} a \\ \tilde{\zeta} \end{smallmatrix}\big) := \bigsqcup_{a = 1}^r \big\{\big(\begin{smallmatrix} a \\ -\tilde{\zeta} \end{smallmatrix}\big)\big\},
\]
and if $Z \subseteq \mathfrak{f}(z)$, we shall denote $Z_- := Z \cap \mathfrak{f}_-(z)$. We  introduce the  notation 
\begin{equation}
	\varepsilon_z(z') = \left\{
	\begin{array}{lr}
		1 & \text{if } z' \in \, \mathfrak{f}_-(z),\\
		0 & \text{otherwise}.
	\end{array} \right.
\end{equation}

We define correlators $\omega_{g,n}$ indexed by $n \in \mathbb{Z}_{>0}$ and $g \in \frac{1}{2} \mathbb{Z}_{\geq 0 }$ in terms of the partition function $\mathcal{Z}$ for the Airy structure of Theorem~\ref{BAiry} exactly as before:
\[
\omega_{g,n}(z_1,\ldots,z_n) = \sum_{\substack{a_1,\ldots,a_n \in [r] \\ k_1,\ldots,k_n \in \mathbb{Z}_{>0}}} F_{g,n}\big[\begin{smallmatrix} a_1 & \cdots & a_n \\ k_1 & \cdots & k_n \end{smallmatrix}\big]\,\prod_{l = 1}^n \dd \xi_{-k_l}^{a_l}(z_i).
\]

\subsubsection{Abstract loop equations}

We will turn the differential constraints on the partition function $ \mathcal Z $ associated to the Airy structure of Theorem~\ref{BAiry} into abstract loop equations for the correlators $ \omega_{g,n} $. This section is analogous to Section~\ref{sec:aleA}, and we introduce the following notation. 

\noindent Following Definition~\ref{d:H1f}, we define the Heisenberg $1$-forms $ \mathcal J(z) $ on $ \tilde{C} $ as 
\begin{equation}
	\mathcal{J}(z)  := \sum_{\substack{a \in [r] \\ k \in \mathbb Z} }\mathsf{J}_{k}^a \,\dd \xi_{-k}^a(z).
\end{equation} Separating the Heisenberg zero mode from the rest, we also define
\begin{equation}
	\mathcal J^*(z) := \sum_{\substack{a \in [r] \\ k \in \mathbb Z_{\neq 0}} }\mathsf{J}_{k}^a \,\dd \xi_{-k}^a(z), \qquad \mathcal J_0(z) := \sum_{a \in [r]} \mathsf{J}^a_0 \dd \xi_0^a(z). 
\end{equation}   Using Definition~\ref{d:adjop} of the operators $ \operatorname{ad}_{g,n} $, we also introduce the $ \Omega'_{g,i,n}(z_{[i]}; w_{[n]}) $.	For any $i \in [2r]$, we define
\begin{equation}\label{eq:Omega2}
	\Omega'_{g,i,n}(z_{[i]}; w_{[n]}) = \ad_{g,n} \left(\mathcal{Z}^{-1} \prod_{k=1}^i \mathcal{J}^*(z_k) \mathcal{Z} \right).
\end{equation}
The resulting $ \Omega'_{g,i,n}(z_{[i]}; w_{[n]}) $ is a symmetric differential on $\tilde{C}^n$ in the variables $w_{1}, \ldots, w_{n}$. 

Note that it is important here that the expression in \eqref{eq:Omega2} is not normal ordered. In the setting of Section~\ref{Sec:TRA}, this did not play a role, as the Heisenberg one-forms commuted with each other. However, in the present case, if for a moment we think of $ z_{[i]} $ as a subset of $ \mathfrak f(z) $, we see that, due to the presence of the points $ \big(\begin{smallmatrix} a \\ \tilde{\zeta}  \end{smallmatrix}\big)$ and $\big(\begin{smallmatrix} a \\ -\tilde{\zeta} \end{smallmatrix}\big)$ in the set $ \mathfrak f(z) $ ~\eqref{eq:fB}, the $ \mathcal J^*(z_k) $ do not commute in general.  Thus, it is  important to define $ \Omega'_{g,i,n}(z_{[i]}; w_{[n]}) $ without normal ordering.

With this definition, Lemma~\ref{l:omeg} still holds. Indeed, the computations done in \cite[Section 4]{BKS} prove that:

\begin{lem}
	For any $i \in [2r]$,
	\begin{equation}
		\Omega'_{g,i,n}(z_{[i]}; w_{[n]}) =\sum^{\text{{\rm no} $\omega_{0,1}$}}_{\substack{\mathbf{L} \vdash [i] \\ \sqcup_{L \in \mathbf{L}} N_L = [n] \\ i + \sum_L (g_L - 1) = g}} \prod_{L \in \mathbf{L}} \omega_{g_L,|L| + |N_L|}(z_L,w_{N_L}),
	\end{equation}
In the sum we exclude all summands that contain $\omega_{0,1}$. (See Section \ref{s:notation} for notation.) \qed
\end{lem} The subtlety of the non-commuting $ \mathcal J^*(z_k) $ is captured  precisely by the $ \omega_{0,2} $ present in the above expression. This is apparent as soon as we notice, using equation~\eqref{eq:w02B}, that
\[
\omega_{0,2}\big(\begin{smallmatrix} a  & a \\ \tilde{\zeta} & -\tilde{\zeta} \end{smallmatrix}\big) = - \frac{(\dd \tilde{\zeta})^2}{4 \tilde{\zeta}^2}.
\] Such terms arise upon using the product formula for twisted modules as in \cite[Section 4.1.2]{BBCCN18} to the product of the $  \mathcal J^*(z_k) $ in \ref{eq:Omega2}.

 Finally, by modifying the output of  Lemma~\ref{l:Qi} to the present situation, we  define the  following differentials for $ i \in [2r] $, which incorporate the action of the zero modes:
\begin{equation}\label{eq:QB}
\begin{split}
	&\left[\mathcal Q^{(i)} \Omega'_{g,i,n}(\cdot\,;w_{[n]}) \right]
	(z) :=  \\
	 & \qquad \sum_{\substack{  Z \subseteq \mathfrak{f}(z) \\ |Z| = i \\ \{z\} \subseteq Z}} (-1)^{|Z_-|}\Omega'_{g,|Z|,n}(Z;w_{[n]}) \prod_{z' \in \mathfrak{f}(z) \setminus Z} \big((-1)^{\varepsilon_z(z')}\omega_{0,1}(z') - \omega_{0,1}(z)\big) .
	\end{split}
\end{equation}

Using these differentials, we can formulate the $ B $-type abstract loop equations. 

\begin{lem}\label{lem:aleB}
	The set of  differential constraints 
	\[
	\forall i \in [2r],\,\, m \in \left(-\epsilon(i) + \mathbb{Z}_{>0}\right), \qquad \mathsf{W}^i_m \mathcal Z = \hbar^{\frac{2r - 1}{2}} T\delta_{i,2r-1} \delta_{m,\frac{1}{2}} \mathcal Z
	\] on the partition function $ \mathcal Z $ imply that 
	\begin{equation}\label{eq:aleB}
		\sum_{i=1}^{2r} 	\left[\mathcal Q^{(i)}\Omega'_{g,i,n}(\cdot\,;w_{[n]})\right](z)  + \delta_{g,\frac{r}{2} - 1} \delta_{n,0}  2^{2r-1} T \omega_{0,1}(z) \frac{(\dd \tilde{\zeta})^{2r-1}}{\tilde{\zeta}^{2r}} = O\bigg(\frac{\dd \tilde{\zeta}}{\tilde{\zeta}} \bigg)^{2r} .
	\end{equation}This set of constraints is referred to as the \emph{B-type abstract loop equations}.
\end{lem}

\begin{proof}
	The proof is analogous  to the proof of Lemma~\ref{l:ALE}.  As mentioned before, the main difference here is that we are constructing the spectral curve topological recursion from a twisted module of the Heisenberg algebra.  Moreover, in the case of such twisted modules (as studied in \cite{BBCCN18} and \cite{BKS}), we should use $\tilde x = \tilde{\zeta}^2$ as the variable of the $\mathcal{W}$-algebra fields. Then, the quantum Miura transform~\eqref{Miura} yields the following differential on $ \tilde{C} $
	\begin{equation}
		\label{Wux}
		\begin{split}
			\mathsf{W}(u;\tilde x) & := \sum_{i = 1}^{2r} u^{2r - i}\,\mathsf{W}^i(\tilde x)(\dd \tilde x)^{i} \\
			& = \prod_{b = 1}^r  \big(u + \tfrac{1}{2}\mathsf{v}^{+,b}(\tilde x)\dd \tilde x + \tfrac{1}{2}\mathsf{v}^{-,b}(\tilde x)\dd \tilde x\big)\big(u + \tfrac{1}{2}\mathsf{v}^{-,b}(\tilde x)\dd \tilde x - \tfrac{1}{2} \mathsf{v}^{+,b}(\tilde x)\dd \tilde x\big)   \\
			&=  \prod_{b = 1}^r  \Big(u +  \mathcal{J}\big(\begin{smallmatrix} b \\ \tilde{\zeta} \end{smallmatrix}\big) \Big)\Big(u - \mathcal{J}\big(\begin{smallmatrix} b \\ - \tilde{\zeta} \end{smallmatrix}\big)\Big)  ,
		\end{split}	
	\end{equation} where the fields $ \mathsf{v}^{\pm,b}(z) $ were defined in \eqref{eq:vfield} (we are using the rescaled fields here). Notice that the expressions are not normally ordered. The normal ordering will introduce quantum corrections as we have discussed above, but we want to keep the expressions without normally-ordering them precisely because of how the $ 	\Omega'_{g,i,n}(z_{[i]}; w_{[n]}) $ are defined.
	
	Upon setting $ u = - \omega_{0,1}\big(\begin{smallmatrix} a \\  \tilde{\zeta} \end{smallmatrix}\big) $ in the last equation, and separating the zero modes, we get 
\begin{equation*}
\mathsf{W}( - \omega_{0,1}\big(\begin{smallmatrix} a \\  \tilde{\zeta} \end{smallmatrix}\big);\tilde x)  =  \prod_{b = 1}^r  \Big(  - \omega_{0,1}\big(\begin{smallmatrix} a \\  \tilde{\zeta} \end{smallmatrix}\big)  + \mathcal J_0\big(\begin{smallmatrix} b \\  \tilde{\zeta} \end{smallmatrix}\big) +  \mathcal{J}^*\big(\begin{smallmatrix} b \\ \tilde{\zeta} \end{smallmatrix}\big) \Big)\Big( - \omega_{0,1}\big(\begin{smallmatrix} a \\  \tilde{\zeta} \end{smallmatrix}\big) - \mathcal J_0\big(\begin{smallmatrix} b \\  -\tilde{\zeta} \end{smallmatrix}\big) - \mathcal{J}^*\big(\begin{smallmatrix} b \\ - \tilde{\zeta} \end{smallmatrix}\big)\Big) \, .
\end{equation*}
We can rewrite this in a form that is more convenient for us:
\begin{equation}
\label{eq:WQ}
\begin{split}
		\mathsf{W}( - \omega_{0,1}(z) ;\tilde{x}) & =   \prod_{  z' \in \mathfrak f(z)}  \Big( (-1)^{\varepsilon_z(z')} \mathcal J_0(z') - \omega_{0,1}(z)  +  (-1)^{\varepsilon_z(z')} \mathcal{J}^*(z')  \Big)   \,\,, \\ 
		& =   \sum_{\substack{ Z \subseteq \mathfrak{f}(z) \\ |Z| = i }} (-1)^{|Z_-|} \prod_{z' \in Z} \mathcal J^*(z') \prod_{z'' \in \mathfrak{f}(z) \setminus Z} \big((-1)^{\varepsilon_z(z'')}\mathcal J_0(z'') - \omega_{0,1}(z)\big),
\end{split}
\end{equation}
where we used that the $ \mathcal J_0 $ commute with all the $ \mathcal J^* $ in order to obtain the second equality.
	
	The constraints defining the partition function are $(\mathsf{W}_m^i - T\hbar^{\frac{r}{2} - 1}\delta_{i,2r - 1}\delta_{m,\frac{1}{2}})\mathcal{Z} = 0$ for $m \in \big( - \epsilon(i) + \mathbb{Z}_{> 0}\big)$ and $i \in [2r]$. This can be recast as
	\begin{equation}
		\label{Zws}\mathcal{Z}^{-1} \mathsf{W}^i(\tilde{x}) (\dd \tilde x)^i \mathcal{Z} - T\hbar^{\frac{r}{2} - 1}\delta_{i,2r - 1}\,\frac{(\dd \tilde{x})^{2r - 1}}{\tilde{x}^{2r - \frac{1}{2}}}  = O\bigg(\frac{\dd \tilde{x}}{\tilde{x}}\bigg)^i,
	\end{equation}
	where we note that for odd $i$, the error would be $\tilde{x}^{\frac{1}{2}}\,O\big(\frac{\dd \tilde{x}}{\tilde{x}}\big)^i$ but as $\mathsf{W}^i(\tilde{x})$ has only half-integer powers of $\tilde{x}$, this amounts to \eqref{Zws}. Since $\frac{\dd \tilde{x}}{\tilde{x}} = \frac{2\dd \tilde{\zeta}}{\tilde{\zeta}}$ and for any $a \in [r]$ we have $\omega_{0,1}\big(\begin{smallmatrix} a \\ \tilde{\zeta} \end{smallmatrix}\big) = O\big(\frac{\dd \tilde{\zeta}}{\tilde{\zeta}}\big)$,  we can combine the equations~\eqref{Zws} for all $ i $, using the definition of $ \mathsf{W}(u;\tilde{x}) $ in \eqref{Wux}, into the condition
	\begin{equation}
		\label{iufnisugn}\mathcal{Z}^{-1} \left( 	\mathsf{W}( - \omega_{0,1}(z) ;\tilde{x}) + 2^{2r-1} T \hbar^{\frac{r}{2}-1} \omega_{0,1}(z) \frac{(\dd \tilde{\zeta})^{2r-1}}{\tilde{\zeta}^{2r}} \right)  \mathcal{Z}  = O\bigg(\frac{\dd \tilde{\zeta}}{\tilde{\zeta}} \bigg)^{2r}.
	\end{equation}

	Now we need to write the left hand side of the above equation in terms of the correlators $ \omega_{g,n} $. In order to do this, we act by the operator $ \operatorname{ad}_{g,n} $, defined in \ref{d:adjop}, on the LHS. Noting that the  Heisenberg $1$-form  $ \mathcal{J}_0(z) $ acts on $ \mathcal Z $ as $ \omega_{0,1}(z) $, and applying $ \operatorname{ad}_{g,n} $ to the expression for $ \mathsf{W}(-\omega_{0,1}(z);\tilde{x}) $ obtained in equation~\ref{eq:WQ}, we recognize the expression~\eqref{eq:QB} for the $ \mathcal Q^{(i)}(\cdot\,; w_{[n]}) (z) $:
\begin{equation*}
\begin{split}
		& \quad \operatorname{ad}_{g,n} \left( \mathcal{Z}^{-1} \left(\sum_{ \substack{  Z \subseteq \mathfrak{f}(z) \\ |Z| = i  \\ \{z\}\subseteq  Z}} (-1)^{|Z_-|} \prod_{z' \in Z} \mathcal J^*(z') \prod_{z'' \in \mathfrak{f}(z) \setminus Z} \big((-1)^{\varepsilon_z(z'')}\mathcal J^0(z'') - \omega_{0,1}(z)\big) \right) \mathcal{Z} \right) \\
		& =  \left[\mathcal Q^{(i)}(\cdot\,; w_{[n]}) \right](z).
\end{split}
\end{equation*}
We restricted to the subsets $ Z $ that contain $ z $ because  the  above expression vanishes if $ z'' = z $.
	Finally, we note that 
	\[
	\operatorname{ad}_{g,n} \left( 2^{2r-1} T \omega_{0,1}(z) \hbar^{\frac{r}{2}-1} \frac{(\dd \tilde{\zeta})^{2r-1}}{\tilde{\zeta}^{2r}} \right) = \delta_{g,\frac{r}{2} - 1} \delta_{n,0}  2^{2r-1} T \omega_{0,1}(z)\frac{(\dd \tilde{\zeta})^{2r-1}}{\tilde{\zeta}^{2r}},
	\] and this proves the lemma.
\end{proof}

\subsubsection{ The $ B $-type topological recursion}

Finally, we solve the abstract loop equations obtained in Lemma~\ref{lem:aleB} using residue analysis on the spectral curve. We call the resulting formula the \textit{$ B $-type topological recursion}.
\begin{thm}
	\label{TRBcase}For $2g - 2 + n \geq 0$, the correlators associated to the partition function of Theorem~\ref{BAiry} satisfy the B-type topological recursion
	\begin{equation}
		\label{lemTRBEq}\begin{split}
			\omega_{g,n+1}(z_0,z_{[n]}) & = \sum_{o \in \mathfrak{a}} \mathop{{\rm Res}}_{z = o} \bigg(\sum_{i=2}^{2r} K(z_0;z)  \left[ \mathcal{Q}^{(i)} \Omega'_{g,i,n}(\cdot ;z_{[n]} )\right](z)\bigg) \\
			& \quad + 2^{2r - 2}T \,\delta_{g,\frac{r}{2}-1}\delta_{n,0}\, \sum_{a=1}^r \frac{\dd \xi^a_{-1}(z_0)}{\prod_{b \neq a} (Q_a^2 - Q_b^2)},
		\end{split}
	\end{equation}
	with the $  \left[ \mathcal{Q}^{(i)} \Omega'_{g,i,n}(\cdot ;z_{[n]}) \right](z)$ defined in \eqref{eq:QB}, and the recursion kernel given by:
	\begin{equation}\label{eq:recB}
		\begin{split}
			K(z_0;z) & := -\frac{\int_o^{z} \omega_{0,2}(z_0, \cdot)}{\prod_{z' \in \mathfrak{f}'(z)} \big((-1)^{\varepsilon_z(z')}\omega_{0,1}(z') - \omega_{0,1}(z)\big)},
		\end{split}
		\end{equation}
		where $o$ is the point in $\mathfrak{a}$ in the same component of $\tilde{C}$ as $z$. 
\end{thm}
\begin{rem}
	Although the nature of the spectral curve is different, the only structural differences between the $B_r$-type topological recursion formula~\eqref{lemTRBEq} and the $\mathfrak{gl}_{2r}$ topological recursion formula~\eqref{TRA} is the presence of the signs $(-1)^{|Z_-|}$ and $(-1)^{\epsilon_z(z')}$ (\emph{i.e.}, evaluating a correlator at a point in $\mathfrak{f}_-(z)$ comes with a minus sign) and a summation over $r$ (instead of $2r$) ramification points. 
\end{rem}

\begin{proof}
	We apply the steps in the proof of Theorem~\ref{lemTRA} to the $ B $-type abstract loop equations of Lemma~\ref{lem:aleB}. We notice that
\[
\left[\mathcal Q^{(1)} \Omega'_{g,1,n}(\cdot\,;w_{[n]}) \right]
	(z)  = \Omega'_{g,1,n}(z; w_{[n]})  \prod_{z' \in \mathfrak{f}(z) \setminus Z} \big((-1)^{\varepsilon_z(z')}\omega_{0,1}(z') - \omega_{0,1}(z)\big) .
	\]
	Thus, we can invert it with the recursion kernel \eqref{eq:recB}, which gives
	\[
	 \Omega'_{g,1,n}(z; w_{[n]}) = - \sum_{o \in \mathfrak{a}} \mathop{{\rm Res}}_{z = o} K(z_0;z) \left[\mathcal Q^{(1)} \Omega'_{g,1,n}(\cdot\,;w_{[n]}) \right]
	(z).
	\]
	We substitute back in the abstract loop equations \eqref{eq:aleB}. Noting that the kernel has a zero of degree $2r$ at $o \in \mathfrak{a}$, the residue of $K(z_0,z)$ times the right-hand-side of \eqref{eq:aleB} at $o \in \mathfrak{a}$ vanishes. We thus obtain:
\[
\begin{split}
 \Omega'_{g,1,n}&(z_0; w_{[n]})  =\\
&  \sum_{o \in \mathfrak{a}} \Res_{z = o} K(z_0,z) \left( \sum_{i=2}^r \left[ \mathcal{Q}^{(i)} \Omega'_{g,i,n}(\cdot\,; w_{[n]}) \right](z)  +  \delta_{g,\frac{r}{2}-1} \delta_{n,0} 2^{2r-1} T \omega_{0,1}(z) \frac{(\dd \tilde \zeta)^{2r-1}}{ \tilde \zeta^{2r}} \right).
 \end{split}
\]
We evaluate the term proportional to $T$:
\[
\begin{split}
2^{2r-1} \sum_{o \in \mathfrak{a}}& \Res_{z = o} K(z_0,z)  T \omega_{0,1}(z) \frac{(\dd \tilde \zeta)^{2r-1}}{ \tilde \zeta^{2r}}  \\
&=- 2^{2r-1}\sum_{a=1}^r \Res_{\tilde \zeta = 0} \frac{\tilde \zeta \dd \xi^a_{1}(z_0)}{
\tilde \zeta_0 (\tilde \zeta_0 - \tilde \zeta)} \left(\frac{T Q_a}{- 2 Q_a \prod_{b \neq a}(Q_a^2 - Q_b^2)}\right) \frac{\tilde \zeta^{2r-1}( \dd \tilde \zeta)^{2r}}{\tilde \zeta^{2r+1} (\dd \tilde \zeta)^{2r-1}} \\
 & =2^{2r-2} T \sum_{a=1}^r \frac{ \dd \xi^a_{-1}(z_0).}{\prod_{b \neq a}(Q_a^2 - Q_b^2)},
\end{split}
\]
where we used the fact that, for $z = \big(\begin{smallmatrix} a \\ \tilde \zeta \end{smallmatrix}\big)$:
\[
 \prod_{z' \in \mathfrak{f}(z) \setminus Z} \big((-1)^{\varepsilon_z(z')}\omega_{0,1}(z') - \omega_{0,1}(z) \big)=  - 2 Q_a  \frac{(\dd \tilde \zeta)^{2r-1}}{\tilde \zeta^{2r-1}} \prod_{b \neq a} (Q_a^2- Q_b^2).
\]
This concludes the proof.

\end{proof}

\begin{rem}
	The discussion of Section~\ref{Generald} also applies here: we can consider more general spectral curves of type B, that is spectral curves in which ramification points come in pairs. Equation~\eqref{lemTRBEq} defines a topological recursion of type $B_r$, for general $\omega_{0,1}$, $\omega_{\frac{1}{2},1}$ and $\omega_{0,2}$ as in \eqref{eq:w01shift}, \eqref{eq:w12shift} and \eqref{eq:genw02}, except that  well-definedness now requires $Q_a^2 \neq Q_b^2$ for any $a \neq b$, with $ Q_a  \in \mathbb C^*$.
\end{rem}

\begin{rem}
	In this section we assumed $r \geq 2$ merely because $B_2 = \mathfrak{so}_3 = \mathfrak{sl}_2 = A_1$, but the theory applies as well for $r = 1$ provided that $Q_1 \neq 0$, $T = 0$. In this case the Airy structure is not shifted, and hence its partition function should correspond to a highest-weight vector. Furthermore, $\omega_{0,1}$ is not odd, and \eqref{lemTRBEq} gives symmetric $\omega_{g,n}$. Since the spectral curve has only one component, we write the points on $\tilde C$ directly as $\tilde \zeta$. The recursion in that case reads
	\[
	\begin{split}
	\omega_{g,n + 1}&(\tilde \zeta_0, \tilde \zeta_{[n]}) \\
	&= \mathop{{\rm Res}}_{\tilde \zeta= 0} \frac{\int_0^{\tilde \zeta} \omega_{0,2}(\cdot,\tilde \zeta_0)}{2 \omega_{0,1}(\tilde \zeta)}\bigg(\omega_{g - 1,n + 2}(\tilde \zeta,-\tilde \zeta) + \sum_{\substack{h + h' = g \\ J \sqcup J' = \{\tilde \zeta_1,\ldots,\tilde \zeta_n\}}}^{{\rm no}\,(0,1)} \omega_{h,1 + |J|}(\tilde \zeta,J)\omega_{h',1+|J'|}(-\tilde \zeta,J')\bigg).
	\end{split}
	\]
\end{rem}

\section{Gaiotto vectors of type C and D}

\label{Sec:AiryCD}

We now construct Gaiotto vectors for type C and D using Airy structures. In this section we focus on the self-dual level $\kappa=1$.

If $G$ is of type $D$, then it is simply-laced, and the connection between Gaiotto vectors and Whittaker vectors in Theorem   \ref{ThmBFN} holds. Thus the Gaiotto vector of type D is a Whittaker vector for $\mathcal{W}^{\mathsf{k}}(\mathfrak{so}_{2n})$. For type C, Theorem \ref{ThmBFN} does not hold, but from physics \cite{KMST} it is expected that the Gaiotto vector should be a Whittaker vector for an appropriate twisted module of $\mathcal{W}^{\mathsf{k}}(\mathfrak{so}_{2r})$, as in Section \ref{Sec:AiryB}. In this section we show that these Whittaker vectors can be constructed using Airy structures.

We do not study here the spectral curve topological recursion for type C and D, as it will be more complicated given the form of the $\mathcal{W}$-generators in type D. We also do not discuss general level $\kappa$ as the $\mathcal{W}$-generators are not completely explicitly known when $\kappa \neq 1$.

\subsection{Type \texorpdfstring{$D$}{D}}

Let $r \geq 2$. We consider $\mathcal{W}^{\mathsf{k}}(\mathfrak{so}_{2r}) \subset \mathcal{H}(\mathfrak{so}_{2r})$ at self-dual level.
The roots of $  \mathfrak{so}_{2r} $ can be described as $\pm \chi^i\pm \chi^j$ where $(\chi^i)_{i = 1}^{r}$ is an orthonormal basis for the Cartan subalgebra $\mathfrak{h} = \mathbb{C}^r$. The following vectors in $\mathcal{H}(\mathfrak{so}_{2r})$ --- viewed as a subalgebra of the lattice VOA, see \cite[Section 3.2.2]{BBCCN18} --- strongly generate the $\mathcal{W}^{\mathsf{k}}(\mathfrak{so}_{2r}) $-algebra at self-dual level:
\begin{equation}
\label{eq:dbasis}  
\begin{split}
	\nu^d &= \bigg(\sum_{i=1}^r {\rm e}^{\chi^i}_{-d}{\rm e}^{-\chi^i}_{-1} + {\rm e}^{-\chi^i}_{-d}{\rm e}^{\chi^i}_{-1}\bigg) \ket{0} \qquad d \in \{2,4,6,\ldots, 2r-2\}\,,\\ 
	\widetilde{\nu}^r &= \chi^{1}_{-1} \chi^{2}_{-1}\,\cdots\,\chi^{r}_{-1} \ket{0}\,.
\end{split}
\end{equation}
The conformal weight of these vectors are $ 2,4,\ldots,2r-2$ and $ r $, which are indeed the Dynkin exponents of $ \mathfrak{so}_{2r} $.  The dual Coxeter number is $h^\vee = 2r - 2$.

We define the modes of the generators associated to these elements as 
\begin{equation}\label{eq:modesD}
	Y(\nu^d,z) := \sum_{m \in \mathbb Z} W^d_m z^{-m-d}, \qquad Y(\tilde{\nu}^r,z) := \sum_{m \in \mathbb Z}\widetilde{W}^r_m z^{-m-r},
\end{equation}
with the corresponding rescaled modes $\mathsf{W}^d_m$ and $\widetilde{\mathsf{W}}^r_m$,
and construct a representation of $ \mathcal{W}^{\mathsf{k}}(\mathfrak{so}_{2r}) $ by restricting the following representation of the Heisenberg algebra:
\[
\chi^a(z) = \sum_{ k \in \mathbb Z} J^a_m z^{-m-1}, \qquad a \in [r].
\] We use our usual differential representation for the rescaled Heisenberg modes $ \mathsf{J}^a_m $ (recall that $\alpha_0 = 0$ since $\kappa=1$):
\begin{equation} \label{eq:hrepD}
	\mathsf{J}^a_{m} = \left\{
	\begin{array}{lcr}
		\hbar \frac{\partial}{\partial x^a_m}& &  m > 0,\\[0.3ex]
		Q_a & & m = 0, \\[0.3ex]
		-m x^{a}_{-m}& &  m < 0. \end{array}\right.
\end{equation}

By now our strategy should be familiar. We want to construct a Whittaker vector $\ket{w}$ annihilated by all rescaled positive modes $\mathsf{W}^i_m$ except for the one-mode $\mathsf{W}^{2r-2}_1$ that acts as a constant on $\ket{w}$. (We recall the definition of the rescaled modes in $\mathcal{A}^{\hbar}$ from Lemma \ref{l:Ahbar}.) We proceed as usual: we first show that $S' = (\mathsf{W}^{2r-2}_1) \subset S_+$ is extraneous, and then we construct an Airy structure differential representation for the modes in $S_+$. The partition function of the shifted Airy structure corresponds to the Whittaker vector $\ket{w}$.

As we are back in the realm of Proposition  \ref{thecoth}, we already know that $S' \subset S_+$ is extraneous. So what remains is to construct the Airy structure differential representation.

\begin{thm}\label{prop:typeD}
	Let $T \in \mathbb{C}$, $Q_1, \ldots, Q_r \in \mathbb{C^*}$, and suppose that $Q_a \neq \pm Q_b$ for any distinct $a,b \in [r]$.  Let $\mathsf{W}^i_m$ and $\widetilde{\mathsf{W}}^r_m$  be the differential operators associated to the rescaled modes of the strong generators in \eqref{eq:modesD} obtained via \eqref{eq:hrepD}. Then the family of differential operators 
	$$
	\begin{array}{ll} \overline{\mathsf{W}}^d_m:= \mathsf{W}^d_m -\hbar^{r - 1}T\delta_{d,2r - 2} \delta_{m, 1} & \quad m \in \mathbb{Z}_{>0},\quad d \in \{2,4, \ldots, 2r-2\} \,,\\[0.5ex]
	\widetilde{\mathsf{W}}^r_m & \quad m \in \mathbb{Z}_{>0} ,  \end{array}
	$$
	forms an Airy structure.
\end{thm}

As usual, the direct consequence of this Proposition is that the partition function $\mathcal{Z}$ uniquely associated to this Airy structure can be identified with the searched-for Whittaker vector $\ket{w}$,  which in turn corresponds to the Gaiotto vector $\ket{\mathfrak{G}}$ of type D.

\begin{proof}
	By Proposition \ref{thecoth} we know that  $S' = \{\mathsf{W}^{2r-2}_1\} \subset S_+$ is extraneous, which also implies that the set of positive modes in $S_+$ satisfy the subalgebra condition. We need to show that the operators also satisfy the degree condition.
	
	As we only allow positive modes, the above  operators do not have terms of degree zero. Let us now analyze the degree one terms. The field $\widetilde{\mathsf{W}}^r(z)$ coincides with the top degree field in $\mathcal{H}(\mathfrak{gl}_r)$, so we know from \eqref{pi1critlev} that its degree one projection reads:
	$$
	\pi_1(\widetilde{\mathsf{W}}_m^r) = \sum_{a = 1}^r \bigg(\prod_{\substack{b \in [r] \\ b \neq a}} Q_{b}\bigg) \mathsf{J}_m^{a} = \bigg(\prod_{b = 1}^r Q_b\bigg) \sum_{a = 1}^r  (Q_a^{-1}\mathsf{J}_m^a).
	$$
	We also know (see for example, \cite[Lemma 3.7]{BakalovMilanov2}) that
	$$
	\pi_1(W^d(z)) = \sum_{a = 1}^{r } (\chi^a(z))^{d}.
	$$
	Hence the term of degree $ 1 $ for its rescaled positive modes is
	$$
	\pi_1(\mathsf{W}^d_m)  = \sum_{a=1}^N d(Q_a)^{d-1} \mathsf{J}_m^a,\qquad m > 0.
	$$
	This can be put in the form
	$$
	\left(\begin{array}{ccc} 1 & \cdots & 1 \\[0.5ex] Q_1^2 & \cdots & Q_r^2 \\[0.5ex]  \vdots & & \vdots \\[0.5ex]  Q_1^{2(r - 1)} & \cdots & Q_r^{2(r - 1)} \end{array}\right) \left(\begin{array}{c} Q_1^{-1} \mathsf{J}_m^1 \\[0.5ex] Q_2^{-1}\mathsf{J}_m^2 \\[0.5ex]  \vdots \\[0.5ex]  Q_r^{-1} \mathsf{J}_m^r \end{array}\right) = \left(\begin{array}{c} (Q_1\,\cdots\, Q_r)^{-1} \pi_1(\widetilde{\mathsf{W}}_m^r) \\[0.5ex]  \frac{1}{2}\pi_1(\mathsf{W}_m^2) \\[0.5ex]  \vdots \\[0.5ex]  \frac{1}{2(r - 1)} \pi_1(\mathsf{W}_m^{2r - 2}) \end{array}\right).
	$$
	The square matrix on the left-hand side is a Vandermonde matrix. It is invertible if and only if $Q_a^2 \neq Q_b^2 $ for any distinct $a,b \in [r]$, and then its inverse $A$ is given as follows,
	$$
	A_{a,i} = \frac{(-1)^{r - i} e_{r - i}\big((Q_b^2)_{b \neq a}\big)}{\prod_{b \neq a} (Q_a^2 - Q_b^2)},\qquad a,i \in [r].
	$$	
	where $e_j$ is the $j$-th elementary symmetric function. In particular, we have
	$$
	A_{a,1} = (-1)^{r - 1} \prod_{b \neq a} \frac{Q_b^2}{Q_a^2 - Q_b^2}.
	$$
	Thus, we conclude that the operators
	$$
	\mathsf{H}^a_m = \frac{Q_a}{\prod_{b \neq a} (Q_a^2 - Q_b^2)}\bigg(\frac{\widetilde{\mathsf{W}}_m^r}{Q_1\,\cdots\,Q_r} + \sum_{d \in \{2,4,\ldots,2r - 2\}} (-1)^{r - 1 - d/2}\,e_{r - 1 - d/2}\big((Q^2_b)_{b \neq a}\big)\,\frac{\mathsf{W}_m^d}{d}\bigg)
	$$
	define an Airy structure in normal form, with $\pi_1(\mathsf{H}^a_m) = \mathsf{J}_{m}^a$. 
\end{proof}

\subsection{Type C}

The situation for type C is similar to that of type B. Theorem \ref{ThmBFN} does not apply, as $G$ is not simply connected. However, 
the Gaiotto vector corresponding to $ G = \operatorname{Sp}(2r) $ (type C) can be constructed as a state in a certain  twisted module of $ \mathcal W^{\mathsf k}(\mathfrak{so}_{2r+2}) $ (type D). The twist corresponds to the folding symmetry of the Dynking diagram used to obtain type $ C $ from type $ D $. Thus we can construct the Gaiotto vector of type C using Airy structures.

Recall that the  Cartan subalgebra of $ \mathfrak{so}_{2r+2} $ has an orthonormal basis $ \left(\chi^i\right)_{i=1}^{r+1} $. The folding symmetry  group $ \mathbb Z/2\mathbb Z $ is generated by the automorphism $ \sigma $ which acts as
\[
\chi^{r+1} \leftrightarrow -\chi^{r+1}, \qquad 	\chi^{i} \leftrightarrow \chi^{i} \text{ for }  i \in [r] .
\]  As the above action is already diagonal, by following the usual algorithm, we construct a $ \sigma $-twisted $ \mathcal H(\mathfrak{so}_{2r+2}) $-module as 
\[
Y_\sigma(\chi^a_{-1}\ket{0}, z) = \sum_{m\in  \mathbb Z} J^a_{m} z^{-m-1}, \text{ for } a \in [r], \,\, \text{ and } \,\,
Y_\sigma( \chi^{r+1}_{-1}\ket{0}, z) = \sum_{k \in \mathbb Z +\frac{1}{2}} J^{r+1}_{k} z^{-k-1}.
\]  Note that the twist field associated to $  \chi^{r+1}_{-1}\ket{0} $ does not have a zero-mode. Here we choose a representation of the Heisenberg algebra as follows. For $ a $ in $ [r] $ and $ m \in \mathbb Z $, we define the rescaled Heisenberg modes:
\begin{equation}\label{eq:heiC1}
	\mathsf{J}^a_{m} = \left\{
	\begin{array}{lr}
		\hbar \frac{\partial}{\partial x^a_m}&  m > 0,\\
		-m x^{a}_{-m}& m < 0, \\
		Q_a  & m = 0, 
	\end{array} \right.
\end{equation} and for $ r+1 $, we define

\begin{equation}\label{eq:heiC2}
	\mathsf{J}^{r+1}_{k} = \left\{
	\begin{array}{lr}
		\hbar \frac{\partial}{\partial x^{r+1}_{k}}&  k > 0,\\
		-k x^{r+1}_{-k}& k < 0,
	\end{array} \right.
\end{equation} where $ k \in \frac{1}{2} + \mathbb Z $.

\noindent 	The action of $ \sigma $ on the generators~\eqref{eq:dbasis} of  $ \mathcal W_{\mathsf k}(\mathfrak{so}_{2r+2}) $ is diagonal:
\[
\qquad	\nu^d \to \nu^d,\qquad \tilde{\nu}^{r+1} \to - \tilde{\nu}^{r+1},
\] where $ d  = 2,4,\ldots, 2r $. Hence we get the following mode expansion
\begin{equation}\label{eq:modesC}
	Y_\sigma(\nu^d,z) := \sum_{k \in \mathbb Z} W^d_k z^{-k-d}, \qquad Y_\sigma(\tilde{\nu}^{r+1},z) := \sum_{k \in \frac{1}{2}+ \mathbb Z}\widetilde{W}^{r+1}_k z^{-k-r-1},
\end{equation}
with the corresponding rescaled modes $\mathsf{W}^d_k$ and $\widetilde{\mathsf{W}}^{r+1}_m$ (those are the modes in $\mathcal{A}^{\hbar}$ as in Lemma \ref{l:Ahbar}).

We are looking for a Whittaker vector $\ket{w}$ that is annihilated by all positive modes $\mathsf{W}^d_k$ and $\widetilde{\mathsf{W}}^{r+1}_k$ from \eqref{eq:modesC}, except for $\widetilde{\mathsf{W}}^{r+1}_{1/2}$ which acts on $\ket{w}$ as constant. We proceed as usual: we need to show that $S' = \{\widetilde{\mathsf{W}}^{r+1}_{1/2}\} \in S_+$ is extraneous, and then construct an Airy structure differential representation for the positive modes in $S_+$. 

However, as in the type B case, Proposition \ref{thecoth} does not apply here, and hence we do not have a proof that $S' \subset S_+$ is extraneous. We formulate this as a conjecture:

\begin{conj}
	\label{theconjC} Let $S_+$ be the subset of positive modes in \eqref{eq:modesC}. Then the subset $S' = \{\widetilde{\mathsf{W}}^{r+1}_{1/2}\} \in S_+$ is extraneous, according to Definition \ref{d:extra}.
\end{conj}

\begin{thm}
	\label{Prop73C} Assume that Conjecture \ref{theconjC} holds. Let $T \in \mathbb{C}$, $Q_1, \ldots, Q_r \in \mathbb{C}^*$, and suppose that $Q_a \neq \pm Q_b$ for any distinct $a,b \in [r]$. Let $\mathsf{W}^i_m$ and $\widetilde{\mathsf{W}}^{r+1}_m$ be the differential operators associated to the rescaled modes in \eqref{eq:modesC} via \eqref{eq:heiC1} and \eqref{eq:heiC2}. Then the family of differential operators 
	$$
	\begin{array}{ll} \mathsf{W}^d_m  & \quad m \in \mathbb{Z}_{>0},\quad d \in \{2,4, \ldots, 2r\} \,,\\[0.5ex]
	\overline{\mathsf{W}}^{r+1}_m := \widetilde{\mathsf{W}}^{r+1}_m  -\hbar^{\frac{r+1}{2}}T  \delta_{m, \frac{1}{2}} & \quad m \in \left(\mathbb{Z}_{>0}-\frac{1}{2} \right) \end{array}
	$$
	forms an Airy structure.
\end{thm}

As usual, the immediate consequence is that the partition function $\mathcal{Z}$ uniquely associated to this Airy structure can be identified with the Gaiotto vector $\ket{\mathfrak{G}}$ of type C.

\begin{proof}
	By Proposition 3.14 of \cite{BBCCN18} (see also footnote \ref{f:subalgebra}), we know that the subset of positive modes $S_+$ satisfies the subalgebra condition in Definition \ref{d:Airy}. Assuming that Conjecture~\ref{theconjC} holds,  the subset $S' = \{\widetilde{\mathsf{W}}^{r+1}_{1/2}\} \in S_+$ is extraneous, and so the shifted differential operators also satisfy the subalgebra condition in Definition \ref{d:Airy}. What remains is to show that they satisfy the degree condition.

	As in the case of type $ B $, the correction terms in the fields $ Y_\sigma(\nu^d,z) $ and $ Y_\sigma(\tilde{\nu}^{r+1},z) $ due to the twisting come with powers of $ \hbar $,  and thus contribute terms of degree $\geq 2$ to the differential operators. Therefore we can safely ignore them in the proof.	
	
	As usual, it is clear that there are no terms of degree zero. Let us analyze the degree one terms. This proof is analogous to the proof of  Theorem~\ref{prop:typeD}, \emph{i.e.} the case of type $ D $, and hence we will be brief. As we have already noted, $ \nu^d $ is invariant under the automorphism, hence we have 
	\[
	\pi_1(\mathsf{W}^d_m) = \sum_{a=1}^{r} d (Q_a)^{d-1} \mathsf{J}^a_m.
	\] Note that $ \mathsf{J}^{r+1}_m $ does not appear on the right hand side of the above equation. This can be inverted as in the type B case:
	$$
	\mathsf{J}_a^m = Q_a^{-1} \sum_{i = 1}^r \frac{(-1)^{r - i}\,e_{r - i}\big((Q_b^2)_{b \neq a}\big)}{\prod_{b \neq a} (Q_a^2 - Q_b^2)}\,\frac{\pi_1(\mathsf{W}_m^{2i})}{2i}. 
	$$
	From the definition of the generator $ \tilde{\nu}^{r+1} $, we see that
	\[
	\pi_1(\widetilde{\mathsf{W}}_m^{r+1}) =  \bigg(\prod_{a = 1}^r Q_{a}\bigg) \mathsf{J}_m^{r+1}.
	\]
	Therefore, the family of operators
	\begin{equation}
		\begin{split}
			\mathsf{H}^a_m &= \sum_{d \in \{2,4,\ldots,2r\}} \frac{(-1)^{r  - \frac{d}{2}}}{d}\, \frac{e_{r  - \frac{d}{2}}\big((Q^2_b)_{b \neq a}\big)}{Q_a \prod_{b \neq a} (Q_a^2 - Q_b^2)}\,\mathsf{W}_m^d, \qquad  a \in [r], \\
			\mathsf{H}^{r+1}_m &= \frac{\widetilde{\mathsf{W}}^{r+1}_m}{\prod_{a = 1}^r Q_a},
		\end{split}
	\end{equation}
	define an Airy structure in normal form with $\pi_1(\mathsf{H}^a_m) = \mathsf{J}_{m}^a$. (The $ e_i $ in the above equations denotes the $ i $-th elementary symmetric function.) It then follows that the shifted differential operators also form an Airy structure, assuming that Conjecture~\ref{theconjC} holds.
	
\end{proof}

%
%\bibliographystyle{plain}
%\bibliography{BibliWhit}

\end{document}